\documentclass[aps,prx,floatfix,nofootinbib,onecolumn,superscriptaddress,longbibliography]{revtex4-2}
\pdfoutput=1
\usepackage{dsfont}

\usepackage{braket}
\usepackage{siunitx}

\usepackage{amsthm}
\usepackage{amsmath}
\usepackage{amsfonts}
\usepackage{braket}
\usepackage{amssymb}
\usepackage{graphicx}
\usepackage{color}
\usepackage{accents}
\usepackage{bbold}
\usepackage[colorlinks=true, citecolor=blue]{hyperref}
\usepackage[normalem]{ulem}
\usepackage[shortlabels]{enumitem}
\usepackage{float}
\usepackage{subcaption}
\usepackage{siunitx}
\usepackage[version=3]{mhchem}
\usepackage{booktabs}
\usepackage{todonotes}
\usepackage{multirow}
\usepackage{standalone}
\usepackage{chngcntr}

\usepackage{xcolor}
\definecolor{COVgreen}{cmyk}{0.75,0,0.6,0}

\newcommand{\detcc}{\ensuremath{\left|\Delta E_\mathrm{TCCSD}^\mathrm{noise}\right|}}
\newcommand{\decc}{\ensuremath{\left|\Delta E_\mathrm{ec-CC}^\mathrm{noise}\right|}}

\newcommand{\ccwfn}{\ensuremath{\ket{\Psi_\mathrm{CC}}}}
\newcommand{\hfdet}{\ensuremath{\ket{\Phi_0}}}
\newcommand{\trial}{\ensuremath{\ket{\Psi^T}}}

\newcommand{\cre}[1]{\ensuremath{\hat{a}^\dagger_{#1}}}
\newcommand{\ani}[1]{\ensuremath{\hat{a}_{#1}}}
\newcommand{\eri}[2]{\ensuremath{\langle #1 || #2 \rangle}}
\newcommand{\simH}{\ensuremath{\bar{H}}}

\makeatletter %
\newsavebox{\@brx}
\newcommand{\llangle}[1][]{\savebox{\@brx}{\(\m@th{#1\langle}\)}%
  \mathopen{\copy\@brx\kern-0.5\wd\@brx\usebox{\@brx}}}
\newcommand{\rrangle}[1][]{\savebox{\@brx}{\(\m@th{#1\rangle}\)}%
  \mathclose{\copy\@brx\kern-0.5\wd\@brx\usebox{\@brx}}}
\makeatother

\newtheorem{theorem}{Theorem}[section]

\theoremstyle{definition}

\theoremstyle{remark}

\setcitestyle{super}

\usepackage{listings} 

\usepackage{xcolor}

\definecolor{codegreen}{rgb}{0,0.6,0}
\definecolor{codegray}{rgb}{0.5,0.5,0.5}
\definecolor{codepurple}{rgb}{0.58,0,0.82}
\definecolor{backcolour}{rgb}{0.95,0.95,0.92}
\lstdefinestyle{mystyle}{
  backgroundcolor=\color{backcolour},   commentstyle=\color{codegreen},
  keywordstyle=\color{magenta},
  numberstyle=\tiny\color{codegray},
  stringstyle=\color{codepurple},
  basicstyle=\ttfamily\scriptsize,
  breakatwhitespace=false,         
  breaklines=true,                 
  captionpos=b,                    
  keepspaces=true,                 
  showspaces=false,                
  showstringspaces=false,
  showtabs=false,
  xleftmargin=0.02\textwidth,
  rulecolor=\color[RGB]{200,200,200},
  frame=bt,
  framextopmargin=2pt,
  framexbottommargin=2pt,
  framexleftmargin=10pt,
  tabsize=2
}
\long\def\/*#1*/{}

\newcommand{\google}{\affiliation{%
Google Quantum AI, Venice, CA, United States}}

\newcommand{\covestro}{\affiliation{%
Covestro Deutschland AG, 51373 Leverkusen, Germany}}

\begin{document}
\title{Tailored and Externally Corrected Coupled Cluster with Quantum Inputs}
\date{\today}

\author{Maximilian Scheurer}
\email[Corresponding Author: ]{maximilian.scheurer@covestro.com}
\covestro

\author{Gian-Luca R.~Anselmetti}
\thanks{Current Address: Quantum Lab, Boehringer Ingelheim, 55218 Ingelheim am Rhein, Germany}
\covestro

\author{Oumarou Oumarou}
\covestro

\author{Christian Gogolin}
\email{christian.gogolin@covestro.com}
\covestro

\author{Nicholas C.~Rubin}
\email[Corresponding Author: ]{nickrubin@google.com}
\google

\begin{abstract}
We propose to use wavefunction overlaps obtained from a quantum computer as inputs for the classical split-amplitude techniques, tailored and externally corrected coupled cluster, to achieve balanced treatment of static and dynamic correlation effects in molecular electronic structure simulations.
By combining insights from statistical properties of matchgate shadows, which are used to measure quantum trial state overlaps, with classical
correlation diagnostics, we are able to provide quantum resource estimates well into the classically no longer exactly solvable regime.
We find that rather imperfect wavefunctions and remarkably low shot counts are sufficient to cure qualitative failures of plain coupled cluster singles doubles
and to obtain chemically precise dynamic correlation energy corrections. 
We provide insights into which wavefunction preparation schemes have a chance of yielding quantum advantage, and we test our proposed method using overlaps measured on Google's Sycamore device.
\end{abstract}

\maketitle

\section{Introduction}
Recent work devising how to use a quantum computer to model electronic structure continues to suggest fermionic simulation as a valuable and viable application of near-term and fault-tolerant quantum devices. Just as in the classical modeling of electronic structure where various strategies are explored for treating different aspects of electron correlation, quantum algorithms have largely followed a similar divide-and-conquer trajectory. While treatment of strong correlation is cited as the most likely motivation for applying quantum computers to chemistry,\cite{mcardle2020quantum} an accurate treatment of dynamic electronic correlation effects is an important aspect when considering total quantum resource counts and the viability of applying a quantum computing protocol to solve electronic structure problems.

In this work, we further integrate classical electronic structure methodologies and quantum computation to lower measurement requirements needed for dynamical correlation corrections. Our focus is on split-amplitude coupled cluster (CC) methodologies that have been widely studied in the quantum chemistry community.\cite{kinoshita2005coupled,paldus2017externally}
Specifically, we analyze the robustness of quantum inputs to these theories along with motivating wavefunction characteristics that a quantum computer must satisfy to improve over classical approximate external amplitude inputs. While it is generally possible to layer quantum state preparation with many dynamic correlation corrections, split-amplitude methods provide the necessary framework in which we can analyze the sampling costs and robustness beyond common, and usually loose, tomographic sampling bounds.

To date, there exist several approaches to include dynamic correlation effects in quantum simulations of chemistry, e.g.,
virtual quantum subspace expansion (VQSE),\cite{takeshita2020increasing} second-order perturbation theory using the variational quantum eigensolver (VQE)\cite{peruzzo2014variational,mcclean2016theory,mcardle2020quantum}
together with quantum subspace expansion (QSE), named NEVPT2(VQE,QSE),\cite{tammaro2023n} non-orthogonal configuration interaction approaches,\cite{baek2023say} and NEVPT2 based on qubit reduced density matrices (RDMs), named QRDM-NEVPT2.\cite{krompiec2022strongly}
We note that a plethora of other methods for studies of chemical systems exists, e.g., based on embedding techniques.\cite{kawashima2021optimizing,rossmannek2021quantum,li2022toward,liu2023bootstrap,he2020zeroth,he2022second,evangelista2018perspective,PRXQuantum.4.020313}
The QRDM-NEVPT2 method, in addition to VQE state preparation and optimization, requires quantum evaluation of the three-particle reduced density matrix (3RDM)
and four-particle RDM-like terms, exploiting a cumulant approximation. For quantum active space (AS) methods, it is already quite resource-intense with respect to
required number of measurement repetitions, so-called shots, to accurately determine the 2RDM for energy evaluation. As a consequence, a multitude of approaches
for efficient measurements of the 2RDM has evolved so far, reducing the number of distinct measurement bases and variances significantly.\cite{huggins2021efficient,cohn2021cdf,choi2023fluid,oumarou2022accelerating, zhao_fermionic_2021, low_classical_2022,PhysRevX.10.031064, wan2022matchgate}
Thus, the burden of even more measurements required for higher-order reduced density matrices in, e.g., QRDM-NEVPT2 poses the question whether one can construct
a hybrid quantum/classical approach which does not rely on RDMs.  Recent work leveraging classical shadows~\cite{huang2020predicting} combined with auxiliary-field Quantum Monte Carlo (AFQMC)~\cite{huggins_unbiasing_2022} suggests an efficient path forward by shifting focus to wavefunction corrections in lieu of the RDMs.
\begin{figure}[ht]
    \centering
    \includegraphics[width=0.4\textwidth]{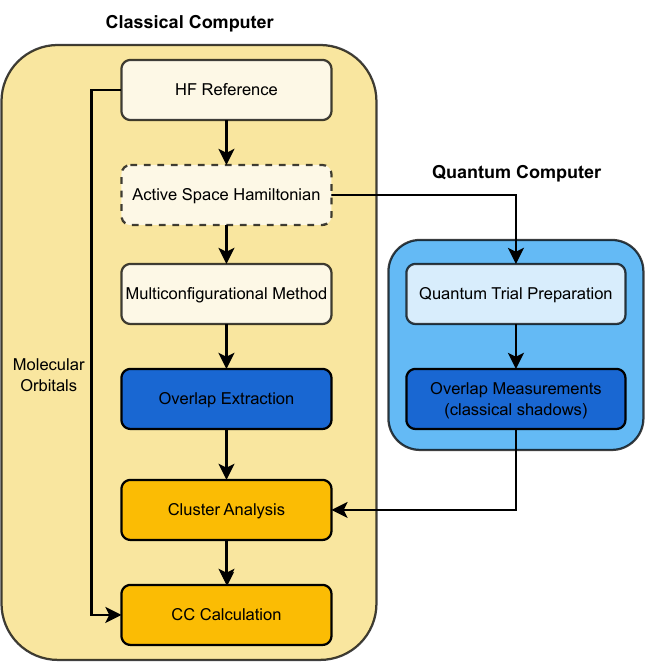}
    \caption{\label{fig:method_flow} A visual representation of the quantum split-amplitude protocol analyzed in this work.}
\end{figure}

To this end, we propose to use quantum trial states as input for tailored coupled cluster (TCC) and externally corrected coupled cluster (ec-CC) methods.
On the quantum device level, this requires preparing a suitable state and then measuring the overlaps of that state with
(excited) Slater determinants, corresponding to computational basis state overlaps.
The steps required to obtain quantum inputs and run the subsequent classical CC calculation is depicted in Figure~\ref{fig:method_flow},
comparing the workflow with its purely classical counterpart. Only two building blocks differ between the quantum and classical version.
To extract computational basis overlaps we propose to leverage various forms of classical shadows,\cite{huang2020predicting} a measurement technique allowing for efficient estimation of multiple expectation values, which have recently been extended to use the matchgate group for efficient overlap estimation.~\cite{wan2022matchgate, low_classical_2022}
These protocols fulfill the requirement of providing overlap \emph{magnitudes} and \emph{signs}, which are needed for subsequent cluster analysis.
In contrast to classical TCC and ec-CC methodologies developed to date, quantum TCC and quantum ec-CC require analysis of
i) when improvements over classical approximate methods are possible and ii) sensitivity of these methods to device and shot noise. We take steps towards addressing both questions in this work.
We note that the proposed methodologies are not limited to any particular quantum computing setting, since the CC methods are completely agnostic of the
origin of the cluster amplitudes. Hence, they are applicable to near-term and fault-tolerant setups alike, provided a protocol for evaluating the overlaps exists.
While the total number of computational basis state overlaps scales exponentially in the number of qubits $n$, TCC with singles and doubles (TCCSD) and ec-CC only require a polynomial number
of overlaps as input for cluster analysis, i.e., less than $n^4$ for TCCSD and less than $n^8$ for ec-CC.

A substantial component of this work involves constructing a noise model for matchgates shadows and mapping a model for measurement overheads given common electronic structure metrics that are cheaply available. We implement the matchgate shadows described in Ref.~\citenum{wan2022matchgate} and numerically verify that the measurements of computational basis state overlaps follow a Gaussian distribution and are independent and identically distributed.
To determine a shot budget for TCCSD using the Gaussian noise model of matchgate shadows, we fit the total shot requirements to a particular error to a power law distribution with constants determined by well-known $\mathcal{T}_{1}$ or $\mathcal{D}_{1}$ diagnostics.\cite{lee1989diagnostic,janssen1998new} We analyze the robustness ec-CC with experimental Clifford classical shadows to emphasize the regime of utility and demonstrate the methods robustness to real experimental noise on using data taken from the recent QC-QMC experiment by Google and collaborators.~\cite{huggins_unbiasing_2022}

We continue with a methodological account and an overview of split-amplitude methods, but expert readers can directly skip to Section~\ref{sec:state_prep},
where we discuss the impact of possible quantum trial states and their quality requirements for split-amplitude quantum CC methods.
Furthermore, we present numerical results of quantum TCCSD by simulating of the \ce{N2} dissociation curve in a VQE-based NISQ setting, further highlighting state preparation robustness of TCCSD.
We then describe overlap measurement strategies in Section~\ref{sec:overlap_measurements}, which are the key ingredient for providing quantum inputs.
Based on numerical studies of the statistical properties of matchgate shadows, we develop the noise model to mimic finite shot noise, and we employ
this noise model to obtain quantum resource estimates for beyond classically tractable molecular systems. For an \ce{N2} dissociation curve of 21 points
within a double-zeta basis set, we find that a total of 30 million shots will be enough for chemically precise results.
Finally, in Section~\ref{sec:eccc_results_main}, we study the performance of quantum split-amplitude CC energetics on \ce{H4} and diamond,
using classical shadows data from an experiment performed on Google's Sycamore device, demonstrating that ec-CC can in principle compete with QC-QMC.

\section{Background on Electronic Structure Methods}
The field of quantum chemistry has seen remarkable progress over the past decades, with \textit{ab-initio} coupled cluster (CC)
methods emerging as a reliable and systematically improvable framework for electronic structure simulations of molecules.\cite{ciczek1966correlation,crawford2006introduction,shavitt2009many}
The most successful method of the CC family is probably coupled cluster with singles and doubles (CCSD)\cite{purvis1982full} in conjunction with
a perturbative triples correction, CCSD(T),\cite{raghavachari1989fifth} commonly referred to as the ``gold standard'' for quantum chemistry.\cite{bartlett2005how,guo2018communication}
Therefore, CCSD and CCSD(T) are possibly the most used wavefunction methods to investigate properties of small and medium-sized molecules,
limited in applicability by the steep scaling.
Recent endeavours, however, have made it possible to push the boundaries with respect to system size even further.\cite{guo2018communication}
Due to favorable properties, e.g., size consistency in contrast to truncated configuration interaction (CI) ansätze,\cite{sherrill1999configuration}
these single-reference CC (SRCC) methods usually perform well in describing molecular systems where dynamic
correlation effects are dominant.\cite{liakos2020comprehensive} On the contrary, SRCC approaches
can struggle in multi-reference (MR) scenarios, such as describing the potential energy surface
of bond breaking, and lead to catastrophic failures in several cases.\cite{kinoshita2005coupled,morchen2020tailored}
The MR failures can in principle be remedied through an even more computationally demanding framework, that is,
multi-reference CC (MRCC).\cite{oliphant1993multireference,evangelista2018perspective}
In case a system at hand is almost purely statically correlated, i.e., where the solution to the electronic Schrödinger equation is well described
by a couple of Slater determinants, accurate results can be obtained through so-called multiconfigurational
active space (AS) methods, such as complete active space configuration interaction (CASCI).\cite{knowles1984new,olsen1988determinant,knowles1989determinant,zarrabian1989vectorizable,bendazzoli1993vector}
The active space consists of a subset of molecular orbitals (MOs) and electrons in which the full configuration interaction (FCI)
problem is then solved. This approach covers static correlation effects well, but neglects major part of the dynamic correlation effects outside
of the chosen active space. Hence, the division into an active and external orbital space introduces approximations that, while qualitative phenomena
might be described correctly, can become unacceptable when chemical precision is required.\cite{stein2016delicate,takeshita2020increasing}
With techniques such as selected CI methods\cite{harrison1991approximating,holmes2016heat,sharma2017semistochastic,schriber2016communication,tubman2020modern} and density matrix renormalization group (DMRG),\cite{zhai2023block2} it becomes possible to simulate quite large active spaces,
or even aim toward FCI quality.
If the active space size cannot be increased enough, the dynamic correlation contributions can be approximately included through perturbation theory on the CAS method, yielding, e.g., the well-known CASPT2 and NEVPT2 methods.\cite{angeli2001introduction,angeli2001nelectron,angeli2002nelectron,pulay2011perspective,battaglia2023multiconfigurational}

Fortunately, the well-known framework of SRCC offers the possibility to capture static correlation properties from an AS-type wavefunction through
a so-called split-amplitude ansatz. The most prominent methods of this family are tailored coupled cluster (TCC)\cite{kinoshita2005coupled} and externally corrected coupled cluster (ec-CC),\cite{paldus1984approximate,paldus1994valence,planelles1994valence,planelles1994valence3,li1997reduced,paldus2017externally}
however, approaching the encoding of static correlation effects from different directions.
The high-level idea of both methods is to extract information about the static correlation from the AS wavefunction and inject them through specific partitioning of the cluster operator
and the corresponding amplitudes in a SRCC wavefunction. Originally, these methods were developed to remedy failures of SRCC approaches in the strongly correlated regime, as explained previously.
TCC and ec-CC are schematically illustrated in Figure~\ref{fig:illustrations}, and more theoretical details are outlined in Appendix~\ref{apx:tcc_eccc_theory}.
\begin{figure}[t]
    \centering
    \includegraphics[width=0.7\textwidth]{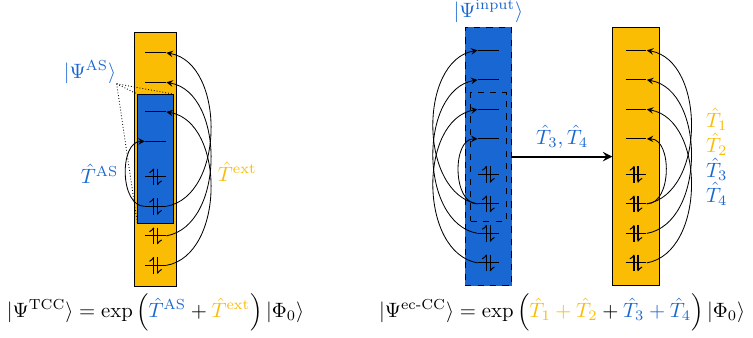}
    \caption{Schematic illustration of TCC (left) and ec-CC (right).
    In TCC, a subset of the cluster amplitudes are obtained from an active space wavefunction $\ket{\Psi}^\mathrm{AS}$ through cluster analysis.
    The external set of amplitudes is then solved in presence of frozen AS amplitudes, which encode the information on static correlation.
    The ec-CC ansatz uses T3 and T4 amplitudes from an input wavefunction (which can stem either from an active space or from the full space)
    and solves the full CCSDTQ singles and doubles equations in presence of the external T3 and T4 amplitudes.
    } \label{fig:illustrations}
\end{figure}
In the TCC ansatz, the cluster operator is split into an active space part and an external part,
\begin{align}
    \hat T^\mathrm{TCC} = \hat T^{\mathrm{AS}} + \hat T^{\mathrm{ext}}.
\end{align}
The cluster amplitudes, accompanying the active space cluster operator, are obtained from the AS-type multiconfigurational wavefunction through cluster analysis.
Cluster analysis relies on the equivalence of the exponential CC ansatz and the linear CI expansion using intermediate normalization, making it possible to recursively convert the amplitudes
from a CI-like wavefunction to their corresponding CC counterpart.\cite{monkhorst1977calculation,kinoshita2005coupled,lehtola2017cluster}
For CC with singles and doubles, this yields the so-called TCCSD approach,
where the strategy is to extract T1 and T2 amplitudes from the AS-type wavefunction. The active space T1 and T2 amplitudes are then
frozen during the CC iterations on the external amplitudes, assuming the active space amplitudes retain the MR information.
The original work by Bartlett and co-workers showed that, despite its simplicity, TCCSD yields dissociation energies and potential energy surfaces which
are in good agreement with higher-level MR methods.\cite{kinoshita2005coupled}
Note that the only inputs required for TCC are overlaps of (excited) Slater determinants with the AS wavefunction, which can be formulated as the
following projection,
\begin{align}
    \mathbf{c}_\nu = \braket{\Phi_\nu | \Psi^\mathrm{AS}}. \label{eq:overlap_determinants}
\end{align}
In the equation above, $\mathbf{c}_\nu$ is the CI coefficient vector of excitation level $\nu$, $\ket{\Phi_\nu}$ is a $\nu$-fold excited Slater determinant with respect to the
Fermi vacuum/reference determinant $\hfdet$, and $\ket{\Psi^\mathrm{AS}}$ is an AS-type wavefunction from which the coefficients are to be extracted.
For TCCSD, only the C1 ($\nu = 1$) and C2 ($\nu = 2$) coefficients need to be extracted, i.e.,
\begin{align} 
    c_i^a &= \braket{\Phi_i^a | \Psi^\mathrm{AS}}, \label{eq:C1}\\
    c_{ij}^{ab} &= \braket{\Phi_{ij}^{ab} | \Psi^\mathrm{AS}}, \label{eq:C2}
\end{align}
in addition to the overlap with the reference determinant, $c_0$, if $\ket{\Psi^\mathrm{AS}}$ is not intermediate-normalized.
The indices $i,j$ refer to occupied orbitals in \hfdet, whereas $a,b$ are so-called virtual orbital indices, i.e., unoccupied in \hfdet.
Subsequently, the C amplitudes are recursively converted to T amplitudes, mapped from the active orbitals to the full MO space, and kept fixed during solution
of the CCSD projection equations. It is quite appealing that only minor modifications to a standard CCSD code need to be made to support tailoring.
Different types of input wavefunctions have been successfully used for TCCSD, e.g., CASCI,\cite{kinoshita2005coupled,hino2006tailored} DMRG for large active spaces,\cite{veis2016coupled} and pair coupled cluster doubles (pCCD).\cite{leszczyk2022assessing}
A recent extension to excited state calculations has been presented by Bartlett and co-workers.\cite{ravi2023excited}
There exists the tailored counterpart to CCSD(T), that is, TCCSD(T),\cite{lyakh2011tailored} which computes the perturbative
triples correction solely using the external contribution of the cluster amplitudes.
Through combination with the LPNO and DLPNO framework, TCC has been successfully employed to study the electronic structure of large
molecular complexes.\cite{antalik2019toward,antalik2020ground,lang2020near} Furthermore, the underlying numerical and theoretical
aspects have been thoroughly analyzed from a mathematical point of view.\cite{faulstich2019analysis,faulstich2019numerical}
Even though the tailoring of T1 and T2 amplitudes captures the static correlation effects to some extent, the final TCCSD wavefunction is still
a SR wavefunction, which is known to fail in some cases.\cite{melnichuk2012relaxed,morchen2020tailored,demel2023hilbert}
To some level, these shortcomings can be circumvented by increasing the active space size and the choice of orbitals.\cite{bartlett2005how,morchen2020tailored}

The second split-amplitude approach, ec-CC, addresses the inclusion of multi-reference phenomena in a SRCC wavefunction from a different direction.\cite{paldus2017externally}
Viewed as a variation of tailored coupled cluster but on the full space~\cite{deustua2018communication, PhysRevLett.119.223003} we use the fact that the singles and doubles residual equations can only involve singles amplitudes through quadruples amplitudes
\begin{align}
\langle \Phi_{i}^{a}|\left[\hat H_{N}\left(1 + \hat T_{1} + \hat T_{2} + \frac12 \hat T_{1}^{2} + \hat T_{3} + \hat T_{1} \hat T_{2} + \frac16 \hat T_{1}^{3}\right)\right]_{c}|\Phi_0\rangle &= 0, \\
\langle \Phi_{ij}^{ab}|\left[\hat H_{N}\left(1 + \hat T_{1} + \hat T_{2} + \frac12 \hat T_{1}^{2} + \hat T_{3} + \hat T_{1}\hat T_{2} + \frac16 \hat T_{1}^{3} + \hat T_{4} + \hat T_{1} \hat T_{3} + \frac12 \hat T_{2}^{2} + \frac12 \hat T_{1}^{2} \hat T_{2} + \frac{1}{24} \hat T_{1}^{4}\right)\right]_{c}|\Phi_0\rangle &= 0 ,
\end{align}
which is true for any rank coupled cluster operator, above four, because the normal-ordered Hamiltonian $\hat H_N$ has at most two-body interactions.
If the triples and quadruples amplitudes in $\hat T_3$ and $\hat T_4$ are exactly determined with, e.g., FCI, then the \emph{exact} correlation energy can be recovered.
This is to be expected since the exact amplitudes in $\hat T_3$ and $\hat T_4$ depend on CI coefficients of rank $1,2,3$, and $4$.
This can, however, provide an alternative route to better than CCSD calculations if approximate triples and quadruples are obtained from an AS-type method.
While there are obvious instance where ec-CC provides no value (e.g., wavefunctions where cluster operators of rank 2 to 4 are used and their active space versions -- CCSDt, and CCSDtq), there are cases where improvements are observed.\cite{magoulas2021is} The key components of when an ec-CC treatement can provide improvement are described in Ref.~\citenum{magoulas2021is} and will be summarized in the context of quantum wavefunctions in Section~\ref{sec:state_prep}.
The main idea is that if one uses frozen triples and quadruples amplitudes obtained in some way and solves for the singles and doubles in their presence, one can determine a corrected form of the quadruples CC involving MR character.
In practice, a non-CC wavefunction theory in an active space is used to determine the dominant triple and quadruple excitations, and the cluster operator takes on the form
\begin{align}
\hat T = \hat T_{1} + \hat T_{2} + \hat P_{3} \hat T_{3} + \hat P_{4} \hat T_{4},
\end{align}
where $\hat P_{3,4}$ project out a subset of triples and quadruples excitation that are dictated from a correlated multi-reference calculations such as DMRG, heat-bath CI,\cite{lee2021externally} or adaptive CI.\cite{aroeira2020coupled}
Alternatively, one can use an uncoverged FCIQMC calculation to recover approximate triples
and quadruples amplitudes to then solve for for the $\hat T_1$ and $\hat T_2$ parts.\cite{magoulas2021is}
It was demonstrated that roughly converged FCIQMC calculations can result in quite accurate CCSDTQ energies.
The role of the FCIQMC calculation is to determine dominant triple and quadruple amplitudes.
A caveat is that the classical computational effort of ec-CC scales like $N^8$ (see Appendix~\ref{apx:ec_cc_scaling}), which, contrary to TCCSD with a $N^6$ scaling, limits the applicability of ec-CC to systems of modest size.
TCC and ec-CC require the same kind of inputs, namely the overlaps of Slater determinants with the AS-type wavefunction which is then put into a cluster
analysis protocol. In the context of quantum computation, this has the appealing advantage that no higher-order RDMs are required to enable
a split-amplitude CC calculation based on a quantum trial state.

\section{Impact of Quantum Trial States Quality}\label{sec:state_prep}
In this section we first investigate the stability of TCCSD against overlaps derived from imperfect wavefunctions, such as those states prepared by a shallow VQE circuit. 
We find that even wavefunctions whose AS energy differs significantly from CASCI do yield a TCCSD dynamic correlation energy correction that is in good agreement with both TCCSD and NEVPT2 based on the exact CASCI wavefunction.
What is more surprising, even a rather inaccurate wavefunction as input to TCCSD turns out to be sufficient to cure the appearance of a qualitatively incorrect reaction barrier in plain CCSD.
We are further able to provide insights into which wavefunctions have a chance of leading to an improved energy under ec-CC on the full orbital space. 

\subsection{Quantum TCCSD with Approximate VQE Wavefunctions} \label{sec:vqe_exales}
We model the dissociation of the \ce{N2} triple bond, a textbook example in which the dynamic and static correlation regime need
to be converged well for quantitatively reliable results.\cite{kinoshita2005coupled} In the first TCCSD work,\cite{kinoshita2005coupled} it has been shown
that TCCSD tailored by CASCI with 6 electrons in 6 orbitals, CAS(6,6), gives a qualitatively and quantitatively correct dissociation curve, in contrast to plain CCSD.
Our goal here is to assess whether TCCSD tailored by one of the most prominent quantum methods for simulating electronic structure, VQE, behaves similarly. Since the underlying
VQE circuits which would be required to prepare the exact CASCI state are not affordable, we want to analyze how inaccuracies in the wavefunction ansatz translate to TCCSD, or whether
this even leads to a breakdown of the beneficial properties of TCCSD after all.
For this purpose, we set up a quantum-number-preserving (QNP)\cite{anselmetti2021local} circuit of 10 layers (50 parameters), which was optimized for each geometry along the
dissociation curve. The reference state was obtained with restricted HF/cc-pVDZ.\cite{dunning1989a} For the ``stress test'' with respect to rather shallow VQE circuits,
we then ran TCCSD based on the overlaps obtained analytically from the final VQE state (through read-out of the coefficients of the VQE state vector), referred to as VQE-TCCSD in the following.
The dissociation curves for CCSD, CASCI, NEVPT2, TCCSD, VQE, and VQE-TCCSD are shown in Figure~\ref{fig:n2_dissociation} in the left panel.
\begin{figure}[ht]
    \centering
    \includegraphics[width=1\textwidth]{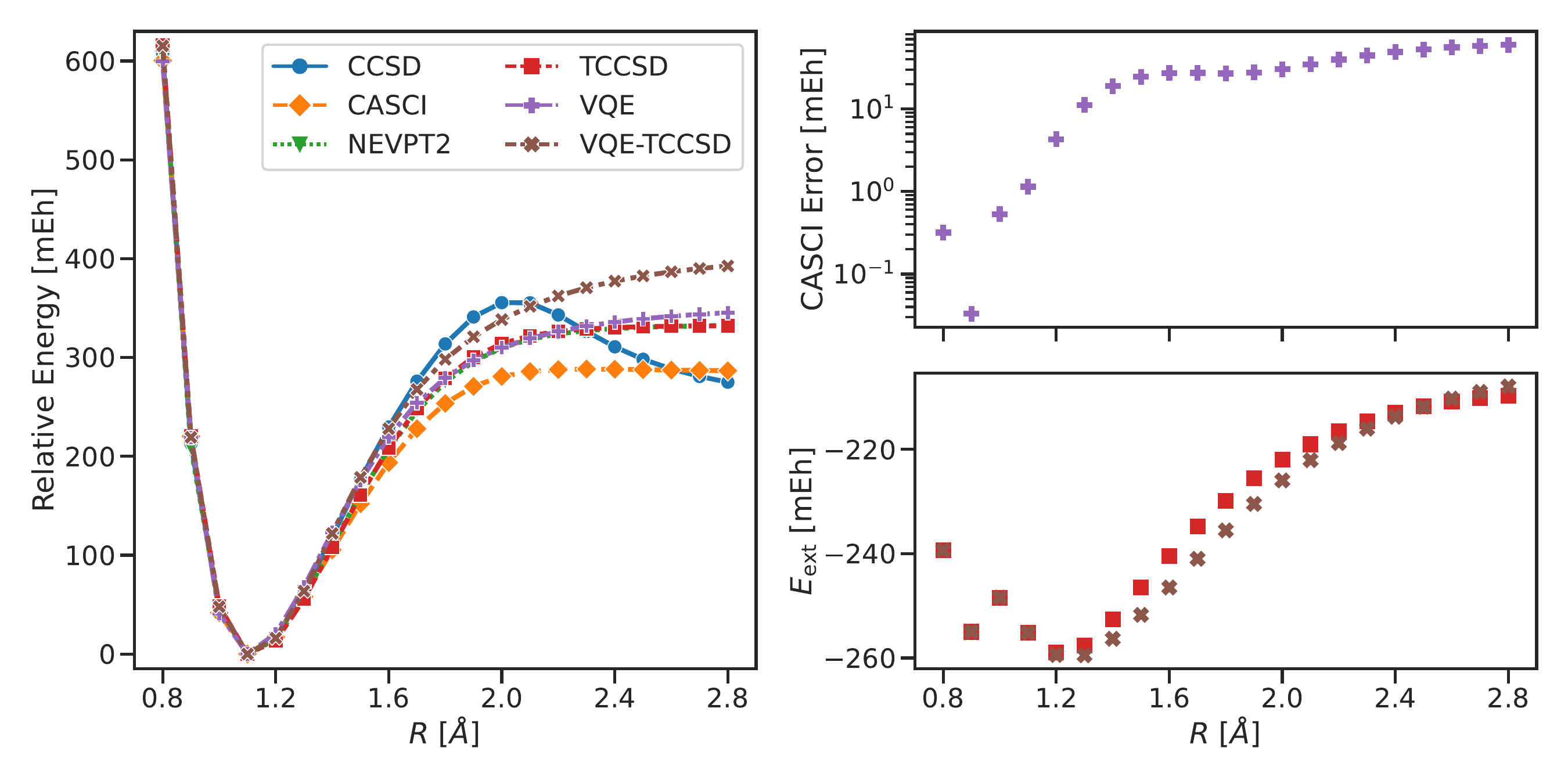}
    \caption{
    Dissociation curve of \ce{N2} using classical and quantum methods (left panel), CASCI energy error of the
    employed VQE ansatz (upper right panel), and external correlation energy $E_\mathrm{ext}$ of TCCSD methods (lower right panel). For the dissociation curve, energies are relative to the minimum energy the inter-nuclear distance $R = 1.1$~\AA. Note that the relative energies of TCCSD and NEVPT2 almost lie on top of each other.
    }
    \label{fig:n2_dissociation}
\end{figure}
As previously observed, the equilibrium region is well described by CCSD, but this method yields a qualitatively incorrect virtual reaction barrier upon triple bond breaking,
which is fully recovered by TCCSD.\cite{kinoshita2005coupled} Plus, TCCSD provides an accurate numerical value for the dissociation energy, whereas the CASCI curve results in a too
small dissociation energy. This shows how important the inclusion of dynamic correlation effect becomes for accurate simulations. Comparing CASCI to TCCSD, the barrier is shifted
from $286.8$ milli-Hartree (mEh) by $45.4$ mEh to $332.0$ mEh, which is quite a significant change.
For comparison, we ran NEVPT2 on CASCI for this system, which yields almost exactly the same dissociation curve as TCCSD.
Now, due to the shallow VQE circuit used in our simulations, the predicted reaction barrier is artificially too high and VQE-TCCSD cannot repair the incorrect energy contribution from the active space.
If the VQE circuit were expressive enough,
one would obtain a dissociation curve identical to CASCI. In the upper right panel in Figure~\ref{fig:n2_dissociation}, the error of the VQE energy with respect to CASCI is shown.
For bond distances $R > 1.2$~\AA, the energy error is larger than 10 mEh and steadily increases toward the dissociation limit. Around the equilibrium bond distance, the system
possesses only weak static correlation effects, such that the shallow VQE circuit is of course much more accurate in that region.
Looking at the VQE-TCCSD dissociation curve, no artificial reaction barrier is present, i.e., despite the poor quality of the VQE ansatz, the quality of the underlying wavefunction seems
good enough to heal the physically wrong CCSD behavior. Interestingly, the VQE dissociation energy ($345.3$ mEh) is shifted in VQE-TCCSD ($392.6$ mEh) by $47.3$ mEh, which is almost
identical to the energy shift from CASCI to TCCSD, i.e., without imperfections in the wavefunction.
This parallel can be narrowed down to the external energy contribution in TCCSD, $E_\mathrm{ext}$, depicted for TCCSD and VQE-TCCSD in the lower right panel in Figure~\ref{fig:n2_dissociation}.
The contribution from the external TCCSD part is almost identical, showing that the dynamic energy contribution in TCCSD is very robust against a poor CI-like wavefunction produced from the shallow VQE circuit, and
the error in energy shift from the CI-like method to TCCSD is approximately one order of magnitude smaller than the plain VQE error.
Note that the quantum inputs for VQE-TCCSD amount to only measuring less than $n^4$ overlap values, whereas a possible PT approach would require
higher-order RDMs.

\subsection{Wavefunctions Suitable for Externally Corrected CC}
In this section we discuss the characteristics of a quantum trial and the potential benefit over classically accessible wavefunctions as input to ec-CC. Within the NISQ setting there is no shortage of methods for preparing approximate ground states.\cite{peruzzo2014variational, mcclean2016theory, gard2020efficient,anselmetti2021local, grimsley2019adaptive, PRXQuantum.2.020310, PRXQuantum.2.030301, lee2018generalized,o2019generalized, matsuzawa2020jastrow, motta2020determining, kim2017robust, sewell2021preparing, PhysRevA.81.050303, miao2021quantum} If we consider fault-tolerance, an even wider set of methods is available.\cite{lin2020near, malz2023preparation, ge2019faster, PRXQuantum.2.020321, he2022quantum, kyriienko2020quantum, fomichev2023initial} 

Here we emphasize recent work from Magoulas \textit{et al.}~\cite{magoulas2021is} which provides a framework for analyzing types of CI expansions, i.e., the source of the external cluster amplitudes,
that potentially yield an improved ec-CC energy.
That work servers as a blueprint for the type of wavefunction that a quantum computer would need to prepare to potentially see an energy improvement through solving the ec-CC equations. 
At the core of their derivation is the following theorem:
\begin{theorem}
The solution to a truncated configuration interaction set of equations which includes full singles, full doubles, and any set of higher excitations satisfies the coupled cluster singles and doubles equations.
\end{theorem}
They prove this theorem algebraically and diagrammatically. We have reproduced the algebraic proof in Appendix~\ref{app:eccc_proof} for completeness. 
A key takeaway is that in order for the ec-CC equations to provide an improvement over CI expansions one must only include the connected components -- i.e., rank three and four cluster amplitudes with non-zero CI-ampltiudes. Furthermore, to see improvement over CI expansions via ec-CC requires a CI expansion that includes some but not all triple and quadruple excitations. This suggests that a quantum circuit exploring dominant many-body excitations for a large system can be a useful input to ec-CC. Supporting this interpretation are the classical studies using FCIQMC as an external source and other methods that sample high-energy Slater determinants.\cite{deustua2018communication} While the connection to truncated CI should not be considered too strongly as an analogy for quantum circuits, it does provide support for the use of particular quantum circuits that sample high-energy many-body excitations. This new perspective serves as a different design principle when constructing quantum circuit ansatze with the potential for beyond classical computation. The connection rules out quantum state ansatze where it is efficient to estimate amplitudes up to additive error such as a circuit built from a fixed bond dimension MPS.
No comparably generally statements can be made about when quantum inputs can be useful for TCC or ec-CC in an active space.

\section{Quantum Measurement of Overlaps} \label{sec:overlap_measurements}
In the quantum CC methods proposed here, the input amplitudes are determined from quantum state overlaps with the help of a quantum computer.
When a Jordan-Wigner mapping of fermions to qubits is used, the Slater determinant overlaps correspond to overlaps with computational basis states.
As the non-relativistic second-quantized molecular Hamiltonian can, without loss of generality, be chosen to be real when written in the computational basis, the eigenstates are real, so that if a suitable state preparation method is used, all overlaps with computational basis states should be real.
However, since the signs of the overlaps enter the cluster analysis, sampling the quantum trial state $\trial$ on $n$ qubits with $\zeta$ electrons in the computational basis is not sufficient. Various methods are known that can be used to estimate overlaps including signs.
A family of methods suitable for the task that has received a lot of attention recently are classical shadows\cite{huang2020predicting} and extensions of this technique geared specifically towards the fermionic setting.\cite{zhao_fermionic_2021, wan2022matchgate, low_classical_2022}
All these methods have in common that, given a state $\trial$, one draws unitaries $U$ from some ensemble of classically efficiently describable unitaries. A description of the drawn unitaries together with the results of computational basis state measurements in the state $U \ket{\Psi^T}$ is recorded as a so-called classical shadow.
From this classical shadow, one can, using purely classical computation, predict various properties of the state $\trial$.
Particularly relevant for our work is the protocol from Ref.~\citenum{wan2022matchgate} which is based on an ensemble of matchgate circuits (Haar measure over general fermionic Gaussian unitaries), which allows to estimate all overlaps with Slater determinants up to additive error $\epsilon$ from shadows consisting of $s \in \mathcal{O}(\sqrt{n} \log(n)/\epsilon^2)$ quantum measurements or shots.
The fully parallelizable classical computational effort per overlap scales as $\mathcal{O}(s \, (n-\zeta/2)^4)$ = $\mathcal{O}(\sqrt{n} \log(n) \, (n-\zeta/2)^4 /\epsilon^2)$ (improvements to a $(n-\zeta/2)^3$ scaling are possible, see Appendix D of Ref.~\citenum{wan2022matchgate}).
For this protocol, one needs to prepare a superposition of a reference state (e.g., the true vacuum $\ket{0}^{\otimes n}$) and $\trial$.
If the state preparation method preserves the fermion number, this can be achieved by applying the circuit that prepares $\trial$ from the Hartree-Fock state $\hfdet$ to a state that is a superposition of $|\Phi_0\rangle$ and $\ket{0}^{\otimes n}$. This state can be prepared in depth $\mathcal{O}(\log(\zeta))$ by using a single Hadamard gate and replacing the other Pauli $X$ gates needed to prepare $\hfdet$ by their controlled versions (CNOT).
Alternatively one can use the classical shadow protocol from Ref.~\citenum{low_classical_2022}, which randomizes only over number-preserving (passive) Gaussian unitaries $U$ and allows to compute all overlaps to error $\epsilon$ from shadows consisting of just $s \in \mathcal{O}(4\epsilon^{-2}/3)$ shots (which is independent of $n$ and $\zeta$) or the Clifford shadow protocol from Ref.~\citenum{huggins_unbiasing_2022}, which also has an $n$-independent sample complexity, however scaling with the logarithm of the number of overlaps.
In the first case, due to number preservation of the passive Gaussian unitaries, the superposition with the reference state must be prepared on an enlarged set of up to $3n/2$ qubits in the worst case of half-filling $\zeta = n/2$.
In the second case, the classical processing of the shadow data is efficient for computational basis state overlaps required by our method (just not for overlaps with general Slater determinants).\cite{huggins_unbiasing_2022}
In our numerical simulations, we focus on the matchgate shadow protocol from Ref.~\citenum{wan2022matchgate}, which has the worst scaling of the required number of shots among the three alternatives.
This means that the shot budgets reported here can asymptotically be thought of as upper bounds and probably be further improved upon for finite $n$.

\subsection{Statistical Properties of Matchgate Shadow Overlaps}
We implemented the matchgate shadow protocol in order to study the statistical properties of finite shot overlap measurements, summarized in detail in Appendix~\ref{apx:stat_props}.
From our numerical simulations, we obtained the following findings:
1.) The overlap estimates are approximately normal
distributed, 
2.) the covariance matrix of overlap measurements is diagonally dominant, and the covariances vanish asymptotically faster than the variances with increasing $s$, and scatter plots (see Fig.~\ref{fig:overlap_convariance_matrix} in Appendix~\ref{apx:stat_props}) confirm that the overlap estimates are close to independently distributed for large $s$, 
3.) the spread of the variances decays faster than the mean variance, meaning that for large $s$, all overlap estimates have approximately the same variance, and 
4.) the numerically observed mean variance $\bar\sigma^2$ agrees well with the analytical performance guarantees from Ref.~\citenum{wan2022matchgate} (see Appendix~\ref{apx:stat_props} for more details).
For the case of half-filling, $\zeta = n/2$, we numerically found the simple relation 
\begin{align}
    \bar\sigma^2 \lessapprox \sqrt{2 n} / s. \label{eq:variance_bound}
\end{align}
Ultimately, this allowed us to build a synthetic noise model, in which we can efficiently
add Gaussian noise with variance $\bar\sigma^2$ to classically computed exact overlaps to ``mimic'' matchgate shadow overlap measurements without
having to classically simulate the entire shadow protocol.
In the next section, we use the synthetic noise model to estimate
the number of quantum measurements $s$ in order to reach chemical precision for sizeable systems, which are intractable on current quantum devices.
Given a finite shot shadow, one can estimate overlaps and the variance of these estimates.
This enabled us to build a classical post-processing scheme to screen out overlaps that were not statistically significantly determined,
improving the results of ec-CC in Section~\ref{sec:eccc_results_main}.

\subsection{Shot Noise Resilience and Quantum Resources for Tailored Coupled Cluster} \label{sec:results_tccsd}
\begin{figure}[ht]
\centering
    \hfill
    \begin{minipage}{0.49\textwidth}
    \begin{subfigure}[T]{0.81\textwidth}
        \centering
        \includegraphics[width=\textwidth]{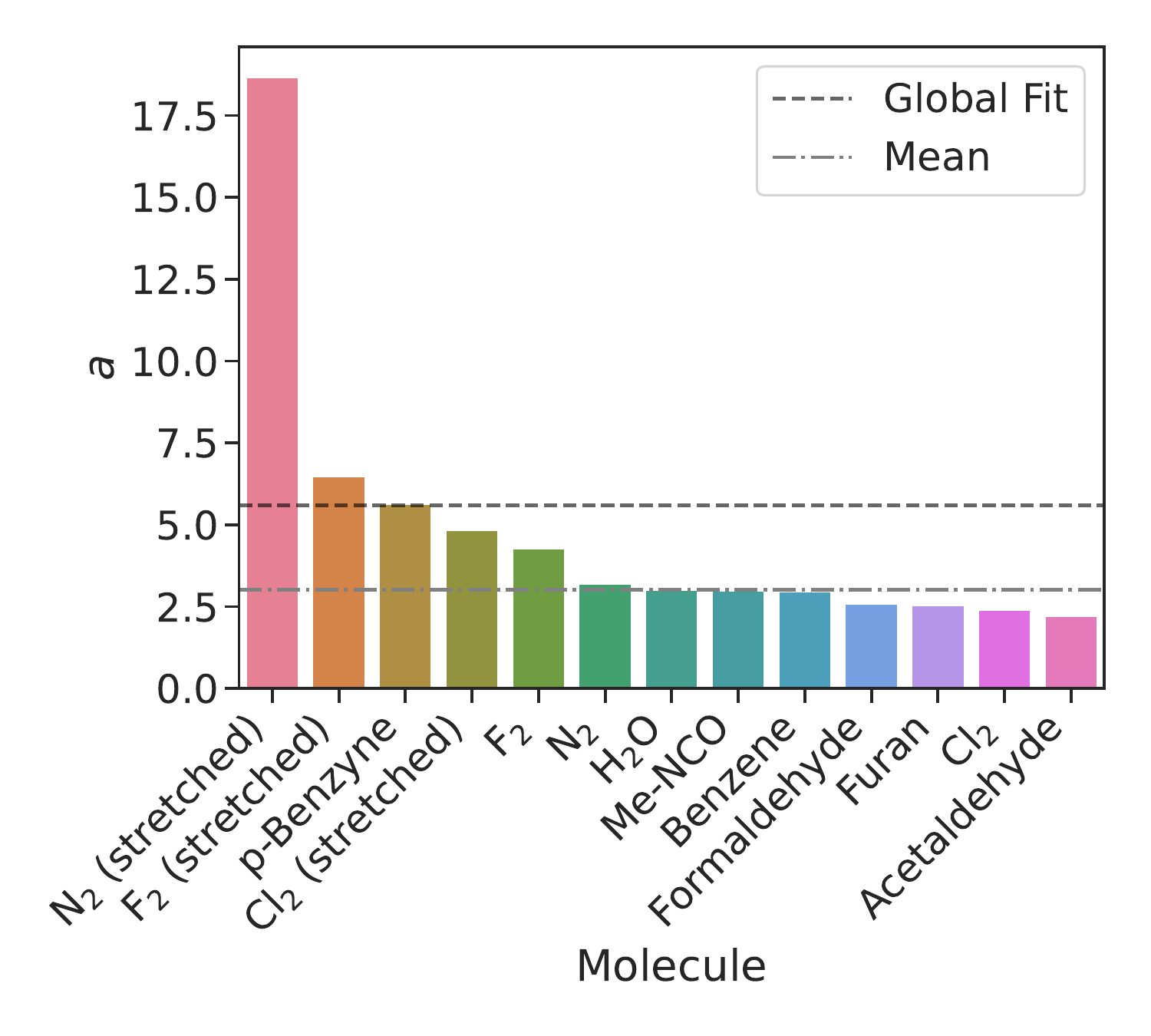}
        \caption{}
    \end{subfigure}
    \\
    \begin{subfigure}[T]{0.81\textwidth}
        \centering
        \includegraphics[width=\textwidth]{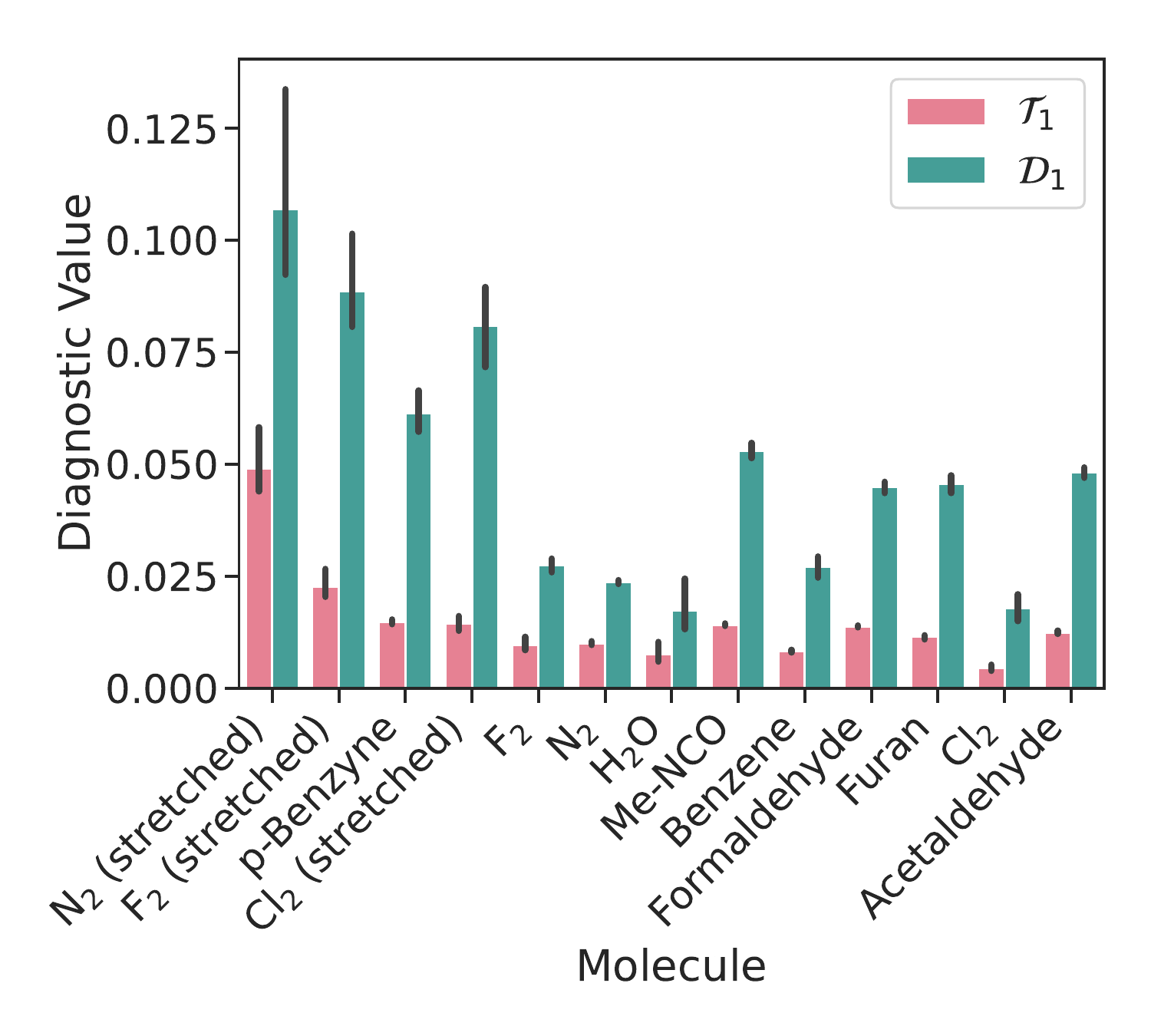}
        \caption{}
    \end{subfigure}
    \end{minipage}
    \hfill
    \begin{minipage}{0.49\textwidth}
    \begin{subfigure}[T]{0.81\textwidth}
        \centering
        \includegraphics[width=\textwidth]{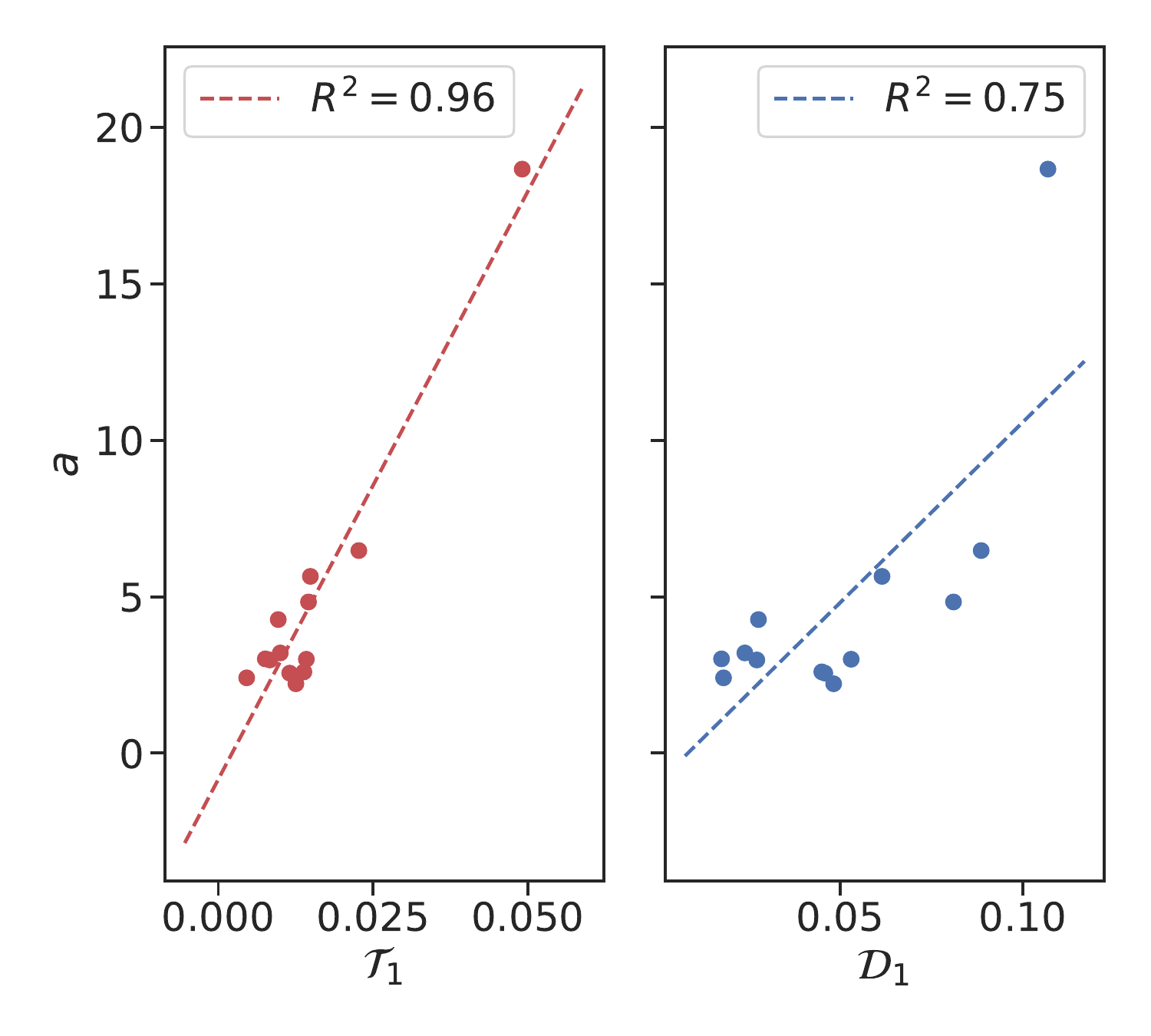}
        \caption{}
    \end{subfigure}
    \end{minipage}
    \hfill
\caption{Fitted power law prefactors $a$ for each molecule at fixed global exponents (a),
Diagnostic values for each molecule (b), scatter plot of the prefactor $a$ vs.\ the diagnostic values including a
linear fit (c).
} \label{fig:prefactors_molecules}
\end{figure}

The purpose of this section is to analyze the finite shot budget for
matchgate shadows in order to obtain (noisy) quantum TCCSD energies for molecular
systems to a given precision.
The observation that noisy overlaps extracted from the matchgate shadow protocol
are independent and identically distributed random variables with an underlying
normal distribution greatly facilitates the construction of an error model of
the TCCSD energy, since we do not need to record many finite shot shadows.
Rather, one can take the exact overlaps from a classical CASCI calculation and
subsequently draw the noise of a given distribution with standard deviation or noise strength $\sigma$, and add it to the exact overlaps. Subsequently, we compute the
``noisy'' TCCSD energy without having to simulate a single quantum circuit.
Ultimately, the goal is to find a heuristic model for the shot budgets, which will incorporate a) the dependence of the error on the noise strength $\sigma$, which can be directly related to the shot budget via the variance bound published by Wan \textit{et al.},\cite{wan2022matchgate} and b) system/molecule-specific parameters (number of spin orbitals, active space size, etc.) making it possible
to estimate the quantum resources across a broad range of molecules of interest without
having to perform the actual sampling of the TCCSD error for that system.
The details of our analysis and construction of the empirical error model are explained in detail in Appendix~\ref{sec:influence_of_noise}.
The quantity we want to extrapolate using only pre-defined system parameters and the noise strength is the absolute TCCSD energy error caused by noise, \detcc. Due to the relation with the noise strength, it will then be possible to substitute $\sigma$ with the variance bound in eq \eqref{eq:variance_bound}, which
in turn contains the shot budget $s$. Thus, if we can convincingly model \detcc, we can directly obtain a shot budget given the desired
accuracy in \detcc, and the system-specific parameters, which we will outline in the following.
Most importantly, we tested the dependence of \detcc~(for a given system) on the noise strength while keeping all other
influential parameters (basis set, etc.) fixed. This analysis revealed a linear dependence of \detcc on $\sigma$, i.e., the
TCCSD energy error is proportional to the underlying noise strength. The linear relationship was observed consistently for different systems
and active space sizes. Even at high noise strengths ($\sigma = 10^{-1}$), the TCCSD calculations converged, indicating that the procedure
is robust even when tailored with almost random overlaps.
\begin{figure}[ht]
    \centering
    \includegraphics[width=0.5\textwidth]{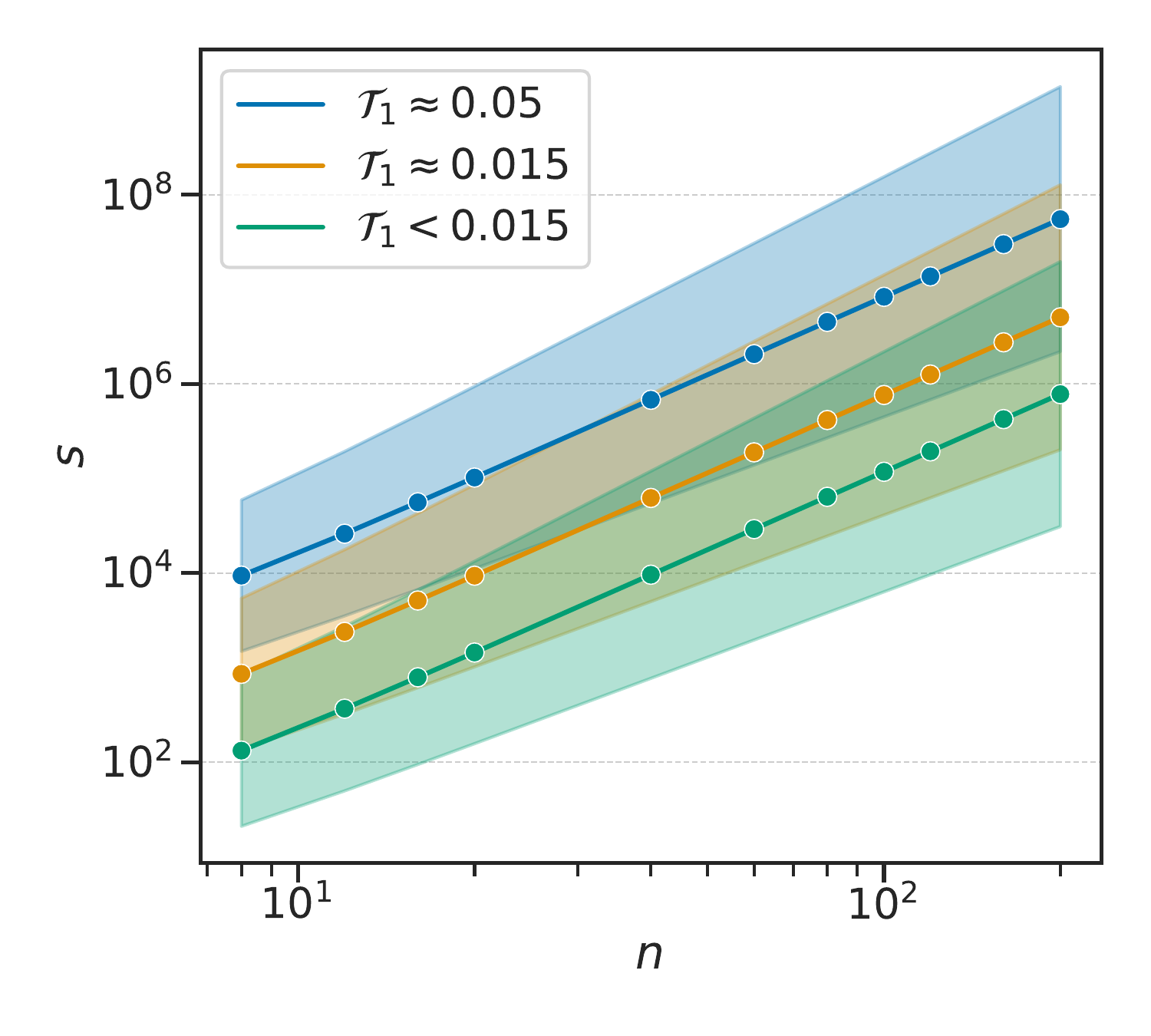}
    \caption{Extrapolation of the shot count $s$ needed for determining the quantum TCCSD energy to milli-Hartree precision in different correlation regimes, as indicated by the $\mathcal{T}_1$ diagnostic.
    The data were generated for a system with $N=600$ orbitals and half-filled active spaces $\zeta = n/2$.
    Shaded regions correspond to power laws with the smallest and largest exponents (i.e., estimator $\pm$ standard deviation, respectively).
    }
    \label{fig:shot_count_extrapolation}
\end{figure}
With the linear dependence on $\sigma$ at hand, we sought to incorporate molecule/system-specific parameters into an extended extrapolation model for
\detcc. We composed a set of molecular systems, basis sets, and active space sizes and sampled the TCCSD energy for these systems at
a given noise strength, such that we could afterwards find the system-dependent variables that explain trends in the error best.
From analysis of the sampled data set, we concluded that the total number of overlaps $d$ that are put in to the TCCSD calculation, which directly depends on $n$ and $\zeta$, i.e., the chosen active space size (see Appendix \ref{apx:tccsd_energy_error_extrapolation}), and the total number of spin orbitals of the system $N$ yield good results in a power law
fit over the whole data set. The power law is given by
    \begin{align}
    \detcc = a \, d^\beta N^\gamma \sigma. %
\end{align}
The exponents of $d$ and $N$ in the power law seem to be largely independent of the chosen
molecular system. The exponent of $d$, given as $\beta = 0.277 \pm 0.054$, shows that the error due to noisy amplitudes
increases with the size of the active space, as expected. The exponent of $N$ is approximately $\gamma = -1.074 \pm 0.116$, such that the absolute error
decreases when the total size of the MO space increases. This makes sense, because the active space contribution to the total energy decreases
when the relative size of the external space increases.

The prefactor $a$ encodes system-specific variables that were omitted by our choice of $d$ and $N$ as the most important
general quantities. We noticed that the stronger the static correlation in the system was, the larger was the corresponding prefactor
$a$ in a per-molecule fit with fixed exponents for $\beta$ and $\gamma$. The per-molecule prefactors are depicted in Figure \ref{fig:prefactors_molecules}.
We found that the prefactor $a$ across the molecule test set correlates with the well-known CCSD diagnostic values, $\mathcal{T}_1$ and $\mathcal{D}_1$ (see Fig.~\ref{fig:prefactors_molecules}).\cite{lee1989diagnostic,janssen1998new}
These diagnostic values are typically used to quantify whether a single-reference CCSD wavefunction is reliable, or whether one should opt for a MR
treatment of the system. The $\mathcal{T}_1$ values correlate best with the molecule-specific prefactor $a$ (Fig.~\ref{fig:prefactors_molecules}),
and we use a linear fit to convert from $\mathcal{T}_1$ to $a$ (fit parameters are shown in Appendix \ref{apx:tccsd_energy_error_extrapolation}).

The empirical error extrapolation model takes into account the following quantities; i) the noise strength $\sigma$,
given through the variance bound in eq \eqref{eq:variance_bound}, ii) the number of overlaps $d$, required
for TCCSD, which depends on the size of the chosen active space, iii) the total number of spin orbitals in the system $N$, and iv)
the prefactor $a$ of the power law, derived from the $\mathcal{T}_1$ diagnostic through a plain CCSD calculation.
Of course, the empirical model is far from general, however, it may give some guidance for the error level one can expect
in a TCCSD calculation based on noisy overlap measurements.
Inserting the variance bound (eq \eqref{eq:variance_bound}) into the power law, one can rearrange the equation for the matchgate shadow shot budget at half-filling,
\begin{align}
s \lessapprox \frac{a^2}{\detcc^2} d^{2\beta} N^{2\gamma} \sqrt{2n}. \label{eq:shot_budget_half_filling}
\end{align}
The shot budget $s$ can then be computed for a target accuracy \detcc~with the other parameters obtained from the fit and the
molecule/setup-specific values. In the following, we set $\detcc = 10^{-3}$~Eh as the required ``chemical precision''.
Note that this shot budget is just assigned for the overlap measurements through matchgate shadows
and does not take into account the shot budget for, e.g., state preparation or optimization. Furthermore, the Gaussian error model is herein
applied to the \emph{exact} overlaps from a CASCI trial wavefunction. For this reason, the error model just covers errors in the overlap
measurements but not the imperfections in the wavefunction we have discussed in the previous section. Since the error model is
independent of the underlying state, however, both error contributions can be studied separately.
We computed the values for $s$ for a fictitious system with an
MO space of $N = 600$, for half-filling active spaces from $n=4$ to $n=200$ qubits with $\zeta = n/2$ electrons. To cover different
correlation regimes, the shot budgets were obtained for three ranges of $\mathcal{T}_1$, i.e., $\mathcal{T}_1 < 0.015$,
$\mathcal{T}_1 \approx 0.015$, and $\mathcal{T}_1 \approx 0.05$, as plotted in Figure~\ref{fig:shot_count_extrapolation}.
The shaded areas around the solid lines in the plot include the worst/best-case scenario, where the standard deviations of
the exponents have been added/subtracted.
The strongly correlated case, $\mathcal{T}_1 \approx 0.05$, corresponds to a system like \ce{N2} with a bond distance of $2.8$~\AA, i.e., an
extreme case for MR character. The mixed/balanced case, $\mathcal{T}_1 \approx 0.015$, where CCSD is almost not reliable anymore (usually the threshold for $\mathcal{T}_1$ is around $0.02$), corresponds to a system like \textit{p}-Benzyne, and the weakly correlated case to something like
closed-shell organic molecules.
For the balanced case, the shot counts range from approximately $10^3$ for 4 qubits to less than $10^7$ for 200 qubits.
The uncertainty of these numbers is given through the standard deviation of the exponent fit parameters, and amounts to
approximately one order of magnitude in each direction.

\subsection{Quantum Resource Estimates for Nitrogen Dissociation}
\begin{figure}[ht]
    \centering
    \includegraphics[width=0.4\textwidth]{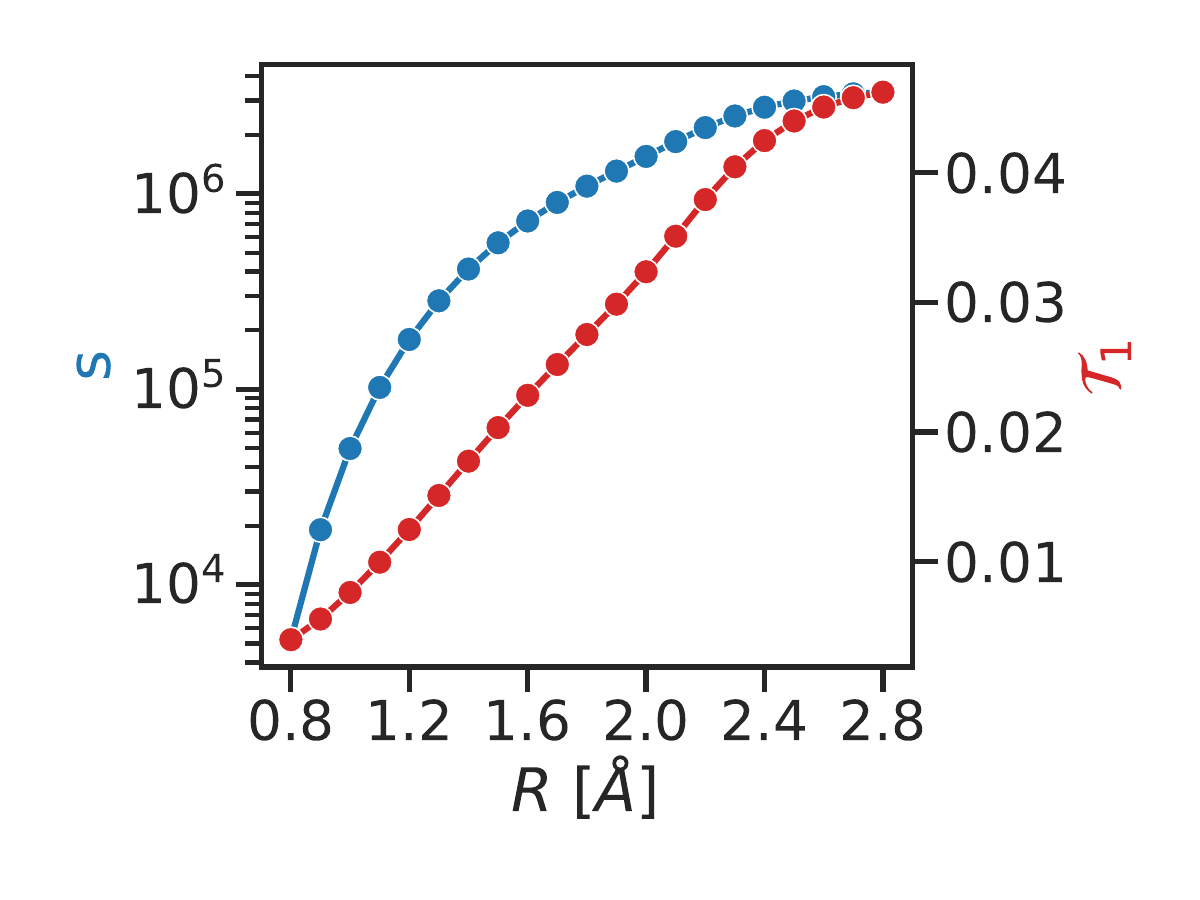}
    \caption{Estimated shot budget $s$ for \ce{N2}/cc-pVDZ/CAS(6,6)
    to achieve milli-Hartree precision in the TCCSD energy as a function of the internuclear distance $R$ along the dissociation curve.
    The value of $\mathcal{T}_1$, obtained from CCSD/cc-pVDZ, determines the increase in $s$, according to eq \eqref{eq:shot_budget_half_filling}.
    The total shot budget amounts to approximately 30 million shots. Raw data are shown in Table \ref{tab:shot_budgets_n2}.
    }
    \label{fig:shots_n2}
\end{figure}
Coming back to the test case in the previous section, we computed the total shot budgets to obtain the \ce{N2}
dissociation curve with chemical precision in a noisy TCCSD setup.
The shot budget per point on the dissociation curve is given in Figure~\ref{fig:shots_n2}. While the number of overlaps and the
number of spin orbitals is of course identical on each point on the potential energy surface, the value of $\mathcal{T}_{1}$ increases
and thus the shot budget is adjusted accordingly. Note that the shot budget grows quadratically in $\mathcal{T}_{1}$, and
the total shot budget for all 21 bond distances is dominated by the strongly correlated regime for $R > 1.7$~\AA, as expected.
We note that the shot budget for obtaining the entire dissociation curve in $0.1$ \AA\ increments is less than $3 \times 10^7$.
We think it is remarkable that so few shots are sufficient to obtain an energy estimator that includes dynamic correlation correction effects since, even when advanced methods such as regularized compressed double factorization\cite{oumarou2022accelerating} or fluid fermionic fragments\cite{choi2023} are employed, measuring just the bare active space VQE energy typically requires a comparable number of shots.
Measurements of 2RDMs, let alone sufficiently high quality higher-order RDMs, can be expected to require significantly more shots.  
Our estimate does not include error mitigation overheads, but is a very feasible shot count on present-day quantum hardware, as is testified by the fact that more shots were taken for individual echo-verified data points in Ref.~\citenum{obrien2023}.
For a complete comparison, a thorough investigation of the shot counts needed for PT-based methods for treating dynamic correlation using higher-order RDMs would be desirable, but this is beyond the scope of this work.
Additionally, it would be interesting to numerically analyze quantum state preparation techniques other than VQE in the context of TCCSD in future work.

\subsection{Resilience to Device Noise}
The combination of how overlaps are input into the split amplitude methods considered here with how they are determined in the matchgate shadow protocol\cite{wan2022matchgate} leads to a further noteworthy built-in error mitigation property of the resulting method.
To see this, we first need to recount how overlaps are measured from shadows.
When recording the shadow one picks a reference state such as the true vacuum $\ket{\mathbf{0}} = \ket{0}^{\otimes n}$, then one aims to prepare the state $\ket{\rho} = (\ket{\mathbf{0}} + \trial)/\sqrt{2}$ on the device and finally takes computational basis state measurements after random Gaussian rotations.
From the shadow obtained in this way the sought-after overlap with a computational basis state $\ket{\varphi}$ is then computed as the expectation value of the non-Hermitian observables $\ket{\mathbf{0}}\bra{\varphi}$, using that $\langle \varphi \trial = 2\, \mathrm{Tr}[|\mathbf{0} \rangle\langle \varphi| \rho]$.
If this state preparation is noisy so that instead of the state $\rho = \ket{\rho}\bra{\rho}$ the device prepares the state $\tilde\rho = (1-p)\,\rho + p \,\rho_{\text{noise}}$ for some error probability $p$ and density matrix $\rho_{\text{noise}}$, one has
\begin{equation}
 2 \,\mathrm{Tr}[|\mathbf{0} \rangle\langle \varphi| \tilde\rho] =  (1-p)\, \langle \varphi \trial + p \langle \varphi| \rho_{\text{noise}} |\mathbf{0} \rangle .
\end{equation}

For several reasonable types of noise, including depolarizing noise, bit flip noise, and amplitude damping $\langle \varphi| \rho_{\text{noise}} |\mathbf{0} \rangle$ will be either zero or very small for all states $|\varphi \rangle$ from a subspace containing a reasonable number of particles $\zeta$.
This means that overlaps such as those in eqs \eqref{eq:C1} and \eqref{eq:C2} can, because of the intermediate normalization of $\ket{\Psi^\mathrm{AS}} = \trial / \langle \Phi_0 \trial$ be very well approximated by quotients of overlaps computed from the shadow of the noisy state.
For the overlap in eq \eqref{eq:C1}, one has for example
\begin{equation}
c_i^a = \braket{\Phi_i^a | \Psi^\mathrm{AS}} = \frac{\langle \Phi_i^a \trial}{\langle \Phi_0 \trial} \approx \frac{\mathrm{Tr}[|\mathbf{0} \rangle\langle \Phi_i^a| \tilde\rho]}{\mathrm{Tr}[|\mathbf{0} \rangle\langle \Phi_0| \tilde\rho]},
\end{equation}
because the common factors of $2/(1-p)$ in the enumerator and denominator cancel.
Even estimation from a finite shot shadow should thus work well as long as $(1-p) \, \langle \varphi \trial$ does not become too small (see Appendix D.6 of Ref.~\citenum{huggins_unbiasing_2022} for similar consideration).
Contrary to this, without error mitigation techniques, expectation values such as that of, e.g., the electronic structure Hamiltonian, are usually first order sensitive to the noise strength $p$.

A further major concern in several platforms are particle number dependent phases that can be caused, e.g., by background magnetic fields.
These are often particularly problematic for algorithms that prepare cat-like coherent superpositions of states of markedly different particle number such as $\ket{\rho}$.
However, since we know that the input overlaps are all real and we are only interested in getting the relative signs of the coefficients, which are quotients of overlaps, correct and the computational basis state overlaps going into the enumerator all come from the same particle number sub-space, a particle number dependent phase results in a global phase affecting all coefficients, which can be easily corrected.
One possible way of doing this is to rotate the coefficients in the complex plane such that the largest coefficient or the principle component of all overlaps or is aligned with the real axis before either discarding the imaginary part or taking the absolute value multiplied with the sign of the real part of each coefficient. 
This renders our methods largely independent of uncontrolled particle number dependent phases and also implies some robustness against general phase errors.

\section{Split-Amplitude CC on Quantum Hardware} \label{sec:eccc_results_main}
In this section, we showcase how externally corrected CC performs on inputs obtained from actual quantum hardware. For this purpose, we revisit the
ground state energy of \ce{H4}/STO-3G arranged in a square geometry with bond distances of $1.23$ \AA, which was previously studied
with QC-QMC in Ref.~\citenum{huggins_unbiasing_2022}. Quantum overlaps were measured with Clifford shadows on Google's Sycamore superconducting quantum processor.\cite{arute2019quantum} 
In this hardware experiment, the state preparation for \ce{H4}/STO-3G on 8 qubits corresponds to preparing the exact FCI state. 
The overlaps, obtained from a Clifford shadow protocol,\cite{zhao_fermionic_2021} were then used to drive a QMC
calculation. Due to device and shot noise, the measured overlaps of course do not exactly reproduce FCI overlaps, making it an
interesting test case for studying noise robustness of ec-CC. For details on how the experimental shadow data were processed, see Appendix~\ref{apx:eccc_experiment}.
\begin{figure}[ht]
    \centering
    \includegraphics[width=0.5\textwidth]{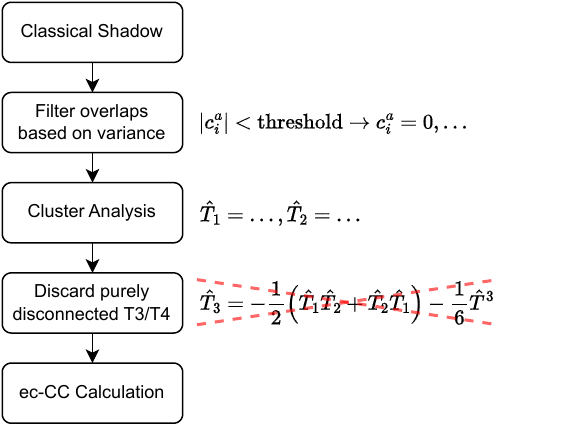}
    \caption{Post-processing protocol for quantum ec-CC.
    Overlaps measured in the classical (matchgate) shadow protocol are first filtered based on variance, and set to zero if the
    overlap estimator has a too large variance. After that, the cluster analysis is performed on the filtered overlaps, cluster amplitudes up to
    and including T4. Then, purely disconnected T3 and T4 amplitudes are disconnected, yielding a so-called ``Type-II'' ec-CC calculation
    in the end.\cite{magoulas2021is,lee2021externally}
    }
    \label{fig:ec_postprocessing}
\end{figure}
The quantum trial state prepared on the device (Q trial) first had to be converted to a wavefunction with real-valued coefficients (Q trial real, see Appendix \ref{apx:eccc_experiment})
before use with ec-CC. The Q trial state was recorded four times with different number of sampled Cliffords, $N_\mathrm{Cliffords}$, in the original
experiment, and we used the best two repetitions in terms of root mean-square error of the reconstructed wavefunction with respect to FCI (the results are qualitatively unchanged when all four runs are included, but the bands depicting uncertainty become unnecessarily wide).
Note that, for hardware reasons, the circuit for each group element was executed 1,000 times, yielding a so-called mulit-shot/thrifty\cite{helsen2022thrifty, zhou2023performanceanalysis} shadow.
To complement the hardware data, we simulated a matchgate shadow (30 repetitions) taken on the FCI wavefunction with both a single shot and 1,000 shots per group
element.
Furthermore, we devised a classical post-processing protocol, illustrated in Figure~\ref{fig:ec_postprocessing}, to ascertain certain properties of
the wavefunction extracted from the device. Therein, we filter overlaps based on their variance and set them to zero if a certain variance threshold
for the measurement is not fulfilled, i.e., the overlap was not determined reliably enough.
In the given experiment, we set overlaps to zero if the overlap value is larger than $2\, \sigma$.
After that, we perform the usual cluster analysis
which then might contain purely disconnected T3 and T4 due to the variance-based thresholding, or due to the underlying wavefunction. In the spirit
of ``Type-II'' ec-CC input,\cite{magoulas2021is,lee2021externally} we discard all purely disconnected T3 and T4 amplitudes and subsequently perform
the ec-CC calculation.
\begin{figure}[h]
\includegraphics[width=0.7\textwidth]{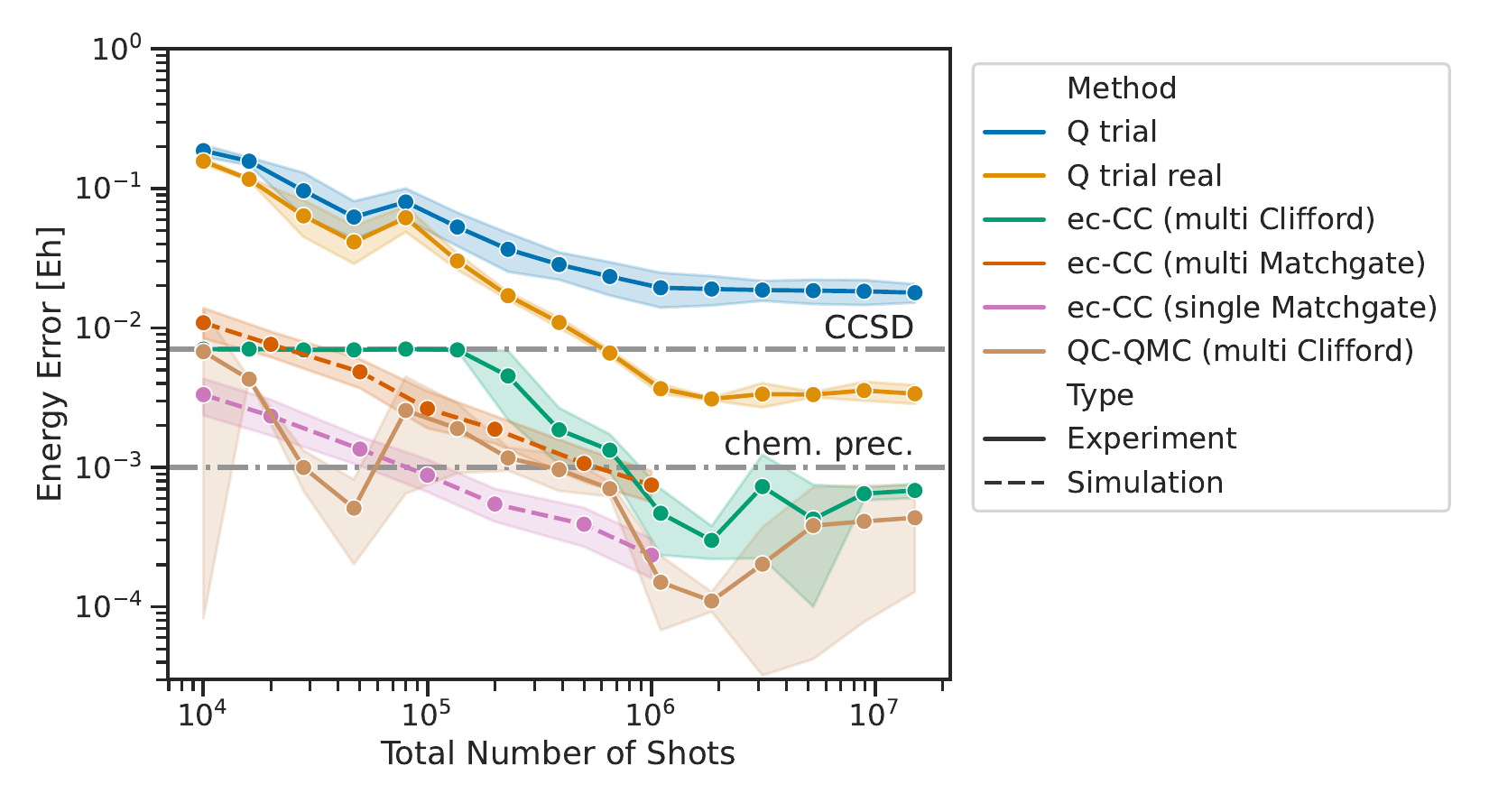}
\caption{
Energy error of trial state (Q trial), Q trial converted to real wavefunction (Q trial real), ec-CC with different shadow protocols, and QC-QMC with respect to FCI. Averages from two repetitions are shown, where the shaded areas indicate the 95th percentile.
Results based on experimental data from Ref.~\citenum{huggins_unbiasing_2022} are drawn a solid lines, results from simulated matchgate shadows with only shot noise are drwan as dashed lines. 
The CCSD energy error ($7.05$ mEh) and the limit for chemical precision (1 mEh) are shown for comparison.
Without circuit noise, the Q trial state would correspond to FCI.
Raw data are shown in Tables \ref{tab:h4_sto3g_fqe_real}, \ref{tab:h4_sto3g_ec}, and \ref{tab:h4_sto3g_ec_matchgate}.
} \label{fig:h4_sto3g_error}
\end{figure}
The absolute energy error with respect to FCI as a function of the total number of shots (i.e., number of group elements times number of circuit
executions per group element) are shown in Figure~\ref{fig:h4_sto3g_error}, including the results for Q trial and
QC-QMC from Ref.~\citenum{huggins_unbiasing_2022}.
For more than $10^6$ shots, the experiment is no more shot-noise-limited. In this limit, ec-CC and QC-QMC energy errors are less than 1 mEh, i.e.,
chemical precision is reached, whereas the plain variational energy of the trial state and its purely real
counterpart does not reach the same level of accuracy.
\begin{figure}[ht]
    \centering
    \includegraphics[width=0.45\textwidth]{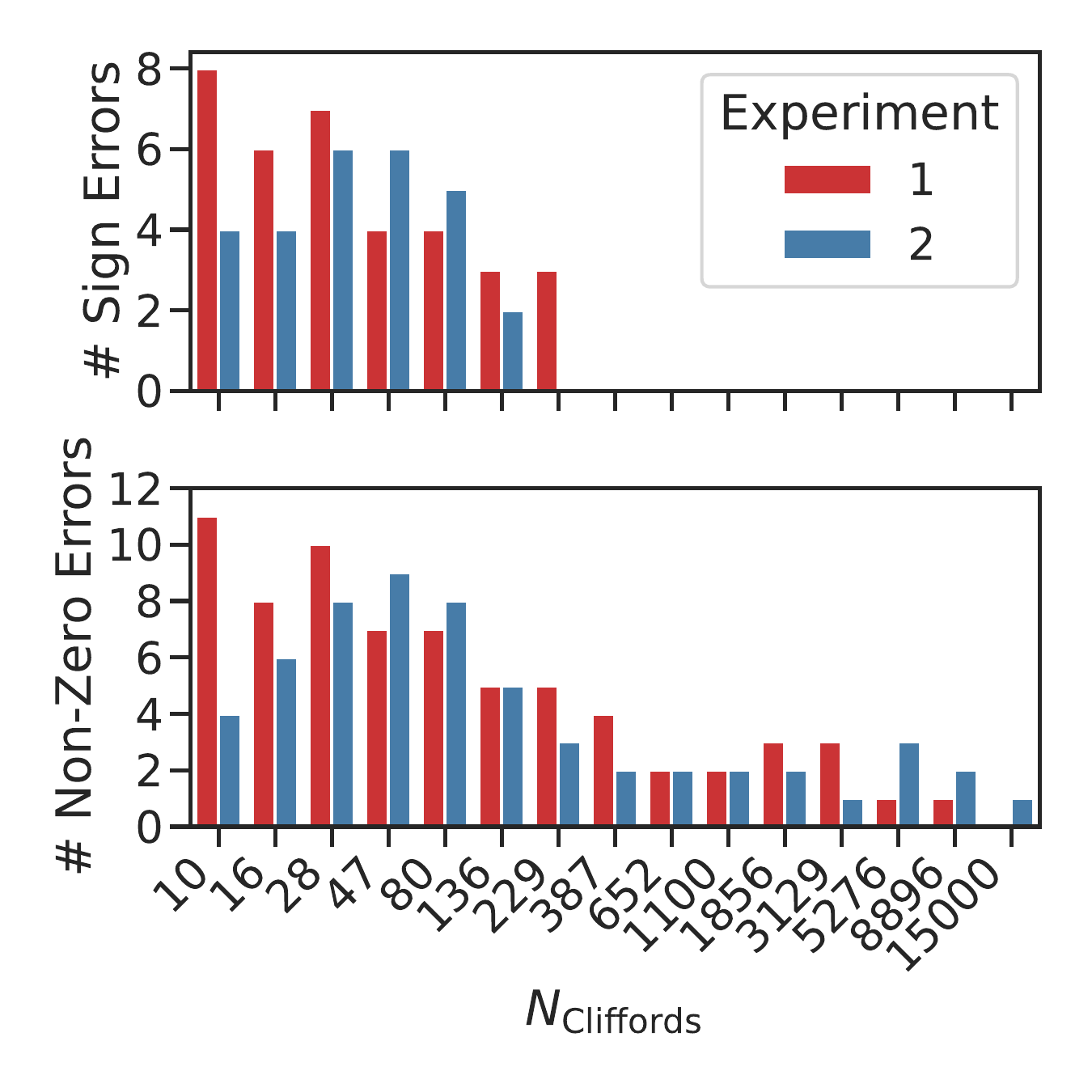}
    \caption{%
    Number of sign errors (upper panel) and non-zero errors (lower panel) in the
    overlap measurements depending on the number of Cliffords, $N_\mathrm{Cliffords}$,
    in the \ce{H4}/STO-3G experiment on 8 qubits, grouped by two runs of the experiment.
    Non-zero errors are defined as overlap values that should be numerically zero (in the FCI vector), but are larger than a given numerical threshold of $10^{-6}$ in the
    shadow-based measurements.
    }
    \label{fig:sign_nonzero_errors}
\end{figure}
Our post-processing protocol ensures that for too few shots (up to about $2 \times 10^5$) the results are never worse than plain CCSD.
To understand this behavior, we analyzed the type of errors that can occur in overlap measurements, namely \emph{sign} errors
and non-zero errors. The latter refer to an overlap value which should be numerically zero (in this case, the corresponding
value in the CI vector is zero), but is measured to be non-zero. The sign and non-zero errors for the two experimental runs
of the \ce{H4}/STO-3G experiment are summarized in Figure \ref{fig:sign_nonzero_errors}. Sign errors completely vanish at $N_\mathrm{Cliffords} = 387$,
which perfectly coincides with the points on the energy error curve where improvements over CCSD are observed. The non-zero errors do not decay as rapidly
as the sign errors, and some elements are still measured with a non-zero value for rather large
number of sampled group elements. Thus, sign errors seem to be most severe and affecting the quality of the
input overlaps for ec-CC, but those errors quickly disappear completely when the total number of shots is
moderate, i.e., chemical precision is only reached for even more shots.

Making the wavefunction real removes ambiguities because of the non-measurable global phase and is a necessity because coupled cluster
with a real molecular Hamiltonian requires real amplitude inputs. In addition, it seems to improve the quality of the trial wavefunction energy.
ec-CC on the experimental data does not seem to quite reach the accuracy of QC-QMC, but gets close and drastically improves over the (real) Q trial energy by a factor of $5$ and $25$, respectively, as well as about a factor of $10$ over the plain CCSD energy. It is reassuring that ec-CC is competitive
with QC-QMC in this setting since we are certain that, due to the wavefunction quality, ec-CC must improve over the plain trial state energy.
The comparison with the simulated matchgate shadow data shows that, with both shadow protocols, chemical precision can be reached with a comparable number of group elements and shots. The 1,000 shots multi-shot variant only needs roughly $10$ times more shots at $100$ fewer distinct circuits, making it attractive on hardware where changing the circuit incurs a run time overhead.
We attribute the generally better performance on simulated data to the absence of gate and detection noise.

\begin{figure}[ht]
    \centering
    \includegraphics[width=0.7\textwidth]{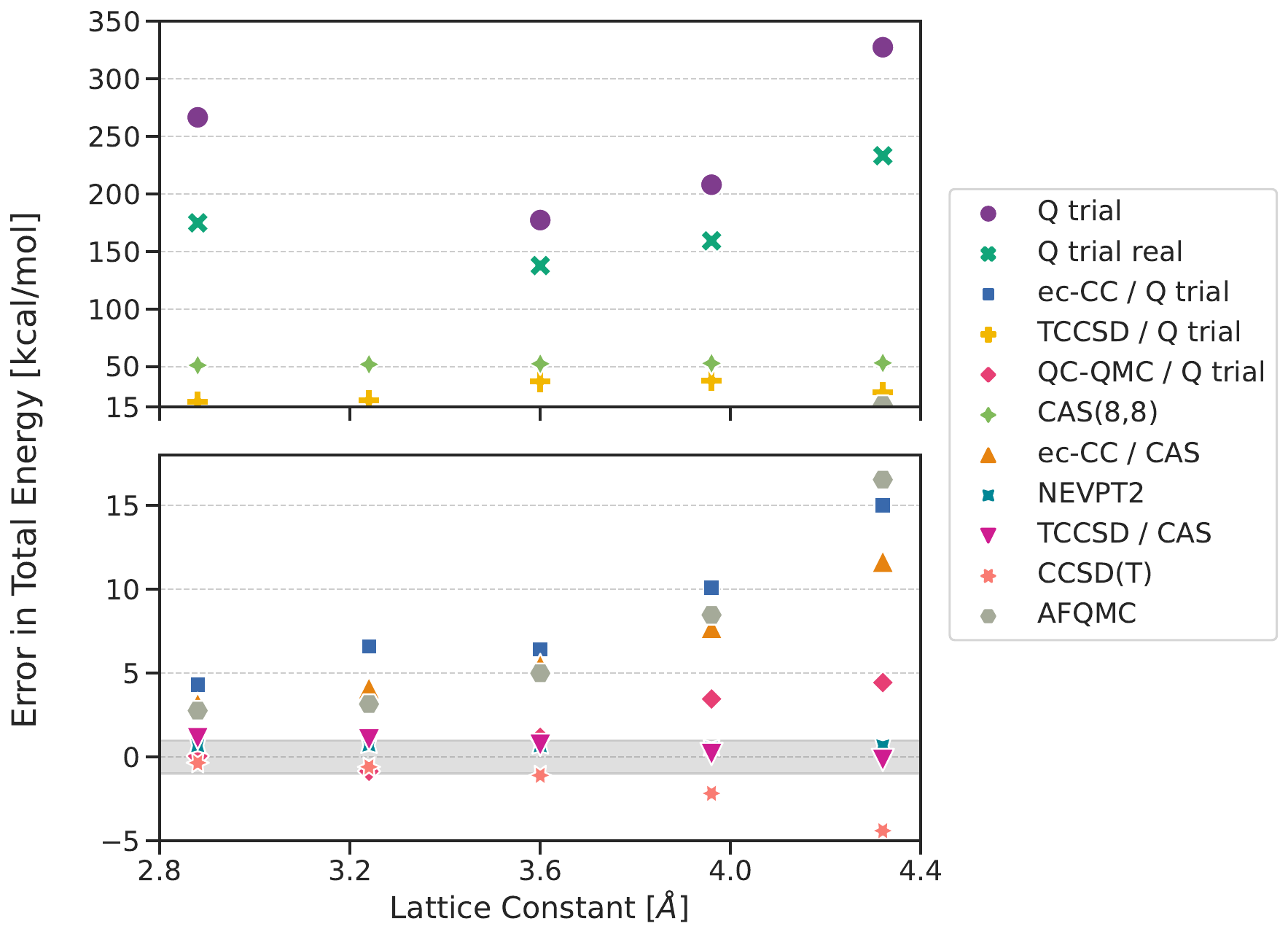}
    \caption{Error in total energy as a function of the lattice constant of a minimal cell diamond (DZVP-GTH basis with 26
    orbitals). In contrast to the \ce{H4} case, the quantum trial state is prepared on an active space of 16 qubits using a perfect-pairing wavefunction,
    as explained in detail in Ref.~\citenum{huggins_unbiasing_2022}.
    Raw data are shown in Table \ref{tab:diamond}.
    The plot is split up in two panels with different vertical axis limits to show the whole range of energy errors.
    The gray area in the lower panel indicates the bounds for chemical precision of 1 kcal/mol.
    Data for Q trial, AFQMC, and QC-QMC are reproduced from Ref.~\citenum{huggins_unbiasing_2022}.
    The ``/'' notation in the plot legend indicates the input wavefunction for the respective method, i.e., `Q trial` from the device or a classically simulated exact CAS(8,8).
    Note that for the lattice constant of $3.24$~\AA, the results for (real) Q trial are not visible on the depicted energy scale.
    For high values of the lattice constant, only TCCSD based on CAS(8,8) and NEVPT2 achieve chemical precision,
    however, all quantum split-amplitude methods run on the hardware data significantly improve the energy of the prepared trial state.
    }
    \label{fig:diamond}
\end{figure}
As a second example run on actual hardware, we study the energy of a minimal diamond unit cell (two carbon atoms)
in a double-zeta basis (GTH-DZVP;\cite{vandevondele2007gaussian} 26 orbitals, with GTH-PADE\cite{goedecker1996separable} pseudopotential, only $\Gamma$ point in Brillouin zone sampling)
as a function of the lattice constant. In Ref.~\citenum{huggins_unbiasing_2022}, the quantum trial corresponded to a 16-qubit
active space perfect-pairing (PP) wavefunction, and the computational basis state overlaps were obtained using the Clifford shadow
protocol, sampling 50,000 group elements for each lattice constant.
In addition to the methods used in Ref.~\citenum{huggins_unbiasing_2022}, we computed the single-point energy for each lattice constant
using ec-CC and TCCSD on the real-valued quantum trial state, CAS(8,8) (i.e., the exact solution to the given active space problem),
ec-CC and TCCSD on the exact CAS state, and NEVPT2 for comparison.
The results are shown in Figure \ref{fig:diamond}.
In comparison with the \ce{H4} hardware experiment, two factors limit the performance guarantees of ec-CC: 1) the fact
that we are dealing with an active space wavefunction, such that we do not have any convergence guarantees to the true FCI result in the limit
of exact external T3 and T4 amplitudes, and 2) the quantum trial PP state is a CC-like wavefunction, which does not accurately treat higher-order
excitations as would be beneficial for ec-CC.
The results for quantum ec-CC on the quantum trial state never reach chemical precision, and the error is between approximately 5 and 15 kcal/mol.
This is comparable with the results from plain AFQMC. Even when the exact CAS(8,8) wavefunction is used to construct the T3 and T4 amplitudes,
the energy error only marginally improves over the quantum input one. This strongly hints that the ``truncation'' of T3 and T4 stemming from an active space
is of course preventing more accurate ec-CC results.
The performance of quantum TCCSD tailored with the quantum trial state is even worse, since the energy error is larger than $15$~kcal/mol for all values of
the lattice constant. Still, for both split-amplitude methods executed with actual quantum trial states, the energy error is vastly better than the
variational energy of the quantum trial state (for the value at $3.24$~\AA, it is literally off the charts).
Interestingly, the energy error of a plain CAS(8,8), due to neglect of dynamic correlation effects, is worse than that of the quantum split-amplitude methods.
The only methods in our setup that reach chemical precision for all values of the lattice constant are TCCSD tailored with CAS(8,8) and NEVPT2.
The energy error observed for those two methods is almost identical. TCCSD thus provides vastly better results when tailored with the exact active
space wavefunction than with the PP quantum trial state. This is expected, since the CC-like PP wavefunction cannot yield a better result in this case than
plain CCSD on the full orbital space. Even though there are no strict theoretical requirements for input wavefunctions to TCCSD, this is a case which
is expected to fail based on the properties of the quantum trial. The fact that ec-CC uses a post-processing protocol for the input amplitudes could
explain why it performs better than TCCSD on actual hardware data.
From the perspective of classical electronic structure methods, the performance of TCCSD with CAS is on par with NEVPT2, which is encouraging.
We conclude that ec-CC of course was not expected to work well, or even better than QC-QMC, in the active space and PP quantum input setting, which
was corroborated by our numerical results. Nonetheless, the split-amplitude methods provide a massive improvement in absolute energy error even when
seeded with noisy quantum input.

\section{Conclusions}
In this work we proposed to use overlaps obtained from a quantum computer by means of classical shadows as inputs to split-amplitude CC methods tailored coupled cluster (TCC) and externally corrected coupled cluster (ec-CC).
The resulting combinations of methods can be seen as a way of adding dynamic electron correlation corrections onto active space trial states prepared on a quantum computer.
These methods can further be viewed as a way of curing failures of plain single-reference coupled cluster theory (such as the appearance of virtual reaction barriers)
by taking into account properties of a multi-reference wavefunction prepared on a quantum computer, while avoiding the dramatic increase in classical computational complexity of classical multi-reference coupled cluster methods.
We showed that the combination of methods displays a range of desirable properties:
1) The dynamic correlation correction along the \ce{N2} dissociation curve and for a stretched diamond cell is found to be comparable in quality to NEVPT2 and remarkably robust against systematic imperfections in the prepared active space wavefunction, as may arise from too shallow VQE circuits.
2) One can predict the number of repetitions (shots) required to obtain accurate results from established correlation diagnostics that are classically efficiently obtainable from CCSD calculations.
3) Extrapolating these resource estimates to the classically no longer exactly solvable regime, we found remarkably low shot counts, which we attribute to the fact that measurement of expensive intermediates such as higher-order reduced density matrices is avoided.
4) When tested with overlaps measured on real quantum hardware, our method provided results with an accuracy comparable to QC-QMC for the ground state energy of \ce{H4} in a minimal basis.
5) The expensive classical part of the computation can be performed with standard coupled cluster codes, opening the possibility for further speedups with, e.g., DLPNO.
6) In particular ec-CC (and to a lesser degree TCC) seems to have some built-in error mitigation abilities, producing energies that can be as good as CASCI or AFQMC, respectively, even from very noisy trial state amplitudes. 
We analyzed the statistical properties of overlaps computed from matchgate shadows\cite{wan2022matchgate} and numerically corroborated the validity of a Gaussian noise model that may be of general interest.
Furthermore, the quantum split-amplitude methods are fully agnostic of the measurement scheme. Hence, the quantum resource estimates provided in our present work
can most likely be improved upon using shadow protocols with lower sample complexity.
It would thus be interesting to see if similar noise models hold for other classical shadow protocols.
In particular, the particle-number-preserving shadow protocol from Ref.~\citenum{low_classical_2022} promises to make the required number of shots to obtain overlaps at a given precision independent of the number of qubits at the expense of an overhead in circuit depth and number of qubits. It will be important to explore these trade-offs further.
Another possible direction is to combine the methods proposed here with shadow-based error mitigation methods.\cite{jnane2023quantum,chan2023algorithmic,brieger2023stability}
While the shot counts we find for the quantum split-amplitude CC methods seem very low, even when comparing to typical shot counts needed to just compute the energy of bare electronic structure Hamiltonians on quantum wavefunctions, a detailed analysis of the shot counts needed to obtain similar quality energies via perturbative methods such as NEVPT2 from higher-order RDMs measured either from shadows or directly remains outstanding.
For ec-CC, we were able to give some hints for what kind of wavefunctions could yield a quantum advantage and hope that this can inspire future work on VQE ansätze and other state preparation schemes.
Finally, it would be interesting to extend TCCSD with T3 amplitudes from the active space as suggested in the original TCC work,\cite{kinoshita2005coupled}
and to compute other molecular properties with split-amplitude methods.

\section*{Acknowledgements}
We acknowledge fruitful discussions with William J.~Huggins, Eugene DePrince III, Fotios Gkritsis, Robert M.~Parrish, and Pauline J.~Ollitrault. 

\section*{Data Availability}
The data for matchgate shadow statistics (Section \ref{sec:influence_of_noise}) and the raw data of noisy TCCSD
energies to devise the power law error model and shot budget estimation, including molecular
geometries and correlation diagnostics, have been uploaded to Zenodo under the DOI \href{https://doi.org/10.5281/zenodo.10470740}{10.5281/zenodo.10470740}.\cite{zenododata}

\bibliography{reference,reference2}

\begin{thebibliography}{135}%
\makeatletter
\providecommand \@ifxundefined [1]{%
 \@ifx{#1\undefined}
}%
\providecommand \@ifnum [1]{%
 \ifnum #1\expandafter \@firstoftwo
 \else \expandafter \@secondoftwo
 \fi
}%
\providecommand \@ifx [1]{%
 \ifx #1\expandafter \@firstoftwo
 \else \expandafter \@secondoftwo
 \fi
}%
\providecommand \natexlab [1]{#1}%
\providecommand \enquote  [1]{``#1''}%
\providecommand \bibnamefont  [1]{#1}%
\providecommand \bibfnamefont [1]{#1}%
\providecommand \citenamefont [1]{#1}%
\providecommand \href@noop [0]{\@secondoftwo}%
\providecommand \href [0]{\begingroup \@sanitize@url \@href}%
\providecommand \@href[1]{\@@startlink{#1}\@@href}%
\providecommand \@@href[1]{\endgroup#1\@@endlink}%
\providecommand \@sanitize@url [0]{\catcode `\\12\catcode `\$12\catcode
  `\&12\catcode `\#12\catcode `\^12\catcode `\_12\catcode `\%12\relax}%
\providecommand \@@startlink[1]{}%
\providecommand \@@endlink[0]{}%
\providecommand \url  [0]{\begingroup\@sanitize@url \@url }%
\providecommand \@url [1]{\endgroup\@href {#1}{\urlprefix }}%
\providecommand \urlprefix  [0]{URL }%
\providecommand \Eprint [0]{\href }%
\providecommand \doibase [0]{https://doi.org/}%
\providecommand \selectlanguage [0]{\@gobble}%
\providecommand \bibinfo  [0]{\@secondoftwo}%
\providecommand \bibfield  [0]{\@secondoftwo}%
\providecommand \translation [1]{[#1]}%
\providecommand \BibitemOpen [0]{}%
\providecommand \bibitemStop [0]{}%
\providecommand \bibitemNoStop [0]{.\EOS\space}%
\providecommand \EOS [0]{\spacefactor3000\relax}%
\providecommand \BibitemShut  [1]{\csname bibitem#1\endcsname}%
\let\auto@bib@innerbib\@empty
\bibitem [{\citenamefont {McArdle}\ \emph {et~al.}(2020)\citenamefont
  {McArdle}, \citenamefont {Endo}, \citenamefont {Aspuru-Guzik}, \citenamefont
  {Benjamin},\ and\ \citenamefont {Yuan}}]{mcardle2020quantum}%
  \BibitemOpen
  \bibfield  {author} {\bibinfo {author} {\bibfnamefont {S.}~\bibnamefont
  {McArdle}}, \bibinfo {author} {\bibfnamefont {S.}~\bibnamefont {Endo}},
  \bibinfo {author} {\bibfnamefont {A.}~\bibnamefont {Aspuru-Guzik}}, \bibinfo
  {author} {\bibfnamefont {S.~C.}\ \bibnamefont {Benjamin}},\ and\ \bibinfo
  {author} {\bibfnamefont {X.}~\bibnamefont {Yuan}},\ }\bibfield  {title}
  {\bibinfo {title} {Quantum computational chemistry},\ }\href
  {https://doi.org/10.1103/REVMODPHYS.92.015003} {\bibfield  {journal}
  {\bibinfo  {journal} {Reviews of Modern Physics}\ }\textbf {\bibinfo {volume}
  {92}},\ \bibinfo {pages} {015003} (\bibinfo {year} {2020})}\BibitemShut
  {NoStop}%
\bibitem [{\citenamefont {Kinoshita}\ \emph {et~al.}(2005)\citenamefont
  {Kinoshita}, \citenamefont {Hino},\ and\ \citenamefont
  {Bartlett}}]{kinoshita2005coupled}%
  \BibitemOpen
  \bibfield  {author} {\bibinfo {author} {\bibfnamefont {T.}~\bibnamefont
  {Kinoshita}}, \bibinfo {author} {\bibfnamefont {O.}~\bibnamefont {Hino}},\
  and\ \bibinfo {author} {\bibfnamefont {R.~J.}\ \bibnamefont {Bartlett}},\
  }\bibfield  {title} {\bibinfo {title} {Coupled-cluster method tailored by
  configuration interaction},\ }\href {https://doi.org/10.1063/1.2000251}
  {\bibfield  {journal} {\bibinfo  {journal} {The Journal of chemical physics}\
  }\textbf {\bibinfo {volume} {123}},\ \bibinfo {pages} {074106} (\bibinfo
  {year} {2005})}\BibitemShut {NoStop}%
\bibitem [{\citenamefont {Paldus}(2017)}]{paldus2017externally}%
  \BibitemOpen
  \bibfield  {author} {\bibinfo {author} {\bibfnamefont {J.}~\bibnamefont
  {Paldus}},\ }\bibfield  {title} {\bibinfo {title} {Externally and internally
  corrected coupled cluster approaches: an overview},\ }\href
  {https://doi.org/10.1007/s10910-016-0688-6} {\bibfield  {journal} {\bibinfo
  {journal} {Journal of Mathematical Chemistry}\ }\textbf {\bibinfo {volume}
  {55}},\ \bibinfo {pages} {477} (\bibinfo {year} {2017})}\BibitemShut
  {NoStop}%
\bibitem [{\citenamefont {Takeshita}\ \emph {et~al.}(2020)\citenamefont
  {Takeshita}, \citenamefont {Rubin}, \citenamefont {Jiang}, \citenamefont
  {Lee}, \citenamefont {Babbush},\ and\ \citenamefont
  {McClean}}]{takeshita2020increasing}%
  \BibitemOpen
  \bibfield  {author} {\bibinfo {author} {\bibfnamefont {T.}~\bibnamefont
  {Takeshita}}, \bibinfo {author} {\bibfnamefont {N.~C.}\ \bibnamefont
  {Rubin}}, \bibinfo {author} {\bibfnamefont {Z.}~\bibnamefont {Jiang}},
  \bibinfo {author} {\bibfnamefont {E.}~\bibnamefont {Lee}}, \bibinfo {author}
  {\bibfnamefont {R.}~\bibnamefont {Babbush}},\ and\ \bibinfo {author}
  {\bibfnamefont {J.~R.}\ \bibnamefont {McClean}},\ }\bibfield  {title}
  {\bibinfo {title} {Increasing the representation accuracy of quantum
  simulations of chemistry without extra quantum resources},\ }\bibfield
  {journal} {\bibinfo  {journal} {Physical Review X}\ }\textbf {\bibinfo
  {volume} {10}},\ \href {https://doi.org/10.1103/physrevx.10.011004}
  {10.1103/physrevx.10.011004} (\bibinfo {year} {2020})\BibitemShut {NoStop}%
\bibitem [{\citenamefont {Peruzzo}\ \emph {et~al.}(2014)\citenamefont
  {Peruzzo}, \citenamefont {McClean}, \citenamefont {Shadbolt}, \citenamefont
  {Yung}, \citenamefont {Zhou}, \citenamefont {Love}, \citenamefont
  {Aspuru-Guzik},\ and\ \citenamefont {O’Brien}}]{peruzzo2014variational}%
  \BibitemOpen
  \bibfield  {author} {\bibinfo {author} {\bibfnamefont {A.}~\bibnamefont
  {Peruzzo}}, \bibinfo {author} {\bibfnamefont {J.}~\bibnamefont {McClean}},
  \bibinfo {author} {\bibfnamefont {P.}~\bibnamefont {Shadbolt}}, \bibinfo
  {author} {\bibfnamefont {M.-H.}\ \bibnamefont {Yung}}, \bibinfo {author}
  {\bibfnamefont {X.-Q.}\ \bibnamefont {Zhou}}, \bibinfo {author}
  {\bibfnamefont {P.~J.}\ \bibnamefont {Love}}, \bibinfo {author}
  {\bibfnamefont {A.}~\bibnamefont {Aspuru-Guzik}},\ and\ \bibinfo {author}
  {\bibfnamefont {J.~L.}\ \bibnamefont {O’Brien}},\ }\bibfield  {title}
  {\bibinfo {title} {A variational eigenvalue solver on a photonic quantum
  processor},\ }\bibfield  {journal} {\bibinfo  {journal} {Nature
  Communications}\ }\textbf {\bibinfo {volume} {5}},\ \href
  {https://doi.org/10.1038/ncomms5213} {10.1038/ncomms5213} (\bibinfo {year}
  {2014})\BibitemShut {NoStop}%
\bibitem [{\citenamefont {McClean}\ \emph {et~al.}(2016)\citenamefont
  {McClean}, \citenamefont {Romero}, \citenamefont {Babbush},\ and\
  \citenamefont {Aspuru-Guzik}}]{mcclean2016theory}%
  \BibitemOpen
  \bibfield  {author} {\bibinfo {author} {\bibfnamefont {J.~R.}\ \bibnamefont
  {McClean}}, \bibinfo {author} {\bibfnamefont {J.}~\bibnamefont {Romero}},
  \bibinfo {author} {\bibfnamefont {R.}~\bibnamefont {Babbush}},\ and\ \bibinfo
  {author} {\bibfnamefont {A.}~\bibnamefont {Aspuru-Guzik}},\ }\bibfield
  {title} {\bibinfo {title} {The theory of variational hybrid quantum-classical
  algorithms},\ }\href {https://doi.org/10.1088/1367-2630/18/2/023023}
  {\bibfield  {journal} {\bibinfo  {journal} {New Journal of Physics}\ }\textbf
  {\bibinfo {volume} {18}},\ \bibinfo {pages} {023023} (\bibinfo {year}
  {2016})}\BibitemShut {NoStop}%
\bibitem [{\citenamefont {Tammaro}\ \emph {et~al.}(2023)\citenamefont
  {Tammaro}, \citenamefont {Galli}, \citenamefont {Rice},\ and\ \citenamefont
  {Motta}}]{tammaro2023n}%
  \BibitemOpen
  \bibfield  {author} {\bibinfo {author} {\bibfnamefont {A.}~\bibnamefont
  {Tammaro}}, \bibinfo {author} {\bibfnamefont {D.~E.}\ \bibnamefont {Galli}},
  \bibinfo {author} {\bibfnamefont {J.~E.}\ \bibnamefont {Rice}},\ and\
  \bibinfo {author} {\bibfnamefont {M.}~\bibnamefont {Motta}},\ }\bibfield
  {title} {\bibinfo {title} {N-electron valence perturbation theory with
  reference wave functions from quantum computing: Application to the relative
  stability of hydroxide anion and hydroxyl radical},\ }\href
  {https://doi.org/10.1021/acs.jpca.2c07653} {\bibfield  {journal} {\bibinfo
  {journal} {The Journal of Physical Chemistry A}\ }\textbf {\bibinfo {volume}
  {127}},\ \bibinfo {pages} {817–827} (\bibinfo {year} {2023})}\BibitemShut
  {NoStop}%
\bibitem [{\citenamefont {Baek}\ \emph {et~al.}(2023)\citenamefont {Baek},
  \citenamefont {Hait}, \citenamefont {Shee}, \citenamefont {Leimkuhler},
  \citenamefont {Huggins}, \citenamefont {Stetina}, \citenamefont
  {Head-Gordon},\ and\ \citenamefont {Whaley}}]{baek2023say}%
  \BibitemOpen
  \bibfield  {author} {\bibinfo {author} {\bibfnamefont {U.}~\bibnamefont
  {Baek}}, \bibinfo {author} {\bibfnamefont {D.}~\bibnamefont {Hait}}, \bibinfo
  {author} {\bibfnamefont {J.}~\bibnamefont {Shee}}, \bibinfo {author}
  {\bibfnamefont {O.}~\bibnamefont {Leimkuhler}}, \bibinfo {author}
  {\bibfnamefont {W.~J.}\ \bibnamefont {Huggins}}, \bibinfo {author}
  {\bibfnamefont {T.~F.}\ \bibnamefont {Stetina}}, \bibinfo {author}
  {\bibfnamefont {M.}~\bibnamefont {Head-Gordon}},\ and\ \bibinfo {author}
  {\bibfnamefont {K.~B.}\ \bibnamefont {Whaley}},\ }\bibfield  {title}
  {\bibinfo {title} {Say no to optimization: A nonorthogonal quantum
  eigensolver},\ }\href {https://doi.org/10.1103/PRXQuantum.4.030307}
  {\bibfield  {journal} {\bibinfo  {journal} {PRX Quantum}\ }\textbf {\bibinfo
  {volume} {4}},\ \bibinfo {pages} {030307} (\bibinfo {year}
  {2023})}\BibitemShut {NoStop}%
\bibitem [{\citenamefont {Krompiec}\ and\ \citenamefont
  {Ramo}(2022)}]{krompiec2022strongly}%
  \BibitemOpen
  \bibfield  {author} {\bibinfo {author} {\bibfnamefont {M.}~\bibnamefont
  {Krompiec}}\ and\ \bibinfo {author} {\bibfnamefont {D.~M.}\ \bibnamefont
  {Ramo}},\ }\bibfield  {title} {\bibinfo {title} {Strongly contracted
  n-electron valence state perturbation theory using reduced density matrices
  from a quantum computer},\ }\href@noop {} {\bibfield  {journal} {\bibinfo
  {journal} {arXiv preprint arXiv:2210.05702}\ } (\bibinfo {year}
  {2022})}\BibitemShut {NoStop}%
\bibitem [{\citenamefont {Kawashima}\ \emph {et~al.}(2021)\citenamefont
  {Kawashima}, \citenamefont {Lloyd}, \citenamefont {Coons}, \citenamefont
  {Nam}, \citenamefont {Matsuura}, \citenamefont {Garza}, \citenamefont
  {Johri}, \citenamefont {Huntington}, \citenamefont {Senicourt}, \citenamefont
  {Maksymov} \emph {et~al.}}]{kawashima2021optimizing}%
  \BibitemOpen
  \bibfield  {author} {\bibinfo {author} {\bibfnamefont {Y.}~\bibnamefont
  {Kawashima}}, \bibinfo {author} {\bibfnamefont {E.}~\bibnamefont {Lloyd}},
  \bibinfo {author} {\bibfnamefont {M.~P.}\ \bibnamefont {Coons}}, \bibinfo
  {author} {\bibfnamefont {Y.}~\bibnamefont {Nam}}, \bibinfo {author}
  {\bibfnamefont {S.}~\bibnamefont {Matsuura}}, \bibinfo {author}
  {\bibfnamefont {A.~J.}\ \bibnamefont {Garza}}, \bibinfo {author}
  {\bibfnamefont {S.}~\bibnamefont {Johri}}, \bibinfo {author} {\bibfnamefont
  {L.}~\bibnamefont {Huntington}}, \bibinfo {author} {\bibfnamefont
  {V.}~\bibnamefont {Senicourt}}, \bibinfo {author} {\bibfnamefont {A.~O.}\
  \bibnamefont {Maksymov}}, \emph {et~al.},\ }\bibfield  {title} {\bibinfo
  {title} {Optimizing electronic structure simulations on a trapped-ion quantum
  computer using problem decomposition},\ }\href
  {https://doi.org/10.1038/s42005-021-00751-9} {\bibfield  {journal} {\bibinfo
  {journal} {Communications Physics}\ }\textbf {\bibinfo {volume} {4}},\
  \bibinfo {pages} {245} (\bibinfo {year} {2021})}\BibitemShut {NoStop}%
\bibitem [{\citenamefont {Rossmannek}\ \emph {et~al.}(2021)\citenamefont
  {Rossmannek}, \citenamefont {Barkoutsos}, \citenamefont {Ollitrault},\ and\
  \citenamefont {Tavernelli}}]{rossmannek2021quantum}%
  \BibitemOpen
  \bibfield  {author} {\bibinfo {author} {\bibfnamefont {M.}~\bibnamefont
  {Rossmannek}}, \bibinfo {author} {\bibfnamefont {P.~K.}\ \bibnamefont
  {Barkoutsos}}, \bibinfo {author} {\bibfnamefont {P.~J.}\ \bibnamefont
  {Ollitrault}},\ and\ \bibinfo {author} {\bibfnamefont {I.}~\bibnamefont
  {Tavernelli}},\ }\bibfield  {title} {\bibinfo {title} {Quantum
  hf/dft-embedding algorithms for electronic structure calculations: Scaling up
  to complex molecular systems},\ }\bibfield  {journal} {\bibinfo  {journal}
  {The Journal of Chemical Physics}\ }\textbf {\bibinfo {volume} {154}},\ \href
  {https://doi.org/doi.org/10.1063/5.0029536} {doi.org/10.1063/5.0029536}
  (\bibinfo {year} {2021})\BibitemShut {NoStop}%
\bibitem [{\citenamefont {Li}\ \emph {et~al.}(2022)\citenamefont {Li},
  \citenamefont {Huang}, \citenamefont {Cao}, \citenamefont {Huang},
  \citenamefont {Shuai}, \citenamefont {Sun}, \citenamefont {Sun},
  \citenamefont {Yuan},\ and\ \citenamefont {Lv}}]{li2022toward}%
  \BibitemOpen
  \bibfield  {author} {\bibinfo {author} {\bibfnamefont {W.}~\bibnamefont
  {Li}}, \bibinfo {author} {\bibfnamefont {Z.}~\bibnamefont {Huang}}, \bibinfo
  {author} {\bibfnamefont {C.}~\bibnamefont {Cao}}, \bibinfo {author}
  {\bibfnamefont {Y.}~\bibnamefont {Huang}}, \bibinfo {author} {\bibfnamefont
  {Z.}~\bibnamefont {Shuai}}, \bibinfo {author} {\bibfnamefont
  {X.}~\bibnamefont {Sun}}, \bibinfo {author} {\bibfnamefont {J.}~\bibnamefont
  {Sun}}, \bibinfo {author} {\bibfnamefont {X.}~\bibnamefont {Yuan}},\ and\
  \bibinfo {author} {\bibfnamefont {D.}~\bibnamefont {Lv}},\ }\bibfield
  {title} {\bibinfo {title} {Toward practical quantum embedding simulation of
  realistic chemical systems on near-term quantum computers},\ }\href
  {https://doi.org/10.1039/D2SC01492K} {\bibfield  {journal} {\bibinfo
  {journal} {Chemical science}\ }\textbf {\bibinfo {volume} {13}},\ \bibinfo
  {pages} {8953} (\bibinfo {year} {2022})}\BibitemShut {NoStop}%
\bibitem [{\citenamefont {Liu}\ \emph {et~al.}(2023)\citenamefont {Liu},
  \citenamefont {Meitei}, \citenamefont {Chin}, \citenamefont {Dutt},
  \citenamefont {Tao}, \citenamefont {Chuang},\ and\ \citenamefont
  {Van~Voorhis}}]{liu2023bootstrap}%
  \BibitemOpen
  \bibfield  {author} {\bibinfo {author} {\bibfnamefont {Y.}~\bibnamefont
  {Liu}}, \bibinfo {author} {\bibfnamefont {O.~R.}\ \bibnamefont {Meitei}},
  \bibinfo {author} {\bibfnamefont {Z.~E.}\ \bibnamefont {Chin}}, \bibinfo
  {author} {\bibfnamefont {A.}~\bibnamefont {Dutt}}, \bibinfo {author}
  {\bibfnamefont {M.}~\bibnamefont {Tao}}, \bibinfo {author} {\bibfnamefont
  {I.~L.}\ \bibnamefont {Chuang}},\ and\ \bibinfo {author} {\bibfnamefont
  {T.}~\bibnamefont {Van~Voorhis}},\ }\bibfield  {title} {\bibinfo {title}
  {Bootstrap embedding on a quantum computer},\ }\href
  {https://doi.org/10.1021/acs.jctc.3c00012} {\bibfield  {journal} {\bibinfo
  {journal} {Journal of Chemical Theory and Computation}\ }\textbf {\bibinfo
  {volume} {19}},\ \bibinfo {pages} {2230–2247} (\bibinfo {year}
  {2023})}\BibitemShut {NoStop}%
\bibitem [{\citenamefont {He}\ and\ \citenamefont
  {Evangelista}(2020)}]{he2020zeroth}%
  \BibitemOpen
  \bibfield  {author} {\bibinfo {author} {\bibfnamefont {N.}~\bibnamefont
  {He}}\ and\ \bibinfo {author} {\bibfnamefont {F.~A.}\ \bibnamefont
  {Evangelista}},\ }\bibfield  {title} {\bibinfo {title} {A zeroth-order
  active-space frozen-orbital embedding scheme for multireference
  calculations},\ }\bibfield  {journal} {\bibinfo  {journal} {The Journal of
  Chemical Physics}\ }\textbf {\bibinfo {volume} {152}},\ \href
  {https://doi.org/10.1063/1.5142481} {10.1063/1.5142481} (\bibinfo {year}
  {2020})\BibitemShut {NoStop}%
\bibitem [{\citenamefont {He}\ \emph {et~al.}(2022{\natexlab{a}})\citenamefont
  {He}, \citenamefont {Li},\ and\ \citenamefont {Evangelista}}]{he2022second}%
  \BibitemOpen
  \bibfield  {author} {\bibinfo {author} {\bibfnamefont {N.}~\bibnamefont
  {He}}, \bibinfo {author} {\bibfnamefont {C.}~\bibnamefont {Li}},\ and\
  \bibinfo {author} {\bibfnamefont {F.~A.}\ \bibnamefont {Evangelista}},\
  }\bibfield  {title} {\bibinfo {title} {Second-order active-space embedding
  theory},\ }\href {https://doi.org/10.1021/acs.jctc.1c01099} {\bibfield
  {journal} {\bibinfo  {journal} {Journal of Chemical Theory and Computation}\
  }\textbf {\bibinfo {volume} {18}},\ \bibinfo {pages} {1527} (\bibinfo {year}
  {2022}{\natexlab{a}})}\BibitemShut {NoStop}%
\bibitem [{\citenamefont {Evangelista}(2018)}]{evangelista2018perspective}%
  \BibitemOpen
  \bibfield  {author} {\bibinfo {author} {\bibfnamefont {F.~A.}\ \bibnamefont
  {Evangelista}},\ }\bibfield  {title} {\bibinfo {title} {Perspective:
  Multireference coupled cluster theories of dynamical electron correlation},\
  }\bibfield  {journal} {\bibinfo  {journal} {Journal of Chemical Physics}\
  }\textbf {\bibinfo {volume} {149}},\ \href
  {https://doi.org/10.1063/1.5039496} {10.1063/1.5039496} (\bibinfo {year}
  {2018})\BibitemShut {NoStop}%
\bibitem [{\citenamefont {Huang}\ \emph {et~al.}(2023)\citenamefont {Huang},
  \citenamefont {Li},\ and\ \citenamefont {Evangelista}}]{PRXQuantum.4.020313}%
  \BibitemOpen
  \bibfield  {author} {\bibinfo {author} {\bibfnamefont {R.}~\bibnamefont
  {Huang}}, \bibinfo {author} {\bibfnamefont {C.}~\bibnamefont {Li}},\ and\
  \bibinfo {author} {\bibfnamefont {F.~A.}\ \bibnamefont {Evangelista}},\
  }\bibfield  {title} {\bibinfo {title} {Leveraging small-scale quantum
  computers with unitarily downfolded hamiltonians},\ }\href
  {https://doi.org/10.1103/PRXQuantum.4.020313} {\bibfield  {journal} {\bibinfo
   {journal} {PRX Quantum}\ }\textbf {\bibinfo {volume} {4}},\ \bibinfo {pages}
  {020313} (\bibinfo {year} {2023})}\BibitemShut {NoStop}%
\bibitem [{\citenamefont {Huggins}\ \emph {et~al.}(2021)\citenamefont
  {Huggins}, \citenamefont {McClean}, \citenamefont {Rubin}, \citenamefont
  {Jiang}, \citenamefont {Wiebe}, \citenamefont {Whaley},\ and\ \citenamefont
  {Babbush}}]{huggins2021efficient}%
  \BibitemOpen
  \bibfield  {author} {\bibinfo {author} {\bibfnamefont {W.~J.}\ \bibnamefont
  {Huggins}}, \bibinfo {author} {\bibfnamefont {J.~R.}\ \bibnamefont
  {McClean}}, \bibinfo {author} {\bibfnamefont {N.~C.}\ \bibnamefont {Rubin}},
  \bibinfo {author} {\bibfnamefont {Z.}~\bibnamefont {Jiang}}, \bibinfo
  {author} {\bibfnamefont {N.}~\bibnamefont {Wiebe}}, \bibinfo {author}
  {\bibfnamefont {K.~B.}\ \bibnamefont {Whaley}},\ and\ \bibinfo {author}
  {\bibfnamefont {R.}~\bibnamefont {Babbush}},\ }\bibfield  {title} {\bibinfo
  {title} {Efficient and noise resilient measurements for quantum chemistry on
  near-term quantum computers},\ }\bibfield  {journal} {\bibinfo  {journal}
  {npj Quantum Information}\ }\textbf {\bibinfo {volume} {7}},\ \href
  {https://doi.org/10.1038/s41534-020-00341-7} {10.1038/s41534-020-00341-7}
  (\bibinfo {year} {2021})\BibitemShut {NoStop}%
\bibitem [{\citenamefont {Cohn}\ \emph {et~al.}(2021)\citenamefont {Cohn},
  \citenamefont {Motta},\ and\ \citenamefont {Parrish}}]{cohn2021cdf}%
  \BibitemOpen
  \bibfield  {author} {\bibinfo {author} {\bibfnamefont {J.}~\bibnamefont
  {Cohn}}, \bibinfo {author} {\bibfnamefont {M.}~\bibnamefont {Motta}},\ and\
  \bibinfo {author} {\bibfnamefont {R.~M.}\ \bibnamefont {Parrish}},\
  }\bibfield  {title} {\bibinfo {title} {{Quantum Filter Diagonalization with
  Compressed Double-Factorized Hamiltonians}},\ }\href
  {https://doi.org/10.1103/PRXQuantum.2.040352} {\bibfield  {journal} {\bibinfo
   {journal} {PRX Quantum}\ }\textbf {\bibinfo {volume} {2}},\ \bibinfo {pages}
  {040352} (\bibinfo {year} {2021})}\BibitemShut {NoStop}%
\bibitem [{\citenamefont {Choi}\ \emph
  {et~al.}(2023{\natexlab{a}})\citenamefont {Choi}, \citenamefont {Loaiza},\
  and\ \citenamefont {Izmaylov}}]{choi2023fluid}%
  \BibitemOpen
  \bibfield  {author} {\bibinfo {author} {\bibfnamefont {S.}~\bibnamefont
  {Choi}}, \bibinfo {author} {\bibfnamefont {I.}~\bibnamefont {Loaiza}},\ and\
  \bibinfo {author} {\bibfnamefont {A.~F.}\ \bibnamefont {Izmaylov}},\
  }\bibfield  {title} {\bibinfo {title} {Fluid fermionic fragments for
  optimizing quantum measurements of electronic {H}amiltonians in the
  variational quantum eigensolver},\ }\href
  {https://doi.org/10.22331/q-2023-01-03-889} {\bibfield  {journal} {\bibinfo
  {journal} {{Quantum}}\ }\textbf {\bibinfo {volume} {7}},\ \bibinfo {pages}
  {889} (\bibinfo {year} {2023}{\natexlab{a}})}\BibitemShut {NoStop}%
\bibitem [{\citenamefont {Oumarou}\ \emph {et~al.}(2022)\citenamefont
  {Oumarou}, \citenamefont {Scheurer}, \citenamefont {Parrish}, \citenamefont
  {Hohenstein},\ and\ \citenamefont {Gogolin}}]{oumarou2022accelerating}%
  \BibitemOpen
  \bibfield  {author} {\bibinfo {author} {\bibfnamefont {O.}~\bibnamefont
  {Oumarou}}, \bibinfo {author} {\bibfnamefont {M.}~\bibnamefont {Scheurer}},
  \bibinfo {author} {\bibfnamefont {R.~M.}\ \bibnamefont {Parrish}}, \bibinfo
  {author} {\bibfnamefont {E.~G.}\ \bibnamefont {Hohenstein}},\ and\ \bibinfo
  {author} {\bibfnamefont {C.}~\bibnamefont {Gogolin}},\ }\bibfield  {title}
  {\bibinfo {title} {Accelerating quantum computations of chemistry through
  regularized compressed double factorization},\ }\href@noop {} {\bibfield
  {journal} {\bibinfo  {journal} {arXiv preprint arXiv:2212.07957}\ } (\bibinfo
  {year} {2022})}\BibitemShut {NoStop}%
\bibitem [{\citenamefont {Zhao}\ \emph {et~al.}(2021)\citenamefont {Zhao},
  \citenamefont {Rubin},\ and\ \citenamefont {Miyake}}]{zhao_fermionic_2021}%
  \BibitemOpen
  \bibfield  {author} {\bibinfo {author} {\bibfnamefont {A.}~\bibnamefont
  {Zhao}}, \bibinfo {author} {\bibfnamefont {N.~C.}\ \bibnamefont {Rubin}},\
  and\ \bibinfo {author} {\bibfnamefont {A.}~\bibnamefont {Miyake}},\
  }\bibfield  {title} {\bibinfo {title} {Fermionic {Partial} {Tomography} via
  {Classical} {Shadows}},\ }\href
  {https://doi.org/10.1103/PhysRevLett.127.110504} {\bibfield  {journal}
  {\bibinfo  {journal} {Physical Review Letters}\ }\textbf {\bibinfo {volume}
  {127}},\ \bibinfo {pages} {110504} (\bibinfo {year} {2021})},\ \bibinfo
  {note} {publisher: American Physical Society}\BibitemShut {NoStop}%
\bibitem [{\citenamefont {Low}(2022)}]{low_classical_2022}%
  \BibitemOpen
  \bibfield  {author} {\bibinfo {author} {\bibfnamefont {G.~H.}\ \bibnamefont
  {Low}},\ }\href {https://doi.org/10.48550/arXiv.2208.08964} {\bibinfo {title}
  {Classical shadows of fermions with particle number symmetry}} (\bibinfo
  {year} {2022}),\ \bibinfo {note} {arXiv:2208.08964 [quant-ph]}\BibitemShut
  {NoStop}%
\bibitem [{\citenamefont {Bonet-Monroig}\ \emph {et~al.}(2020)\citenamefont
  {Bonet-Monroig}, \citenamefont {Babbush},\ and\ \citenamefont
  {O'Brien}}]{PhysRevX.10.031064}%
  \BibitemOpen
  \bibfield  {author} {\bibinfo {author} {\bibfnamefont {X.}~\bibnamefont
  {Bonet-Monroig}}, \bibinfo {author} {\bibfnamefont {R.}~\bibnamefont
  {Babbush}},\ and\ \bibinfo {author} {\bibfnamefont {T.~E.}\ \bibnamefont
  {O'Brien}},\ }\bibfield  {title} {\bibinfo {title} {Nearly optimal
  measurement scheduling for partial tomography of quantum states},\ }\href
  {https://doi.org/10.1103/PhysRevX.10.031064} {\bibfield  {journal} {\bibinfo
  {journal} {Phys. Rev. X}\ }\textbf {\bibinfo {volume} {10}},\ \bibinfo
  {pages} {031064} (\bibinfo {year} {2020})}\BibitemShut {NoStop}%
\bibitem [{\citenamefont {Wan}\ \emph {et~al.}(2023)\citenamefont {Wan},
  \citenamefont {Huggins}, \citenamefont {Lee},\ and\ \citenamefont
  {Babbush}}]{wan2022matchgate}%
  \BibitemOpen
  \bibfield  {author} {\bibinfo {author} {\bibfnamefont {K.}~\bibnamefont
  {Wan}}, \bibinfo {author} {\bibfnamefont {W.~J.}\ \bibnamefont {Huggins}},
  \bibinfo {author} {\bibfnamefont {J.}~\bibnamefont {Lee}},\ and\ \bibinfo
  {author} {\bibfnamefont {R.}~\bibnamefont {Babbush}},\ }\bibfield  {title}
  {\bibinfo {title} {Matchgate shadows for fermionic quantum simulation},\
  }\href {https://doi.org/10.1007/s00220-023-04844-0} {\bibfield  {journal}
  {\bibinfo  {journal} {Communications in Mathematical Physics}\ }\textbf
  {\bibinfo {volume} {404}},\ \bibinfo {pages} {629–700} (\bibinfo {year}
  {2023})}\BibitemShut {NoStop}%
\bibitem [{\citenamefont {Huang}\ \emph {et~al.}(2020)\citenamefont {Huang},
  \citenamefont {Kueng},\ and\ \citenamefont {Preskill}}]{huang2020predicting}%
  \BibitemOpen
  \bibfield  {author} {\bibinfo {author} {\bibfnamefont {H.-Y.}\ \bibnamefont
  {Huang}}, \bibinfo {author} {\bibfnamefont {R.}~\bibnamefont {Kueng}},\ and\
  \bibinfo {author} {\bibfnamefont {J.}~\bibnamefont {Preskill}},\ }\bibfield
  {title} {\bibinfo {title} {Predicting many properties of a quantum system
  from very few measurements},\ }\href
  {https://doi.org/10.1038/s41567-020-0932-7} {\bibfield  {journal} {\bibinfo
  {journal} {Nature Physics}\ }\textbf {\bibinfo {volume} {16}},\ \bibinfo
  {pages} {1050} (\bibinfo {year} {2020})}\BibitemShut {NoStop}%
\bibitem [{\citenamefont {Huggins}\ \emph {et~al.}(2022)\citenamefont
  {Huggins}, \citenamefont {O’Gorman}, \citenamefont {Rubin}, \citenamefont
  {Reichman}, \citenamefont {Babbush},\ and\ \citenamefont
  {Lee}}]{huggins_unbiasing_2022}%
  \BibitemOpen
  \bibfield  {author} {\bibinfo {author} {\bibfnamefont {W.~J.}\ \bibnamefont
  {Huggins}}, \bibinfo {author} {\bibfnamefont {B.~A.}\ \bibnamefont
  {O’Gorman}}, \bibinfo {author} {\bibfnamefont {N.~C.}\ \bibnamefont
  {Rubin}}, \bibinfo {author} {\bibfnamefont {D.~R.}\ \bibnamefont {Reichman}},
  \bibinfo {author} {\bibfnamefont {R.}~\bibnamefont {Babbush}},\ and\ \bibinfo
  {author} {\bibfnamefont {J.}~\bibnamefont {Lee}},\ }\bibfield  {title}
  {\bibinfo {title} {Unbiasing fermionic quantum {Monte} {Carlo} with a quantum
  computer},\ }\href {https://doi.org/10.1038/s41586-021-04351-z} {\bibfield
  {journal} {\bibinfo  {journal} {Nature}\ }\textbf {\bibinfo {volume} {603}},\
  \bibinfo {pages} {416} (\bibinfo {year} {2022})},\ \bibinfo {note} {number:
  7901 Publisher: Nature Publishing Group}\BibitemShut {NoStop}%
\bibitem [{\citenamefont {Lee}\ and\ \citenamefont
  {Taylor}(1989)}]{lee1989diagnostic}%
  \BibitemOpen
  \bibfield  {author} {\bibinfo {author} {\bibfnamefont {T.~J.}\ \bibnamefont
  {Lee}}\ and\ \bibinfo {author} {\bibfnamefont {P.~R.}\ \bibnamefont
  {Taylor}},\ }\bibfield  {title} {\bibinfo {title} {A diagnostic for
  determining the quality of single-reference electron correlation methods},\
  }\href {https://doi.org/10.1002/QUA.560360824} {\bibfield  {journal}
  {\bibinfo  {journal} {International Journal of Quantum Chemistry}\ }\textbf
  {\bibinfo {volume} {36}},\ \bibinfo {pages} {199} (\bibinfo {year}
  {1989})}\BibitemShut {NoStop}%
\bibitem [{\citenamefont {Janssen}\ and\ \citenamefont
  {Nielsen}(1998)}]{janssen1998new}%
  \BibitemOpen
  \bibfield  {author} {\bibinfo {author} {\bibfnamefont {C.~L.}\ \bibnamefont
  {Janssen}}\ and\ \bibinfo {author} {\bibfnamefont {I.~M.}\ \bibnamefont
  {Nielsen}},\ }\bibfield  {title} {\bibinfo {title} {New diagnostics for
  coupled-cluster and møller–plesset perturbation theory},\ }\href
  {https://doi.org/10.1016/S0009-2614(98)00504-1} {\bibfield  {journal}
  {\bibinfo  {journal} {Chemical Physics Letters}\ }\textbf {\bibinfo {volume}
  {290}},\ \bibinfo {pages} {423} (\bibinfo {year} {1998})}\BibitemShut
  {NoStop}%
\bibitem [{\citenamefont
  {{\v{C}}{\'\i}{\v{z}}ek}(1966)}]{ciczek1966correlation}%
  \BibitemOpen
  \bibfield  {author} {\bibinfo {author} {\bibfnamefont {J.}~\bibnamefont
  {{\v{C}}{\'\i}{\v{z}}ek}},\ }\bibfield  {title} {\bibinfo {title} {On the
  correlation problem in atomic and molecular systems. calculation of
  wavefunction components in ursell-type expansion using quantum-field
  theoretical methods},\ }\href {https://doi.org/10.1063/1.1727484} {\bibfield
  {journal} {\bibinfo  {journal} {The Journal of Chemical Physics}\ }\textbf
  {\bibinfo {volume} {45}},\ \bibinfo {pages} {4256} (\bibinfo {year}
  {1966})}\BibitemShut {NoStop}%
\bibitem [{\citenamefont {Crawford}\ and\ \citenamefont
  {Schaefer}(2006)}]{crawford2006introduction}%
  \BibitemOpen
  \bibfield  {author} {\bibinfo {author} {\bibfnamefont {T.~D.}\ \bibnamefont
  {Crawford}}\ and\ \bibinfo {author} {\bibfnamefont {H.~F.}\ \bibnamefont
  {Schaefer}},\ }\bibfield  {title} {\bibinfo {title} {An introduction to
  coupled cluster theory for computational chemists},\ }\href
  {https://doi.org/10.1002/9780470125915.CH2} {\bibfield  {journal} {\bibinfo
  {journal} {Reviews in Computational Chemistry}\ }\textbf {\bibinfo {volume}
  {14}},\ \bibinfo {pages} {33} (\bibinfo {year} {2006})}\BibitemShut {NoStop}%
\bibitem [{\citenamefont {Shavitt}\ and\ \citenamefont
  {Bartlett}(2009)}]{shavitt2009many}%
  \BibitemOpen
  \bibfield  {author} {\bibinfo {author} {\bibfnamefont {I.}~\bibnamefont
  {Shavitt}}\ and\ \bibinfo {author} {\bibfnamefont {R.~J.}\ \bibnamefont
  {Bartlett}},\ }\href@noop {} {\emph {\bibinfo {title} {Many-body methods in
  chemistry and physics: MBPT and coupled-cluster theory}}}\ (\bibinfo
  {publisher} {Cambridge University Press},\ \bibinfo {year}
  {2009})\BibitemShut {NoStop}%
\bibitem [{\citenamefont {Purvis~III}\ and\ \citenamefont
  {Bartlett}(1982)}]{purvis1982full}%
  \BibitemOpen
  \bibfield  {author} {\bibinfo {author} {\bibfnamefont {G.~D.}\ \bibnamefont
  {Purvis~III}}\ and\ \bibinfo {author} {\bibfnamefont {R.~J.}\ \bibnamefont
  {Bartlett}},\ }\bibfield  {title} {\bibinfo {title} {A full coupled-cluster
  singles and doubles model: The inclusion of disconnected triples},\ }\href
  {https://doi.org/10.1063/1.443164} {\bibfield  {journal} {\bibinfo  {journal}
  {The Journal of Chemical Physics}\ }\textbf {\bibinfo {volume} {76}},\
  \bibinfo {pages} {1910} (\bibinfo {year} {1982})}\BibitemShut {NoStop}%
\bibitem [{\citenamefont {Raghavachari}\ \emph {et~al.}(1989)\citenamefont
  {Raghavachari}, \citenamefont {Trucks}, \citenamefont {Pople},\ and\
  \citenamefont {Head-Gordon}}]{raghavachari1989fifth}%
  \BibitemOpen
  \bibfield  {author} {\bibinfo {author} {\bibfnamefont {K.}~\bibnamefont
  {Raghavachari}}, \bibinfo {author} {\bibfnamefont {G.~W.}\ \bibnamefont
  {Trucks}}, \bibinfo {author} {\bibfnamefont {J.~A.}\ \bibnamefont {Pople}},\
  and\ \bibinfo {author} {\bibfnamefont {M.}~\bibnamefont {Head-Gordon}},\
  }\bibfield  {title} {\bibinfo {title} {A fifth-order perturbation comparison
  of electron correlation theories},\ }\href
  {https://doi.org/https://doi.org/10.1016/S0009-2614(89)87395-6} {\bibfield
  {journal} {\bibinfo  {journal} {Chemical Physics Letters}\ }\textbf {\bibinfo
  {volume} {157}},\ \bibinfo {pages} {479} (\bibinfo {year}
  {1989})}\BibitemShut {NoStop}%
\bibitem [{\citenamefont {Bartlett}(2005)}]{bartlett2005how}%
  \BibitemOpen
  \bibfield  {author} {\bibinfo {author} {\bibfnamefont {R.~J.}\ \bibnamefont
  {Bartlett}},\ }\bibfield  {title} {\bibinfo {title} {How and why
  coupled-cluster theory became the pre-eminent method in an ab into quantum
  chemistry},\ }in\ \href
  {https://doi.org/https://doi.org/10.1016/B978-044451719-7/50085-8} {\emph
  {\bibinfo {booktitle} {Theory and Applications of Computational Chemistry}}}\
  (\bibinfo  {publisher} {Elsevier},\ \bibinfo {year} {2005})\ pp.\ \bibinfo
  {pages} {1191--1221}\BibitemShut {NoStop}%
\bibitem [{\citenamefont {Guo}\ \emph {et~al.}(2018)\citenamefont {Guo},
  \citenamefont {Riplinger}, \citenamefont {Becker}, \citenamefont {Liakos},
  \citenamefont {Minenkov}, \citenamefont {Cavallo},\ and\ \citenamefont
  {Neese}}]{guo2018communication}%
  \BibitemOpen
  \bibfield  {author} {\bibinfo {author} {\bibfnamefont {Y.}~\bibnamefont
  {Guo}}, \bibinfo {author} {\bibfnamefont {C.}~\bibnamefont {Riplinger}},
  \bibinfo {author} {\bibfnamefont {U.}~\bibnamefont {Becker}}, \bibinfo
  {author} {\bibfnamefont {D.~G.}\ \bibnamefont {Liakos}}, \bibinfo {author}
  {\bibfnamefont {Y.}~\bibnamefont {Minenkov}}, \bibinfo {author}
  {\bibfnamefont {L.}~\bibnamefont {Cavallo}},\ and\ \bibinfo {author}
  {\bibfnamefont {F.}~\bibnamefont {Neese}},\ }\bibfield  {title} {\bibinfo
  {title} {Communication: An improved linear scaling perturbative triples
  correction for the domain based local pair-natural orbital based singles and
  doubles coupled cluster method {[DLPNO-CCSD(T)]}},\ }\bibfield  {journal}
  {\bibinfo  {journal} {The Journal of chemical physics}\ }\textbf {\bibinfo
  {volume} {148}},\ \href {https://doi.org/10.1063/1.5011798}
  {10.1063/1.5011798} (\bibinfo {year} {2018})\BibitemShut {NoStop}%
\bibitem [{\citenamefont {Sherrill}\ and\ \citenamefont
  {Schaefer~III}(1999)}]{sherrill1999configuration}%
  \BibitemOpen
  \bibfield  {author} {\bibinfo {author} {\bibfnamefont {C.~D.}\ \bibnamefont
  {Sherrill}}\ and\ \bibinfo {author} {\bibfnamefont {H.~F.}\ \bibnamefont
  {Schaefer~III}},\ }\bibfield  {title} {\bibinfo {title} {The configuration
  interaction method: Advances in highly correlated approaches},\ }in\ \href
  {https://doi.org/https://doi.org/10.1016/S0065-3276(08)60532-8} {\emph
  {\bibinfo {booktitle} {Advances in Quantum Chemistry}}},\ Vol.~\bibinfo
  {volume} {34}\ (\bibinfo  {publisher} {Elsevier},\ \bibinfo {year} {1999})\
  pp.\ \bibinfo {pages} {143--269}\BibitemShut {NoStop}%
\bibitem [{\citenamefont {Liakos}\ \emph {et~al.}(2020)\citenamefont {Liakos},
  \citenamefont {Guo},\ and\ \citenamefont {Neese}}]{liakos2020comprehensive}%
  \BibitemOpen
  \bibfield  {author} {\bibinfo {author} {\bibfnamefont {D.~G.}\ \bibnamefont
  {Liakos}}, \bibinfo {author} {\bibfnamefont {Y.}~\bibnamefont {Guo}},\ and\
  \bibinfo {author} {\bibfnamefont {F.}~\bibnamefont {Neese}},\ }\bibfield
  {title} {\bibinfo {title} {Comprehensive benchmark results for the domain
  based local pair natural orbital coupled cluster method (dlpno-ccsd(t)) for
  closed- and open-shell systems},\ }\href
  {https://pubs.acs.org/doi/full/10.1021/acs.jpca.9b05734} {\bibfield
  {journal} {\bibinfo  {journal} {Journal of Physical Chemistry A}\ }\textbf
  {\bibinfo {volume} {124}},\ \bibinfo {pages} {90} (\bibinfo {year}
  {2020})}\BibitemShut {NoStop}%
\bibitem [{\citenamefont {M{\"o}rchen}\ \emph {et~al.}(2020)\citenamefont
  {M{\"o}rchen}, \citenamefont {Freitag},\ and\ \citenamefont
  {Reiher}}]{morchen2020tailored}%
  \BibitemOpen
  \bibfield  {author} {\bibinfo {author} {\bibfnamefont {M.}~\bibnamefont
  {M{\"o}rchen}}, \bibinfo {author} {\bibfnamefont {L.}~\bibnamefont
  {Freitag}},\ and\ \bibinfo {author} {\bibfnamefont {M.}~\bibnamefont
  {Reiher}},\ }\bibfield  {title} {\bibinfo {title} {Tailored coupled cluster
  theory in varying correlation regimes},\ }\href
  {https://doi.org/10.1063/5.0032661} {\bibfield  {journal} {\bibinfo
  {journal} {The Journal of Chemical Physics}\ }\textbf {\bibinfo {volume}
  {153}},\ \bibinfo {pages} {244113} (\bibinfo {year} {2020})}\BibitemShut
  {NoStop}%
\bibitem [{\citenamefont {Oliphant}\ and\ \citenamefont
  {Adamowicz}(1993)}]{oliphant1993multireference}%
  \BibitemOpen
  \bibfield  {author} {\bibinfo {author} {\bibfnamefont {N.}~\bibnamefont
  {Oliphant}}\ and\ \bibinfo {author} {\bibfnamefont {L.}~\bibnamefont
  {Adamowicz}},\ }\bibfield  {title} {\bibinfo {title} {Multireference coupled
  cluster method for electronic structure of molecules},\ }\href
  {https://doi.org/10.1080/01442359309353285} {\bibfield  {journal} {\bibinfo
  {journal} {International Reviews in Physical Chemistry}\ }\textbf {\bibinfo
  {volume} {12}},\ \bibinfo {pages} {339} (\bibinfo {year} {1993})}\BibitemShut
  {NoStop}%
\bibitem [{\citenamefont {Knowles}\ and\ \citenamefont
  {Handy}(1984)}]{knowles1984new}%
  \BibitemOpen
  \bibfield  {author} {\bibinfo {author} {\bibfnamefont {P.~J.}\ \bibnamefont
  {Knowles}}\ and\ \bibinfo {author} {\bibfnamefont {N.~C.}\ \bibnamefont
  {Handy}},\ }\bibfield  {title} {\bibinfo {title} {A new determinant-based
  full configuration interaction method},\ }\href
  {https://doi.org/https://doi.org/10.1016/0009-2614(84)85513-X} {\bibfield
  {journal} {\bibinfo  {journal} {Chemical Physics Letters}\ }\textbf {\bibinfo
  {volume} {111}},\ \bibinfo {pages} {315} (\bibinfo {year}
  {1984})}\BibitemShut {NoStop}%
\bibitem [{\citenamefont {Olsen}\ \emph {et~al.}(1988)\citenamefont {Olsen},
  \citenamefont {Roos}, \citenamefont {Jo/rgensen},\ and\ \citenamefont
  {Jensen}}]{olsen1988determinant}%
  \BibitemOpen
  \bibfield  {author} {\bibinfo {author} {\bibfnamefont {J.}~\bibnamefont
  {Olsen}}, \bibinfo {author} {\bibfnamefont {B.~O.}\ \bibnamefont {Roos}},
  \bibinfo {author} {\bibfnamefont {P.}~\bibnamefont {Jo/rgensen}},\ and\
  \bibinfo {author} {\bibfnamefont {H.~J.~A.}\ \bibnamefont {Jensen}},\
  }\bibfield  {title} {\bibinfo {title} {Determinant based configuration
  interaction algorithms for complete and restricted configuration interaction
  spaces},\ }\href {https://doi.org/10.1063/1.455063} {\bibfield  {journal}
  {\bibinfo  {journal} {The Journal of chemical physics}\ }\textbf {\bibinfo
  {volume} {89}},\ \bibinfo {pages} {2185} (\bibinfo {year}
  {1988})}\BibitemShut {NoStop}%
\bibitem [{\citenamefont {Knowles}\ and\ \citenamefont
  {Handy}(1989)}]{knowles1989determinant}%
  \BibitemOpen
  \bibfield  {author} {\bibinfo {author} {\bibfnamefont {P.~J.}\ \bibnamefont
  {Knowles}}\ and\ \bibinfo {author} {\bibfnamefont {N.~C.}\ \bibnamefont
  {Handy}},\ }\bibfield  {title} {\bibinfo {title} {A determinant based full
  configuration interaction program},\ }\href
  {https://doi.org/https://doi.org/10.1016/0010-4655(89)90033-7} {\bibfield
  {journal} {\bibinfo  {journal} {Computer physics communications}\ }\textbf
  {\bibinfo {volume} {54}},\ \bibinfo {pages} {75} (\bibinfo {year}
  {1989})}\BibitemShut {NoStop}%
\bibitem [{\citenamefont {Zarrabian}\ \emph {et~al.}(1989)\citenamefont
  {Zarrabian}, \citenamefont {Sarma},\ and\ \citenamefont
  {Paldus}}]{zarrabian1989vectorizable}%
  \BibitemOpen
  \bibfield  {author} {\bibinfo {author} {\bibfnamefont {S.}~\bibnamefont
  {Zarrabian}}, \bibinfo {author} {\bibfnamefont {C.}~\bibnamefont {Sarma}},\
  and\ \bibinfo {author} {\bibfnamefont {J.}~\bibnamefont {Paldus}},\
  }\bibfield  {title} {\bibinfo {title} {Vectorizable approach to molecular ci
  problems using determinantal basis},\ }\href
  {https://doi.org/https://doi.org/10.1016/0009-2614(89)85346-1} {\bibfield
  {journal} {\bibinfo  {journal} {Chemical physics letters}\ }\textbf {\bibinfo
  {volume} {155}},\ \bibinfo {pages} {183} (\bibinfo {year}
  {1989})}\BibitemShut {NoStop}%
\bibitem [{\citenamefont {Bendazzoli}\ and\ \citenamefont
  {Evangelisti}(1993)}]{bendazzoli1993vector}%
  \BibitemOpen
  \bibfield  {author} {\bibinfo {author} {\bibfnamefont {G.~L.}\ \bibnamefont
  {Bendazzoli}}\ and\ \bibinfo {author} {\bibfnamefont {S.}~\bibnamefont
  {Evangelisti}},\ }\bibfield  {title} {\bibinfo {title} {A vector and parallel
  full configuration interaction algorithm},\ }\href
  {https://doi.org/10.1063/1.464087} {\bibfield  {journal} {\bibinfo  {journal}
  {The Journal of chemical physics}\ }\textbf {\bibinfo {volume} {98}},\
  \bibinfo {pages} {3141} (\bibinfo {year} {1993})}\BibitemShut {NoStop}%
\bibitem [{\citenamefont {Stein}\ \emph {et~al.}(2016)\citenamefont {Stein},
  \citenamefont {von Burg},\ and\ \citenamefont {Reiher}}]{stein2016delicate}%
  \BibitemOpen
  \bibfield  {author} {\bibinfo {author} {\bibfnamefont {C.~J.}\ \bibnamefont
  {Stein}}, \bibinfo {author} {\bibfnamefont {V.}~\bibnamefont {von Burg}},\
  and\ \bibinfo {author} {\bibfnamefont {M.}~\bibnamefont {Reiher}},\
  }\bibfield  {title} {\bibinfo {title} {The delicate balance of static and
  dynamic electron correlation},\ }\href
  {https://doi.org/10.1021/acs.jctc.6b00528} {\bibfield  {journal} {\bibinfo
  {journal} {Journal of Chemical Theory and Computation}\ }\textbf {\bibinfo
  {volume} {12}},\ \bibinfo {pages} {3764–3773} (\bibinfo {year}
  {2016})}\BibitemShut {NoStop}%
\bibitem [{\citenamefont {Harrison}(1991)}]{harrison1991approximating}%
  \BibitemOpen
  \bibfield  {author} {\bibinfo {author} {\bibfnamefont {R.~J.}\ \bibnamefont
  {Harrison}},\ }\bibfield  {title} {\bibinfo {title} {Approximating full
  configuration interaction with selected configuration interaction and
  perturbation theory},\ }\href {https://doi.org/10.1063/1.460537} {\bibfield
  {journal} {\bibinfo  {journal} {The Journal of chemical physics}\ }\textbf
  {\bibinfo {volume} {94}},\ \bibinfo {pages} {5021} (\bibinfo {year}
  {1991})}\BibitemShut {NoStop}%
\bibitem [{\citenamefont {Holmes}\ \emph {et~al.}(2016)\citenamefont {Holmes},
  \citenamefont {Tubman},\ and\ \citenamefont {Umrigar}}]{holmes2016heat}%
  \BibitemOpen
  \bibfield  {author} {\bibinfo {author} {\bibfnamefont {A.~A.}\ \bibnamefont
  {Holmes}}, \bibinfo {author} {\bibfnamefont {N.~M.}\ \bibnamefont {Tubman}},\
  and\ \bibinfo {author} {\bibfnamefont {C.}~\bibnamefont {Umrigar}},\
  }\bibfield  {title} {\bibinfo {title} {Heat-bath configuration interaction:
  An efficient selected configuration interaction algorithm inspired by
  heat-bath sampling},\ }\href {https://doi.org/10.1021/acs.jctc.6b00407}
  {\bibfield  {journal} {\bibinfo  {journal} {Journal of chemical theory and
  computation}\ }\textbf {\bibinfo {volume} {12}},\ \bibinfo {pages} {3674}
  (\bibinfo {year} {2016})}\BibitemShut {NoStop}%
\bibitem [{\citenamefont {Sharma}\ \emph {et~al.}(2017)\citenamefont {Sharma},
  \citenamefont {Holmes}, \citenamefont {Jeanmairet}, \citenamefont {Alavi},\
  and\ \citenamefont {Umrigar}}]{sharma2017semistochastic}%
  \BibitemOpen
  \bibfield  {author} {\bibinfo {author} {\bibfnamefont {S.}~\bibnamefont
  {Sharma}}, \bibinfo {author} {\bibfnamefont {A.~A.}\ \bibnamefont {Holmes}},
  \bibinfo {author} {\bibfnamefont {G.}~\bibnamefont {Jeanmairet}}, \bibinfo
  {author} {\bibfnamefont {A.}~\bibnamefont {Alavi}},\ and\ \bibinfo {author}
  {\bibfnamefont {C.~J.}\ \bibnamefont {Umrigar}},\ }\bibfield  {title}
  {\bibinfo {title} {Semistochastic heat-bath configuration interaction method:
  Selected configuration interaction with semistochastic perturbation theory},\
  }\href {https://doi.org/10.1021/acs.jctc.6b01028} {\bibfield  {journal}
  {\bibinfo  {journal} {Journal of chemical theory and computation}\ }\textbf
  {\bibinfo {volume} {13}},\ \bibinfo {pages} {1595} (\bibinfo {year}
  {2017})}\BibitemShut {NoStop}%
\bibitem [{\citenamefont {Schriber}\ and\ \citenamefont
  {Evangelista}(2016)}]{schriber2016communication}%
  \BibitemOpen
  \bibfield  {author} {\bibinfo {author} {\bibfnamefont {J.~B.}\ \bibnamefont
  {Schriber}}\ and\ \bibinfo {author} {\bibfnamefont {F.~A.}\ \bibnamefont
  {Evangelista}},\ }\bibfield  {title} {\bibinfo {title} {Communication: An
  adaptive configuration interaction approach for strongly correlated electrons
  with tunable accuracy},\ }\bibfield  {journal} {\bibinfo  {journal} {The
  Journal of chemical physics}\ }\textbf {\bibinfo {volume} {144}},\ \href
  {https://doi.org/10.1063/1.4948308} {10.1063/1.4948308} (\bibinfo {year}
  {2016})\BibitemShut {NoStop}%
\bibitem [{\citenamefont {Tubman}\ \emph {et~al.}(2020)\citenamefont {Tubman},
  \citenamefont {Freeman}, \citenamefont {Levine}, \citenamefont {Hait},
  \citenamefont {Head-Gordon},\ and\ \citenamefont
  {Whaley}}]{tubman2020modern}%
  \BibitemOpen
  \bibfield  {author} {\bibinfo {author} {\bibfnamefont {N.~M.}\ \bibnamefont
  {Tubman}}, \bibinfo {author} {\bibfnamefont {C.~D.}\ \bibnamefont {Freeman}},
  \bibinfo {author} {\bibfnamefont {D.~S.}\ \bibnamefont {Levine}}, \bibinfo
  {author} {\bibfnamefont {D.}~\bibnamefont {Hait}}, \bibinfo {author}
  {\bibfnamefont {M.}~\bibnamefont {Head-Gordon}},\ and\ \bibinfo {author}
  {\bibfnamefont {K.~B.}\ \bibnamefont {Whaley}},\ }\bibfield  {title}
  {\bibinfo {title} {Modern approaches to exact diagonalization and selected
  configuration interaction with the adaptive sampling {CI} method},\ }\href
  {https://doi.org/10.1021/acs.jctc.8b00536} {\bibfield  {journal} {\bibinfo
  {journal} {Journal of chemical theory and computation}\ }\textbf {\bibinfo
  {volume} {16}},\ \bibinfo {pages} {2139} (\bibinfo {year}
  {2020})}\BibitemShut {NoStop}%
\bibitem [{\citenamefont {Zhai}\ \emph {et~al.}(2023)\citenamefont {Zhai},
  \citenamefont {Larsson}, \citenamefont {Lee}, \citenamefont {Cui},
  \citenamefont {Zhu}, \citenamefont {Sun}, \citenamefont {Peng}, \citenamefont
  {Peng}, \citenamefont {Liao}, \citenamefont {Tölle}, \citenamefont {Yang},
  \citenamefont {Li},\ and\ \citenamefont {Chan}}]{zhai2023block2}%
  \BibitemOpen
  \bibfield  {author} {\bibinfo {author} {\bibfnamefont {H.}~\bibnamefont
  {Zhai}}, \bibinfo {author} {\bibfnamefont {H.~R.}\ \bibnamefont {Larsson}},
  \bibinfo {author} {\bibfnamefont {S.}~\bibnamefont {Lee}}, \bibinfo {author}
  {\bibfnamefont {Z.-H.}\ \bibnamefont {Cui}}, \bibinfo {author} {\bibfnamefont
  {T.}~\bibnamefont {Zhu}}, \bibinfo {author} {\bibfnamefont {C.}~\bibnamefont
  {Sun}}, \bibinfo {author} {\bibfnamefont {L.}~\bibnamefont {Peng}}, \bibinfo
  {author} {\bibfnamefont {R.}~\bibnamefont {Peng}}, \bibinfo {author}
  {\bibfnamefont {K.}~\bibnamefont {Liao}}, \bibinfo {author} {\bibfnamefont
  {J.}~\bibnamefont {Tölle}}, \bibinfo {author} {\bibfnamefont
  {J.}~\bibnamefont {Yang}}, \bibinfo {author} {\bibfnamefont {S.}~\bibnamefont
  {Li}},\ and\ \bibinfo {author} {\bibfnamefont {G.~K.-L.}\ \bibnamefont
  {Chan}},\ }\href@noop {} {\bibinfo {title} {Block2: a comprehensive open
  source framework to develop and apply state-of-the-art dmrg algorithms in
  electronic structure and beyond}} (\bibinfo {year} {2023}),\ \Eprint
  {https://arxiv.org/abs/2310.03920} {arXiv:2310.03920 [physics.chem-ph]}
  \BibitemShut {NoStop}%
\bibitem [{\citenamefont {Angeli}\ \emph
  {et~al.}(2001{\natexlab{a}})\citenamefont {Angeli}, \citenamefont
  {Cimiraglia}, \citenamefont {Evangelisti}, \citenamefont {Leininger},\ and\
  \citenamefont {Malrieu}}]{angeli2001introduction}%
  \BibitemOpen
  \bibfield  {author} {\bibinfo {author} {\bibfnamefont {C.}~\bibnamefont
  {Angeli}}, \bibinfo {author} {\bibfnamefont {R.}~\bibnamefont {Cimiraglia}},
  \bibinfo {author} {\bibfnamefont {S.}~\bibnamefont {Evangelisti}}, \bibinfo
  {author} {\bibfnamefont {T.}~\bibnamefont {Leininger}},\ and\ \bibinfo
  {author} {\bibfnamefont {J.~P.}\ \bibnamefont {Malrieu}},\ }\bibfield
  {title} {\bibinfo {title} {Introduction of n-electron valence states for
  multireference perturbation theory},\ }\href
  {https://doi.org/10.1063/1.1361246} {\bibfield  {journal} {\bibinfo
  {journal} {The Journal of Chemical Physics}\ }\textbf {\bibinfo {volume}
  {114}},\ \bibinfo {pages} {10252} (\bibinfo {year}
  {2001}{\natexlab{a}})}\BibitemShut {NoStop}%
\bibitem [{\citenamefont {Angeli}\ \emph
  {et~al.}(2001{\natexlab{b}})\citenamefont {Angeli}, \citenamefont
  {Cimiraglia},\ and\ \citenamefont {Malrieu}}]{angeli2001nelectron}%
  \BibitemOpen
  \bibfield  {author} {\bibinfo {author} {\bibfnamefont {C.}~\bibnamefont
  {Angeli}}, \bibinfo {author} {\bibfnamefont {R.}~\bibnamefont {Cimiraglia}},\
  and\ \bibinfo {author} {\bibfnamefont {J.~P.}\ \bibnamefont {Malrieu}},\
  }\bibfield  {title} {\bibinfo {title} {{N-electron} valence state
  perturbation theory: a fast implementation of the strongly contracted
  variant},\ }\href {https://doi.org/10.1016/S0009-2614(01)01303-3} {\bibfield
  {journal} {\bibinfo  {journal} {Chemical Physics Letters}\ }\textbf {\bibinfo
  {volume} {350}},\ \bibinfo {pages} {297} (\bibinfo {year}
  {2001}{\natexlab{b}})}\BibitemShut {NoStop}%
\bibitem [{\citenamefont {Angeli}\ \emph {et~al.}(2002)\citenamefont {Angeli},
  \citenamefont {Cimiraglia},\ and\ \citenamefont
  {Malrieu}}]{angeli2002nelectron}%
  \BibitemOpen
  \bibfield  {author} {\bibinfo {author} {\bibfnamefont {C.}~\bibnamefont
  {Angeli}}, \bibinfo {author} {\bibfnamefont {R.}~\bibnamefont {Cimiraglia}},\
  and\ \bibinfo {author} {\bibfnamefont {J.~P.}\ \bibnamefont {Malrieu}},\
  }\bibfield  {title} {\bibinfo {title} {{N-electron} valence state
  perturbation theory: A spinless formulation and an efficient implementation
  of the strongly contracted and of the partially contracted variants},\ }\href
  {https://doi.org/10.1063/1.1515317} {\bibfield  {journal} {\bibinfo
  {journal} {The Journal of Chemical Physics}\ }\textbf {\bibinfo {volume}
  {117}},\ \bibinfo {pages} {9138} (\bibinfo {year} {2002})}\BibitemShut
  {NoStop}%
\bibitem [{\citenamefont {Pulay}(2011)}]{pulay2011perspective}%
  \BibitemOpen
  \bibfield  {author} {\bibinfo {author} {\bibfnamefont {P.}~\bibnamefont
  {Pulay}},\ }\bibfield  {title} {\bibinfo {title} {A perspective on the
  {CASPT2} method},\ }\href {https://doi.org/https://doi.org/10.1002/qua.23052}
  {\bibfield  {journal} {\bibinfo  {journal} {International Journal of Quantum
  Chemistry}\ }\textbf {\bibinfo {volume} {111}},\ \bibinfo {pages} {3273}
  (\bibinfo {year} {2011})}\BibitemShut {NoStop}%
\bibitem [{\citenamefont {Battaglia}\ \emph {et~al.}(2023)\citenamefont
  {Battaglia}, \citenamefont {Galv{\'{a} }n},\ and\ \citenamefont
  {Lindh}}]{battaglia2023multiconfigurational}%
  \BibitemOpen
  \bibfield  {author} {\bibinfo {author} {\bibfnamefont {S.}~\bibnamefont
  {Battaglia}}, \bibinfo {author} {\bibfnamefont {I.~F.}\ \bibnamefont
  {Galv{\'{a} }n}},\ and\ \bibinfo {author} {\bibfnamefont {R.}~\bibnamefont
  {Lindh}},\ }\bibfield  {title} {\bibinfo {title} {Multiconfigurational
  quantum chemistry: The {CASPT}2 method},\ }in\ \href
  {https://doi.org/10.1016/b978-0-323-91738-4.00016-6} {\emph {\bibinfo
  {booktitle} {Theoretical and Computational Photochemistry}}}\ (\bibinfo
  {publisher} {Elsevier},\ \bibinfo {year} {2023})\ pp.\ \bibinfo {pages}
  {135--162}\BibitemShut {NoStop}%
\bibitem [{\citenamefont {Paldus}\ \emph {et~al.}(1984)\citenamefont {Paldus},
  \citenamefont {{\v{C}}{\'\i}{\v{z}}ek},\ and\ \citenamefont
  {Takahashi}}]{paldus1984approximate}%
  \BibitemOpen
  \bibfield  {author} {\bibinfo {author} {\bibfnamefont {J.}~\bibnamefont
  {Paldus}}, \bibinfo {author} {\bibfnamefont {J.}~\bibnamefont
  {{\v{C}}{\'\i}{\v{z}}ek}},\ and\ \bibinfo {author} {\bibfnamefont
  {M.}~\bibnamefont {Takahashi}},\ }\bibfield  {title} {\bibinfo {title}
  {Approximate account of the connected quadruply excited clusters in the
  coupled-pair many-electron theory},\ }\href
  {https://doi.org/10.1103/PhysRevA.30.2193} {\bibfield  {journal} {\bibinfo
  {journal} {Physical Review A}\ }\textbf {\bibinfo {volume} {30}},\ \bibinfo
  {pages} {2193} (\bibinfo {year} {1984})}\BibitemShut {NoStop}%
\bibitem [{\citenamefont {Paldus}\ and\ \citenamefont
  {Planelles}(1994)}]{paldus1994valence}%
  \BibitemOpen
  \bibfield  {author} {\bibinfo {author} {\bibfnamefont {J.}~\bibnamefont
  {Paldus}}\ and\ \bibinfo {author} {\bibfnamefont {J.}~\bibnamefont
  {Planelles}},\ }\bibfield  {title} {\bibinfo {title} {Valence bond corrected
  single reference coupled cluster approach: {I.} general formalism},\ }\href
  {https://doi.org/10.1007/BF01123868} {\bibfield  {journal} {\bibinfo
  {journal} {Theoretica chimica acta}\ }\textbf {\bibinfo {volume} {89}},\
  \bibinfo {pages} {13} (\bibinfo {year} {1994})}\BibitemShut {NoStop}%
\bibitem [{\citenamefont {Planelles}\ \emph
  {et~al.}(1994{\natexlab{a}})\citenamefont {Planelles}, \citenamefont
  {Paldus},\ and\ \citenamefont {Li}}]{planelles1994valence}%
  \BibitemOpen
  \bibfield  {author} {\bibinfo {author} {\bibfnamefont {J.}~\bibnamefont
  {Planelles}}, \bibinfo {author} {\bibfnamefont {J.}~\bibnamefont {Paldus}},\
  and\ \bibinfo {author} {\bibfnamefont {X.}~\bibnamefont {Li}},\ }\bibfield
  {title} {\bibinfo {title} {Valence bond corrected single reference coupled
  cluster approach: {II.} application to ppp model systems},\ }\href
  {https://doi.org/10.1007/BF01167280} {\bibfield  {journal} {\bibinfo
  {journal} {Theoretica chimica acta}\ }\textbf {\bibinfo {volume} {89}},\
  \bibinfo {pages} {33} (\bibinfo {year} {1994}{\natexlab{a}})}\BibitemShut
  {NoStop}%
\bibitem [{\citenamefont {Planelles}\ \emph
  {et~al.}(1994{\natexlab{b}})\citenamefont {Planelles}, \citenamefont
  {Paldus},\ and\ \citenamefont {Li}}]{planelles1994valence3}%
  \BibitemOpen
  \bibfield  {author} {\bibinfo {author} {\bibfnamefont {J.}~\bibnamefont
  {Planelles}}, \bibinfo {author} {\bibfnamefont {J.}~\bibnamefont {Paldus}},\
  and\ \bibinfo {author} {\bibfnamefont {X.}~\bibnamefont {Li}},\ }\bibfield
  {title} {\bibinfo {title} {Valence bond corrected single reference coupled
  cluster approach: {III.} simple model of bond breaking or formation},\ }\href
  {https://doi.org/10.1007/BF01123870} {\bibfield  {journal} {\bibinfo
  {journal} {Theoretica chimica acta}\ }\textbf {\bibinfo {volume} {89}},\
  \bibinfo {pages} {59} (\bibinfo {year} {1994}{\natexlab{b}})}\BibitemShut
  {NoStop}%
\bibitem [{\citenamefont {Li}\ and\ \citenamefont
  {Paldus}(1997)}]{li1997reduced}%
  \BibitemOpen
  \bibfield  {author} {\bibinfo {author} {\bibfnamefont {X.}~\bibnamefont
  {Li}}\ and\ \bibinfo {author} {\bibfnamefont {J.}~\bibnamefont {Paldus}},\
  }\bibfield  {title} {\bibinfo {title} {Reduced multireference {CCSD} method:
  An effective approach to quasidegenerate states},\ }\href
  {https://doi.org/10.1063/1.474289} {\bibfield  {journal} {\bibinfo  {journal}
  {The Journal of chemical physics}\ }\textbf {\bibinfo {volume} {107}},\
  \bibinfo {pages} {6257} (\bibinfo {year} {1997})}\BibitemShut {NoStop}%
\bibitem [{\citenamefont {Monkhorst}(1977)}]{monkhorst1977calculation}%
  \BibitemOpen
  \bibfield  {author} {\bibinfo {author} {\bibfnamefont {H.~J.}\ \bibnamefont
  {Monkhorst}},\ }\bibfield  {title} {\bibinfo {title} {Calculation of
  properties with the coupled-cluster method},\ }\href
  {https://doi.org/10.1002/QUA.560120850} {\bibfield  {journal} {\bibinfo
  {journal} {International Journal of Quantum Chemistry}\ }\textbf {\bibinfo
  {volume} {12}},\ \bibinfo {pages} {421} (\bibinfo {year} {1977})}\BibitemShut
  {NoStop}%
\bibitem [{\citenamefont {Lehtola}\ \emph {et~al.}(2017)\citenamefont
  {Lehtola}, \citenamefont {Tubman}, \citenamefont {Whaley},\ and\
  \citenamefont {Head-Gordon}}]{lehtola2017cluster}%
  \BibitemOpen
  \bibfield  {author} {\bibinfo {author} {\bibfnamefont {S.}~\bibnamefont
  {Lehtola}}, \bibinfo {author} {\bibfnamefont {N.~M.}\ \bibnamefont {Tubman}},
  \bibinfo {author} {\bibfnamefont {K.~B.}\ \bibnamefont {Whaley}},\ and\
  \bibinfo {author} {\bibfnamefont {M.}~\bibnamefont {Head-Gordon}},\
  }\bibfield  {title} {\bibinfo {title} {Cluster decomposition of full
  configuration interaction wave functions: A tool for chemical interpretation
  of systems with strong correlation},\ }\bibfield  {journal} {\bibinfo
  {journal} {The Journal of chemical physics}\ }\textbf {\bibinfo {volume}
  {147}},\ \href {https://doi.org/10.1063/1.4996044} {10.1063/1.4996044}
  (\bibinfo {year} {2017})\BibitemShut {NoStop}%
\bibitem [{\citenamefont {Hino}\ \emph {et~al.}(2006)\citenamefont {Hino},
  \citenamefont {Kinoshita}, \citenamefont {Chan},\ and\ \citenamefont
  {Bartlett}}]{hino2006tailored}%
  \BibitemOpen
  \bibfield  {author} {\bibinfo {author} {\bibfnamefont {O.}~\bibnamefont
  {Hino}}, \bibinfo {author} {\bibfnamefont {T.}~\bibnamefont {Kinoshita}},
  \bibinfo {author} {\bibfnamefont {G.~K.~L.}\ \bibnamefont {Chan}},\ and\
  \bibinfo {author} {\bibfnamefont {R.~J.}\ \bibnamefont {Bartlett}},\
  }\bibfield  {title} {\bibinfo {title} {Tailored coupled cluster singles and
  doubles method applied to calculations on molecular structure and harmonic
  vibrational frequencies of ozone},\ }\bibfield  {journal} {\bibinfo
  {journal} {Journal of Chemical Physics}\ }\textbf {\bibinfo {volume} {124}},\
  \href {https://doi.org/10.1063/1.2180775/186884} {10.1063/1.2180775/186884}
  (\bibinfo {year} {2006})\BibitemShut {NoStop}%
\bibitem [{\citenamefont {Veis}\ \emph {et~al.}(2016)\citenamefont {Veis},
  \citenamefont {Antalík}, \citenamefont {Brabec}, \citenamefont {Neese},
  \citenamefont {Örs Legeza},\ and\ \citenamefont
  {Pittner}}]{veis2016coupled}%
  \BibitemOpen
  \bibfield  {author} {\bibinfo {author} {\bibfnamefont {L.}~\bibnamefont
  {Veis}}, \bibinfo {author} {\bibfnamefont {A.}~\bibnamefont {Antalík}},
  \bibinfo {author} {\bibfnamefont {J.}~\bibnamefont {Brabec}}, \bibinfo
  {author} {\bibfnamefont {F.}~\bibnamefont {Neese}}, \bibinfo {author}
  {\bibnamefont {Örs Legeza}},\ and\ \bibinfo {author} {\bibfnamefont
  {J.}~\bibnamefont {Pittner}},\ }\bibfield  {title} {\bibinfo {title} {Coupled
  cluster method with single and double excitations tailored by matrix product
  state wave functions},\ }\href
  {https://doi.org/10.1021/ACS.JPCLETT.6B01908/ASSET/IMAGES/LARGE/JZ-2016-01908H_0002.JPEG}
  {\bibfield  {journal} {\bibinfo  {journal} {Journal of Physical Chemistry
  Letters}\ }\textbf {\bibinfo {volume} {7}},\ \bibinfo {pages} {4072}
  (\bibinfo {year} {2016})}\BibitemShut {NoStop}%
\bibitem [{\citenamefont {Leszczyk}\ \emph {et~al.}(2022)\citenamefont
  {Leszczyk}, \citenamefont {Máté}, \citenamefont {Örs Legeza},\ and\
  \citenamefont {Boguslawski}}]{leszczyk2022assessing}%
  \BibitemOpen
  \bibfield  {author} {\bibinfo {author} {\bibfnamefont {A.}~\bibnamefont
  {Leszczyk}}, \bibinfo {author} {\bibfnamefont {M.}~\bibnamefont {Máté}},
  \bibinfo {author} {\bibnamefont {Örs Legeza}},\ and\ \bibinfo {author}
  {\bibfnamefont {K.}~\bibnamefont {Boguslawski}},\ }\bibfield  {title}
  {\bibinfo {title} {Assessing the accuracy of tailored coupled cluster methods
  corrected by electronic wave functions of polynomial cost},\ }\href
  {https://doi.org/10.1021/ACS.JCTC.1C00284/ASSET/IMAGES/LARGE/CT1C00284_0010.JPEG}
  {\bibfield  {journal} {\bibinfo  {journal} {Journal of Chemical Theory and
  Computation}\ }\textbf {\bibinfo {volume} {18}},\ \bibinfo {pages} {96}
  (\bibinfo {year} {2022})}\BibitemShut {NoStop}%
\bibitem [{\citenamefont {Ravi}\ \emph {et~al.}(2023)\citenamefont {Ravi},
  \citenamefont {Perera}, \citenamefont {Park},\ and\ \citenamefont
  {Bartlett}}]{ravi2023excited}%
  \BibitemOpen
  \bibfield  {author} {\bibinfo {author} {\bibfnamefont {M.}~\bibnamefont
  {Ravi}}, \bibinfo {author} {\bibfnamefont {A.}~\bibnamefont {Perera}},
  \bibinfo {author} {\bibfnamefont {Y.~C.}\ \bibnamefont {Park}},\ and\
  \bibinfo {author} {\bibfnamefont {R.~J.}\ \bibnamefont {Bartlett}},\
  }\bibfield  {title} {\bibinfo {title} {Excited states with pair coupled
  cluster doubles tailored coupled cluster theory},\ }\bibfield  {journal}
  {\bibinfo  {journal} {The Journal of Chemical Physics}\ }\textbf {\bibinfo
  {volume} {159}},\ \href {https://doi.org/10.1063/5.0161368}
  {10.1063/5.0161368} (\bibinfo {year} {2023})\BibitemShut {NoStop}%
\bibitem [{\citenamefont {Lyakh}\ \emph {et~al.}(2011)\citenamefont {Lyakh},
  \citenamefont {Lotrich},\ and\ \citenamefont {Bartlett}}]{lyakh2011tailored}%
  \BibitemOpen
  \bibfield  {author} {\bibinfo {author} {\bibfnamefont {D.~I.}\ \bibnamefont
  {Lyakh}}, \bibinfo {author} {\bibfnamefont {V.~F.}\ \bibnamefont {Lotrich}},\
  and\ \bibinfo {author} {\bibfnamefont {R.~J.}\ \bibnamefont {Bartlett}},\
  }\bibfield  {title} {\bibinfo {title} {The ‘tailored’ {CCSD(T)}
  description of the automerization of cyclobutadiene},\ }\href
  {https://doi.org/10.1016/J.CPLETT.2010.11.058} {\bibfield  {journal}
  {\bibinfo  {journal} {Chemical Physics Letters}\ }\textbf {\bibinfo {volume}
  {501}},\ \bibinfo {pages} {166} (\bibinfo {year} {2011})}\BibitemShut
  {NoStop}%
\bibitem [{\citenamefont {Antal{\'\i}k}\ \emph {et~al.}(2019)\citenamefont
  {Antal{\'\i}k}, \citenamefont {Veis}, \citenamefont {Brabec}, \citenamefont
  {Demel}, \citenamefont {Legeza},\ and\ \citenamefont
  {Pittner}}]{antalik2019toward}%
  \BibitemOpen
  \bibfield  {author} {\bibinfo {author} {\bibfnamefont {A.}~\bibnamefont
  {Antal{\'\i}k}}, \bibinfo {author} {\bibfnamefont {L.}~\bibnamefont {Veis}},
  \bibinfo {author} {\bibfnamefont {J.}~\bibnamefont {Brabec}}, \bibinfo
  {author} {\bibfnamefont {O.}~\bibnamefont {Demel}}, \bibinfo {author}
  {\bibfnamefont {{\"O}.}~\bibnamefont {Legeza}},\ and\ \bibinfo {author}
  {\bibfnamefont {J.}~\bibnamefont {Pittner}},\ }\bibfield  {title} {\bibinfo
  {title} {Toward the efficient local tailored coupled cluster approximation
  and the peculiar case of {oxo-Mn (Salen)}},\ }\bibfield  {journal} {\bibinfo
  {journal} {The Journal of Chemical Physics}\ }\textbf {\bibinfo {volume}
  {151}},\ \href {https://doi.org/10.1063/1.5110477} {10.1063/1.5110477}
  (\bibinfo {year} {2019})\BibitemShut {NoStop}%
\bibitem [{\citenamefont {Antalík}\ \emph {et~al.}(2020)\citenamefont
  {Antalík}, \citenamefont {Nachtigallová}, \citenamefont {Lo}, \citenamefont
  {Matoušek}, \citenamefont {Lang}, \citenamefont {Örs Legeza}, \citenamefont
  {Pittner}, \citenamefont {Hobza},\ and\ \citenamefont
  {Veis}}]{antalik2020ground}%
  \BibitemOpen
  \bibfield  {author} {\bibinfo {author} {\bibfnamefont {A.}~\bibnamefont
  {Antalík}}, \bibinfo {author} {\bibfnamefont {D.}~\bibnamefont
  {Nachtigallová}}, \bibinfo {author} {\bibfnamefont {R.}~\bibnamefont {Lo}},
  \bibinfo {author} {\bibfnamefont {M.}~\bibnamefont {Matoušek}}, \bibinfo
  {author} {\bibfnamefont {J.}~\bibnamefont {Lang}}, \bibinfo {author}
  {\bibnamefont {Örs Legeza}}, \bibinfo {author} {\bibfnamefont
  {J.}~\bibnamefont {Pittner}}, \bibinfo {author} {\bibfnamefont
  {P.}~\bibnamefont {Hobza}},\ and\ \bibinfo {author} {\bibfnamefont
  {L.}~\bibnamefont {Veis}},\ }\bibfield  {title} {\bibinfo {title} {Ground
  state of the fe(ii)-porphyrin model system corresponds to quintet: a dft and
  dmrg-based tailored cc study},\ }\href {https://doi.org/10.1039/D0CP03086D}
  {\bibfield  {journal} {\bibinfo  {journal} {Physical Chemistry Chemical
  Physics}\ }\textbf {\bibinfo {volume} {22}},\ \bibinfo {pages} {17033}
  (\bibinfo {year} {2020})}\BibitemShut {NoStop}%
\bibitem [{\citenamefont {Lang}\ \emph {et~al.}(2020)\citenamefont {Lang},
  \citenamefont {Antal{\'\i}k}, \citenamefont {Veis}, \citenamefont {Brandejs},
  \citenamefont {Brabec}, \citenamefont {Legeza},\ and\ \citenamefont
  {Pittner}}]{lang2020near}%
  \BibitemOpen
  \bibfield  {author} {\bibinfo {author} {\bibfnamefont {J.}~\bibnamefont
  {Lang}}, \bibinfo {author} {\bibfnamefont {A.}~\bibnamefont {Antal{\'\i}k}},
  \bibinfo {author} {\bibfnamefont {L.}~\bibnamefont {Veis}}, \bibinfo {author}
  {\bibfnamefont {J.}~\bibnamefont {Brandejs}}, \bibinfo {author}
  {\bibfnamefont {J.}~\bibnamefont {Brabec}}, \bibinfo {author} {\bibfnamefont
  {O.}~\bibnamefont {Legeza}},\ and\ \bibinfo {author} {\bibfnamefont
  {J.}~\bibnamefont {Pittner}},\ }\bibfield  {title} {\bibinfo {title}
  {Near-linear scaling in dmrg-based tailored coupled clusters: an
  implementation of {DLPNO-TCCSD} and {DLPNO-TCCSD(T)}},\ }\href
  {https://doi.org/10.1021/acs.jctc.0c00065} {\bibfield  {journal} {\bibinfo
  {journal} {Journal of Chemical Theory and Computation}\ }\textbf {\bibinfo
  {volume} {16}},\ \bibinfo {pages} {3028} (\bibinfo {year}
  {2020})}\BibitemShut {NoStop}%
\bibitem [{\citenamefont {Faulstich}\ \emph
  {et~al.}(2019{\natexlab{a}})\citenamefont {Faulstich}, \citenamefont
  {Laestadius}, \citenamefont {Örs Legeza}, \citenamefont {Schneider},\ and\
  \citenamefont {Kvaal}}]{faulstich2019analysis}%
  \BibitemOpen
  \bibfield  {author} {\bibinfo {author} {\bibfnamefont {F.~M.}\ \bibnamefont
  {Faulstich}}, \bibinfo {author} {\bibfnamefont {A.}~\bibnamefont
  {Laestadius}}, \bibinfo {author} {\bibnamefont {Örs Legeza}}, \bibinfo
  {author} {\bibfnamefont {R.}~\bibnamefont {Schneider}},\ and\ \bibinfo
  {author} {\bibfnamefont {S.}~\bibnamefont {Kvaal}},\ }\bibfield  {title}
  {\bibinfo {title} {Analysis of the tailored coupled-cluster method in quantum
  chemistry},\ }\href {https://doi.org/10.1137/18M1171436} {\bibfield
  {journal} {\bibinfo  {journal} {SIAM Journal on Numerical Analysis}\ }\textbf
  {\bibinfo {volume} {57}},\ \bibinfo {pages} {2579} (\bibinfo {year}
  {2019}{\natexlab{a}})}\BibitemShut {NoStop}%
\bibitem [{\citenamefont {Faulstich}\ \emph
  {et~al.}(2019{\natexlab{b}})\citenamefont {Faulstich}, \citenamefont
  {Máté}, \citenamefont {Laestadius}, \citenamefont {Csirik}, \citenamefont
  {Veis}, \citenamefont {Antalik}, \citenamefont {Brabec}, \citenamefont
  {Schneider}, \citenamefont {Pittner}, \citenamefont {Kvaal},\ and\
  \citenamefont {Örs Legeza}}]{faulstich2019numerical}%
  \BibitemOpen
  \bibfield  {author} {\bibinfo {author} {\bibfnamefont {F.~M.}\ \bibnamefont
  {Faulstich}}, \bibinfo {author} {\bibfnamefont {M.}~\bibnamefont {Máté}},
  \bibinfo {author} {\bibfnamefont {A.}~\bibnamefont {Laestadius}}, \bibinfo
  {author} {\bibfnamefont {M.~A.}\ \bibnamefont {Csirik}}, \bibinfo {author}
  {\bibfnamefont {L.}~\bibnamefont {Veis}}, \bibinfo {author} {\bibfnamefont
  {A.}~\bibnamefont {Antalik}}, \bibinfo {author} {\bibfnamefont
  {J.}~\bibnamefont {Brabec}}, \bibinfo {author} {\bibfnamefont
  {R.}~\bibnamefont {Schneider}}, \bibinfo {author} {\bibfnamefont
  {J.}~\bibnamefont {Pittner}}, \bibinfo {author} {\bibfnamefont
  {S.}~\bibnamefont {Kvaal}},\ and\ \bibinfo {author} {\bibnamefont {Örs
  Legeza}},\ }\bibfield  {title} {\bibinfo {title} {Numerical and theoretical
  aspects of the dmrg-tcc method exemplified by the nitrogen dimer},\ }\href
  {https://doi.org/10.1021/ACS.JCTC.8B00960} {\bibfield  {journal} {\bibinfo
  {journal} {Journal of Chemical Theory and Computation}\ }\textbf {\bibinfo
  {volume} {15}},\ \bibinfo {pages} {2206} (\bibinfo {year}
  {2019}{\natexlab{b}})}\BibitemShut {NoStop}%
\bibitem [{\citenamefont {Melnichuk}\ and\ \citenamefont
  {Bartlett}(2012)}]{melnichuk2012relaxed}%
  \BibitemOpen
  \bibfield  {author} {\bibinfo {author} {\bibfnamefont {A.}~\bibnamefont
  {Melnichuk}}\ and\ \bibinfo {author} {\bibfnamefont {R.~J.}\ \bibnamefont
  {Bartlett}},\ }\bibfield  {title} {\bibinfo {title} {Relaxed active space:
  Fixing {tailored-CC} with high order coupled cluster. i},\ }\bibfield
  {journal} {\bibinfo  {journal} {The Journal of chemical physics}\ }\textbf
  {\bibinfo {volume} {137}},\ \href {https://doi.org/10.1063/1.4767900}
  {10.1063/1.4767900} (\bibinfo {year} {2012})\BibitemShut {NoStop}%
\bibitem [{\citenamefont {Demel}\ \emph {et~al.}(2023)\citenamefont {Demel},
  \citenamefont {Brandejs}, \citenamefont {Lang}, \citenamefont {Brabec},
  \citenamefont {Veis}, \citenamefont {Legeza},\ and\ \citenamefont
  {Pittner}}]{demel2023hilbert}%
  \BibitemOpen
  \bibfield  {author} {\bibinfo {author} {\bibfnamefont {O.}~\bibnamefont
  {Demel}}, \bibinfo {author} {\bibfnamefont {J.}~\bibnamefont {Brandejs}},
  \bibinfo {author} {\bibfnamefont {J.}~\bibnamefont {Lang}}, \bibinfo {author}
  {\bibfnamefont {J.}~\bibnamefont {Brabec}}, \bibinfo {author} {\bibfnamefont
  {L.}~\bibnamefont {Veis}}, \bibinfo {author} {\bibfnamefont {O.}~\bibnamefont
  {Legeza}},\ and\ \bibinfo {author} {\bibfnamefont {J.}~\bibnamefont
  {Pittner}},\ }\bibfield  {title} {\bibinfo {title} {Hilbert space
  multireference coupled clusters tailored by matrix product states},\
  }\href@noop {} {\bibfield  {journal} {\bibinfo  {journal} {arXiv preprint
  arXiv:2304.01625}\ } (\bibinfo {year} {2023})}\BibitemShut {NoStop}%
\bibitem [{\citenamefont {Deustua}\ \emph {et~al.}(2018)\citenamefont
  {Deustua}, \citenamefont {Magoulas}, \citenamefont {Shen},\ and\
  \citenamefont {Piecuch}}]{deustua2018communication}%
  \BibitemOpen
  \bibfield  {author} {\bibinfo {author} {\bibfnamefont {J.~E.}\ \bibnamefont
  {Deustua}}, \bibinfo {author} {\bibfnamefont {I.}~\bibnamefont {Magoulas}},
  \bibinfo {author} {\bibfnamefont {J.}~\bibnamefont {Shen}},\ and\ \bibinfo
  {author} {\bibfnamefont {P.}~\bibnamefont {Piecuch}},\ }\bibfield  {title}
  {\bibinfo {title} {Communication: Approaching exact quantum chemistry by
  cluster analysis of full configuration interaction quantum monte carlo wave
  functions},\ }\href {https://doi.org/10.1063/1.5055769} {\bibfield  {journal}
  {\bibinfo  {journal} {The Journal of Chemical Physics}\ }\textbf {\bibinfo
  {volume} {149}},\ \bibinfo {pages} {151101} (\bibinfo {year}
  {2018})}\BibitemShut {NoStop}%
\bibitem [{\citenamefont {Deustua}\ \emph {et~al.}(2017)\citenamefont
  {Deustua}, \citenamefont {Shen},\ and\ \citenamefont
  {Piecuch}}]{PhysRevLett.119.223003}%
  \BibitemOpen
  \bibfield  {author} {\bibinfo {author} {\bibfnamefont {J.~E.}\ \bibnamefont
  {Deustua}}, \bibinfo {author} {\bibfnamefont {J.}~\bibnamefont {Shen}},\ and\
  \bibinfo {author} {\bibfnamefont {P.}~\bibnamefont {Piecuch}},\ }\bibfield
  {title} {\bibinfo {title} {Converging high-level coupled-cluster energetics
  by monte carlo sampling and moment expansions},\ }\href
  {https://doi.org/10.1103/PhysRevLett.119.223003} {\bibfield  {journal}
  {\bibinfo  {journal} {Phys. Rev. Lett.}\ }\textbf {\bibinfo {volume} {119}},\
  \bibinfo {pages} {223003} (\bibinfo {year} {2017})}\BibitemShut {NoStop}%
\bibitem [{\citenamefont {Magoulas}\ \emph {et~al.}(2021)\citenamefont
  {Magoulas}, \citenamefont {Gururangan}, \citenamefont {Piecuch},
  \citenamefont {Deustua},\ and\ \citenamefont {Shen}}]{magoulas2021is}%
  \BibitemOpen
  \bibfield  {author} {\bibinfo {author} {\bibfnamefont {I.}~\bibnamefont
  {Magoulas}}, \bibinfo {author} {\bibfnamefont {K.}~\bibnamefont
  {Gururangan}}, \bibinfo {author} {\bibfnamefont {P.}~\bibnamefont {Piecuch}},
  \bibinfo {author} {\bibfnamefont {J.~E.}\ \bibnamefont {Deustua}},\ and\
  \bibinfo {author} {\bibfnamefont {J.}~\bibnamefont {Shen}},\ }\bibfield
  {title} {\bibinfo {title} {Is externally corrected coupled cluster always
  better than the underlying truncated configuration interaction?},\ }\href
  {https://doi.org/https://doi.org/10.1021/acs.jctc.1c00181} {\bibfield
  {journal} {\bibinfo  {journal} {Journal of Chemical Theory and Computation}\
  }\textbf {\bibinfo {volume} {17}},\ \bibinfo {pages} {4006} (\bibinfo {year}
  {2021})}\BibitemShut {NoStop}%
\bibitem [{\citenamefont {Lee}\ \emph {et~al.}(2021)\citenamefont {Lee},
  \citenamefont {Zhai}, \citenamefont {Sharma}, \citenamefont {Umrigar},\ and\
  \citenamefont {Chan}}]{lee2021externally}%
  \BibitemOpen
  \bibfield  {author} {\bibinfo {author} {\bibfnamefont {S.}~\bibnamefont
  {Lee}}, \bibinfo {author} {\bibfnamefont {H.}~\bibnamefont {Zhai}}, \bibinfo
  {author} {\bibfnamefont {S.}~\bibnamefont {Sharma}}, \bibinfo {author}
  {\bibfnamefont {C.~J.}\ \bibnamefont {Umrigar}},\ and\ \bibinfo {author}
  {\bibfnamefont {G.~K.-L.}\ \bibnamefont {Chan}},\ }\bibfield  {title}
  {\bibinfo {title} {Externally corrected ccsd with renormalized perturbative
  triples {(R-ecCCSD(T))} and the density matrix renormalization group and
  selected configuration interaction external sources},\ }\href
  {https://doi.org/10.1021/acs.jctc.1c00205} {\bibfield  {journal} {\bibinfo
  {journal} {Journal of Chemical Theory and Computation}\ }\textbf {\bibinfo
  {volume} {17}},\ \bibinfo {pages} {3414} (\bibinfo {year}
  {2021})}\BibitemShut {NoStop}%
\bibitem [{\citenamefont {Aroeira}\ \emph {et~al.}(2020)\citenamefont
  {Aroeira}, \citenamefont {Davis}, \citenamefont {Turney},\ and\ \citenamefont
  {Schaefer~III}}]{aroeira2020coupled}%
  \BibitemOpen
  \bibfield  {author} {\bibinfo {author} {\bibfnamefont {G.~J.}\ \bibnamefont
  {Aroeira}}, \bibinfo {author} {\bibfnamefont {M.~M.}\ \bibnamefont {Davis}},
  \bibinfo {author} {\bibfnamefont {J.~M.}\ \bibnamefont {Turney}},\ and\
  \bibinfo {author} {\bibfnamefont {H.~F.}\ \bibnamefont {Schaefer~III}},\
  }\bibfield  {title} {\bibinfo {title} {Coupled cluster externally corrected
  by adaptive configuration interaction},\ }\href
  {https://pubs.acs.org/doi/abs/10.1021/acs.jctc.0c00888} {\bibfield  {journal}
  {\bibinfo  {journal} {Journal of chemical theory and computation}\ }\textbf
  {\bibinfo {volume} {17}},\ \bibinfo {pages} {182} (\bibinfo {year}
  {2020})}\BibitemShut {NoStop}%
\bibitem [{\citenamefont {Anselmetti}\ \emph {et~al.}(2021)\citenamefont
  {Anselmetti}, \citenamefont {Wierichs}, \citenamefont {Gogolin},\ and\
  \citenamefont {Parrish}}]{anselmetti2021local}%
  \BibitemOpen
  \bibfield  {author} {\bibinfo {author} {\bibfnamefont {G.-L.~R.}\
  \bibnamefont {Anselmetti}}, \bibinfo {author} {\bibfnamefont
  {D.}~\bibnamefont {Wierichs}}, \bibinfo {author} {\bibfnamefont
  {C.}~\bibnamefont {Gogolin}},\ and\ \bibinfo {author} {\bibfnamefont {R.~M.}\
  \bibnamefont {Parrish}},\ }\bibfield  {title} {\bibinfo {title} {Local,
  expressive, quantum-number-preserving vqe ansätze for fermionic systems},\
  }\href {https://doi.org/10.1088/1367-2630/ac2cb3} {\bibfield  {journal}
  {\bibinfo  {journal} {New Journal of Physics}\ }\textbf {\bibinfo {volume}
  {23}},\ \bibinfo {pages} {113010} (\bibinfo {year} {2021})}\BibitemShut
  {NoStop}%
\bibitem [{\citenamefont {Dunning}(1989)}]{dunning1989a}%
  \BibitemOpen
  \bibfield  {author} {\bibinfo {author} {\bibfnamefont {T.~H.}\ \bibnamefont
  {Dunning}},\ }\bibfield  {title} {\bibinfo {title} {Gaussian basis sets for
  use in correlated molecular calculations. i. the atoms boron through neon and
  hydrogen},\ }\href {https://doi.org/10.1063/1.456153} {\bibfield  {journal}
  {\bibinfo  {journal} {J. Chem. Phys.}\ }\textbf {\bibinfo {volume} {90}},\
  \bibinfo {pages} {1007} (\bibinfo {year} {1989})}\BibitemShut {NoStop}%
\bibitem [{\citenamefont {Gard}\ \emph {et~al.}(2020)\citenamefont {Gard},
  \citenamefont {Zhu}, \citenamefont {Barron}, \citenamefont {Mayhall},
  \citenamefont {Economou},\ and\ \citenamefont {Barnes}}]{gard2020efficient}%
  \BibitemOpen
  \bibfield  {author} {\bibinfo {author} {\bibfnamefont {B.~T.}\ \bibnamefont
  {Gard}}, \bibinfo {author} {\bibfnamefont {L.}~\bibnamefont {Zhu}}, \bibinfo
  {author} {\bibfnamefont {G.~S.}\ \bibnamefont {Barron}}, \bibinfo {author}
  {\bibfnamefont {N.~J.}\ \bibnamefont {Mayhall}}, \bibinfo {author}
  {\bibfnamefont {S.~E.}\ \bibnamefont {Economou}},\ and\ \bibinfo {author}
  {\bibfnamefont {E.}~\bibnamefont {Barnes}},\ }\bibfield  {title} {\bibinfo
  {title} {Efficient symmetry-preserving state preparation circuits for the
  variational quantum eigensolver algorithm},\ }\bibfield  {journal} {\bibinfo
  {journal} {npj Quantum Information}\ }\textbf {\bibinfo {volume} {6}},\ \href
  {https://doi.org/10.1038/s41534-019-0240-1} {10.1038/s41534-019-0240-1}
  (\bibinfo {year} {2020})\BibitemShut {NoStop}%
\bibitem [{\citenamefont {Grimsley}\ \emph {et~al.}(2019)\citenamefont
  {Grimsley}, \citenamefont {Economou}, \citenamefont {Barnes},\ and\
  \citenamefont {Mayhall}}]{grimsley2019adaptive}%
  \BibitemOpen
  \bibfield  {author} {\bibinfo {author} {\bibfnamefont {H.~R.}\ \bibnamefont
  {Grimsley}}, \bibinfo {author} {\bibfnamefont {S.~E.}\ \bibnamefont
  {Economou}}, \bibinfo {author} {\bibfnamefont {E.}~\bibnamefont {Barnes}},\
  and\ \bibinfo {author} {\bibfnamefont {N.~J.}\ \bibnamefont {Mayhall}},\
  }\bibfield  {title} {\bibinfo {title} {An adaptive variational algorithm for
  exact molecular simulations on a quantum computer},\ }\bibfield  {journal}
  {\bibinfo  {journal} {Nature Communications}\ }\textbf {\bibinfo {volume}
  {10}},\ \href {https://doi.org/10.1038/s41467-019-10988-2}
  {10.1038/s41467-019-10988-2} (\bibinfo {year} {2019})\BibitemShut {NoStop}%
\bibitem [{\citenamefont {Tang}\ \emph {et~al.}(2021)\citenamefont {Tang},
  \citenamefont {Shkolnikov}, \citenamefont {Barron}, \citenamefont {Grimsley},
  \citenamefont {Mayhall}, \citenamefont {Barnes},\ and\ \citenamefont
  {Economou}}]{PRXQuantum.2.020310}%
  \BibitemOpen
  \bibfield  {author} {\bibinfo {author} {\bibfnamefont {H.~L.}\ \bibnamefont
  {Tang}}, \bibinfo {author} {\bibfnamefont {V.}~\bibnamefont {Shkolnikov}},
  \bibinfo {author} {\bibfnamefont {G.~S.}\ \bibnamefont {Barron}}, \bibinfo
  {author} {\bibfnamefont {H.~R.}\ \bibnamefont {Grimsley}}, \bibinfo {author}
  {\bibfnamefont {N.~J.}\ \bibnamefont {Mayhall}}, \bibinfo {author}
  {\bibfnamefont {E.}~\bibnamefont {Barnes}},\ and\ \bibinfo {author}
  {\bibfnamefont {S.~E.}\ \bibnamefont {Economou}},\ }\bibfield  {title}
  {\bibinfo {title} {Qubit-adapt-vqe: An adaptive algorithm for constructing
  hardware-efficient ans\"atze on a quantum processor},\ }\href
  {https://doi.org/10.1103/PRXQuantum.2.020310} {\bibfield  {journal} {\bibinfo
   {journal} {PRX Quantum}\ }\textbf {\bibinfo {volume} {2}},\ \bibinfo {pages}
  {020310} (\bibinfo {year} {2021})}\BibitemShut {NoStop}%
\bibitem [{\citenamefont {Stair}\ and\ \citenamefont
  {Evangelista}(2021)}]{PRXQuantum.2.030301}%
  \BibitemOpen
  \bibfield  {author} {\bibinfo {author} {\bibfnamefont {N.~H.}\ \bibnamefont
  {Stair}}\ and\ \bibinfo {author} {\bibfnamefont {F.~A.}\ \bibnamefont
  {Evangelista}},\ }\bibfield  {title} {\bibinfo {title} {Simulating many-body
  systems with a projective quantum eigensolver},\ }\href
  {https://doi.org/10.1103/PRXQuantum.2.030301} {\bibfield  {journal} {\bibinfo
   {journal} {PRX Quantum}\ }\textbf {\bibinfo {volume} {2}},\ \bibinfo {pages}
  {030301} (\bibinfo {year} {2021})}\BibitemShut {NoStop}%
\bibitem [{\citenamefont {Lee}\ \emph {et~al.}(2018)\citenamefont {Lee},
  \citenamefont {Huggins}, \citenamefont {Head-Gordon},\ and\ \citenamefont
  {Whaley}}]{lee2018generalized}%
  \BibitemOpen
  \bibfield  {author} {\bibinfo {author} {\bibfnamefont {J.}~\bibnamefont
  {Lee}}, \bibinfo {author} {\bibfnamefont {W.~J.}\ \bibnamefont {Huggins}},
  \bibinfo {author} {\bibfnamefont {M.}~\bibnamefont {Head-Gordon}},\ and\
  \bibinfo {author} {\bibfnamefont {K.~B.}\ \bibnamefont {Whaley}},\ }\bibfield
   {title} {\bibinfo {title} {Generalized unitary coupled cluster wave
  functions for quantum computation},\ }\href
  {https://doi.org/10.1021/acs.jctc.8b01004} {\bibfield  {journal} {\bibinfo
  {journal} {Journal of Chemical Theory and Computation}\ }\textbf {\bibinfo
  {volume} {15}},\ \bibinfo {pages} {311–324} (\bibinfo {year}
  {2018})}\BibitemShut {NoStop}%
\bibitem [{\citenamefont {O'Gorman}\ \emph {et~al.}(2019)\citenamefont
  {O'Gorman}, \citenamefont {Huggins}, \citenamefont {Rieffel},\ and\
  \citenamefont {Whaley}}]{o2019generalized}%
  \BibitemOpen
  \bibfield  {author} {\bibinfo {author} {\bibfnamefont {B.}~\bibnamefont
  {O'Gorman}}, \bibinfo {author} {\bibfnamefont {W.~J.}\ \bibnamefont
  {Huggins}}, \bibinfo {author} {\bibfnamefont {E.~G.}\ \bibnamefont
  {Rieffel}},\ and\ \bibinfo {author} {\bibfnamefont {K.~B.}\ \bibnamefont
  {Whaley}},\ }\bibfield  {title} {\bibinfo {title} {Generalized swap networks
  for near-term quantum computing},\ }\href {https://arxiv.org/abs/1905.05118}
  {\bibfield  {journal} {\bibinfo  {journal} {arXiv preprint arXiv:1905.05118}\
  } (\bibinfo {year} {2019})}\BibitemShut {NoStop}%
\bibitem [{\citenamefont {Matsuzawa}\ and\ \citenamefont
  {Kurashige}(2020)}]{matsuzawa2020jastrow}%
  \BibitemOpen
  \bibfield  {author} {\bibinfo {author} {\bibfnamefont {Y.}~\bibnamefont
  {Matsuzawa}}\ and\ \bibinfo {author} {\bibfnamefont {Y.}~\bibnamefont
  {Kurashige}},\ }\bibfield  {title} {\bibinfo {title} {Jastrow-type
  decomposition in quantum chemistry for low-depth quantum circuits},\ }\href
  {https://doi.org/10.1021/acs.jctc.9b00963} {\bibfield  {journal} {\bibinfo
  {journal} {Journal of Chemical Theory and Computation}\ }\textbf {\bibinfo
  {volume} {16}},\ \bibinfo {pages} {944–952} (\bibinfo {year}
  {2020})}\BibitemShut {NoStop}%
\bibitem [{\citenamefont {Motta}\ \emph {et~al.}(2019)\citenamefont {Motta},
  \citenamefont {Sun}, \citenamefont {Tan}, \citenamefont {O’Rourke},
  \citenamefont {Ye}, \citenamefont {Minnich}, \citenamefont {Brandão},\ and\
  \citenamefont {Chan}}]{motta2020determining}%
  \BibitemOpen
  \bibfield  {author} {\bibinfo {author} {\bibfnamefont {M.}~\bibnamefont
  {Motta}}, \bibinfo {author} {\bibfnamefont {C.}~\bibnamefont {Sun}}, \bibinfo
  {author} {\bibfnamefont {A.~T.~K.}\ \bibnamefont {Tan}}, \bibinfo {author}
  {\bibfnamefont {M.~J.}\ \bibnamefont {O’Rourke}}, \bibinfo {author}
  {\bibfnamefont {E.}~\bibnamefont {Ye}}, \bibinfo {author} {\bibfnamefont
  {A.~J.}\ \bibnamefont {Minnich}}, \bibinfo {author} {\bibfnamefont {F.~G.
  S.~L.}\ \bibnamefont {Brandão}},\ and\ \bibinfo {author} {\bibfnamefont
  {G.~K.-L.}\ \bibnamefont {Chan}},\ }\bibfield  {title} {\bibinfo {title}
  {Determining eigenstates and thermal states on a quantum computer using
  quantum imaginary time evolution},\ }\href
  {https://doi.org/10.1038/s41567-019-0704-4} {\bibfield  {journal} {\bibinfo
  {journal} {Nature Physics}\ }\textbf {\bibinfo {volume} {16}},\ \bibinfo
  {pages} {205–210} (\bibinfo {year} {2019})}\BibitemShut {NoStop}%
\bibitem [{\citenamefont {Kim}\ and\ \citenamefont
  {Swingle}(2017)}]{kim2017robust}%
  \BibitemOpen
  \bibfield  {author} {\bibinfo {author} {\bibfnamefont {I.~H.}\ \bibnamefont
  {Kim}}\ and\ \bibinfo {author} {\bibfnamefont {B.}~\bibnamefont {Swingle}},\
  }\bibfield  {title} {\bibinfo {title} {Robust entanglement renormalization on
  a noisy quantum computer},\ }\href {https://arxiv.org/abs/1711.07500}
  {\bibfield  {journal} {\bibinfo  {journal} {arXiv preprint arXiv:1711.07500}\
  } (\bibinfo {year} {2017})}\BibitemShut {NoStop}%
\bibitem [{\citenamefont {Sewell}\ and\ \citenamefont
  {Jordan}(2021)}]{sewell2021preparing}%
  \BibitemOpen
  \bibfield  {author} {\bibinfo {author} {\bibfnamefont {T.~J.}\ \bibnamefont
  {Sewell}}\ and\ \bibinfo {author} {\bibfnamefont {S.~P.}\ \bibnamefont
  {Jordan}},\ }\bibfield  {title} {\bibinfo {title} {Preparing renormalization
  group fixed points on nisq hardware},\ }\href
  {https://arxiv.org/abs/2109.09787} {\bibfield  {journal} {\bibinfo  {journal}
  {arXiv preprint arXiv:2109.09787}\ } (\bibinfo {year} {2021})}\BibitemShut
  {NoStop}%
\bibitem [{\citenamefont {Pineda}\ \emph {et~al.}(2010)\citenamefont {Pineda},
  \citenamefont {Barthel},\ and\ \citenamefont {Eisert}}]{PhysRevA.81.050303}%
  \BibitemOpen
  \bibfield  {author} {\bibinfo {author} {\bibfnamefont {C.}~\bibnamefont
  {Pineda}}, \bibinfo {author} {\bibfnamefont {T.}~\bibnamefont {Barthel}},\
  and\ \bibinfo {author} {\bibfnamefont {J.}~\bibnamefont {Eisert}},\
  }\bibfield  {title} {\bibinfo {title} {Unitary circuits for strongly
  correlated fermions},\ }\href {https://doi.org/10.1103/PhysRevA.81.050303}
  {\bibfield  {journal} {\bibinfo  {journal} {Phys. Rev. A}\ }\textbf {\bibinfo
  {volume} {81}},\ \bibinfo {pages} {050303} (\bibinfo {year}
  {2010})}\BibitemShut {NoStop}%
\bibitem [{\citenamefont {Miao}\ and\ \citenamefont
  {Barthel}(2023)}]{miao2021quantum}%
  \BibitemOpen
  \bibfield  {author} {\bibinfo {author} {\bibfnamefont {Q.}~\bibnamefont
  {Miao}}\ and\ \bibinfo {author} {\bibfnamefont {T.}~\bibnamefont {Barthel}},\
  }\bibfield  {title} {\bibinfo {title} {Quantum-classical eigensolver using
  multiscale entanglement renormalization},\ }\bibfield  {journal} {\bibinfo
  {journal} {Physical Review Research}\ }\textbf {\bibinfo {volume} {5}},\
  \href {https://doi.org/10.1103/physrevresearch.5.033141}
  {10.1103/physrevresearch.5.033141} (\bibinfo {year} {2023})\BibitemShut
  {NoStop}%
\bibitem [{\citenamefont {Lin}\ and\ \citenamefont {Tong}(2020)}]{lin2020near}%
  \BibitemOpen
  \bibfield  {author} {\bibinfo {author} {\bibfnamefont {L.}~\bibnamefont
  {Lin}}\ and\ \bibinfo {author} {\bibfnamefont {Y.}~\bibnamefont {Tong}},\
  }\bibfield  {title} {\bibinfo {title} {Near-optimal ground state
  preparation},\ }\href {https://doi.org/10.22331/q-2020-12-14-372} {\bibfield
  {journal} {\bibinfo  {journal} {Quantum}\ }\textbf {\bibinfo {volume} {4}},\
  \bibinfo {pages} {372} (\bibinfo {year} {2020})}\BibitemShut {NoStop}%
\bibitem [{\citenamefont {Malz}\ \emph {et~al.}(2023)\citenamefont {Malz},
  \citenamefont {Styliaris}, \citenamefont {Wei},\ and\ \citenamefont
  {Cirac}}]{malz2023preparation}%
  \BibitemOpen
  \bibfield  {author} {\bibinfo {author} {\bibfnamefont {D.}~\bibnamefont
  {Malz}}, \bibinfo {author} {\bibfnamefont {G.}~\bibnamefont {Styliaris}},
  \bibinfo {author} {\bibfnamefont {Z.-Y.}\ \bibnamefont {Wei}},\ and\ \bibinfo
  {author} {\bibfnamefont {J.~I.}\ \bibnamefont {Cirac}},\ }\bibfield  {title}
  {\bibinfo {title} {Preparation of matrix product states with log-depth
  quantum circuits},\ }\href@noop {} {\bibfield  {journal} {\bibinfo  {journal}
  {arXiv preprint arXiv:2307.01696}\ } (\bibinfo {year} {2023})}\BibitemShut
  {NoStop}%
\bibitem [{\citenamefont {Ge}\ \emph {et~al.}(2019)\citenamefont {Ge},
  \citenamefont {Tura},\ and\ \citenamefont {Cirac}}]{ge2019faster}%
  \BibitemOpen
  \bibfield  {author} {\bibinfo {author} {\bibfnamefont {Y.}~\bibnamefont
  {Ge}}, \bibinfo {author} {\bibfnamefont {J.}~\bibnamefont {Tura}},\ and\
  \bibinfo {author} {\bibfnamefont {J.~I.}\ \bibnamefont {Cirac}},\ }\bibfield
  {title} {\bibinfo {title} {Faster ground state preparation and high-precision
  ground energy estimation with fewer qubits},\ }\href
  {https://pubs.aip.org/aip/jmp/article-abstract/60/2/022202/918222/Faster-ground-state-preparation-and-high-precision?redirectedFrom=fulltext}
  {\bibfield  {journal} {\bibinfo  {journal} {Journal of Mathematical Physics}\
  }\textbf {\bibinfo {volume} {60}} (\bibinfo {year} {2019})}\BibitemShut
  {NoStop}%
\bibitem [{\citenamefont {Lu}\ \emph {et~al.}(2021)\citenamefont {Lu},
  \citenamefont {Ba\~nuls},\ and\ \citenamefont {Cirac}}]{PRXQuantum.2.020321}%
  \BibitemOpen
  \bibfield  {author} {\bibinfo {author} {\bibfnamefont {S.}~\bibnamefont
  {Lu}}, \bibinfo {author} {\bibfnamefont {M.~C.}\ \bibnamefont {Ba\~nuls}},\
  and\ \bibinfo {author} {\bibfnamefont {J.~I.}\ \bibnamefont {Cirac}},\
  }\bibfield  {title} {\bibinfo {title} {Algorithms for quantum simulation at
  finite energies},\ }\href {https://doi.org/10.1103/PRXQuantum.2.020321}
  {\bibfield  {journal} {\bibinfo  {journal} {PRX Quantum}\ }\textbf {\bibinfo
  {volume} {2}},\ \bibinfo {pages} {020321} (\bibinfo {year}
  {2021})}\BibitemShut {NoStop}%
\bibitem [{\citenamefont {He}\ \emph {et~al.}(2022{\natexlab{b}})\citenamefont
  {He}, \citenamefont {Zhang},\ and\ \citenamefont {Wang}}]{he2022quantum}%
  \BibitemOpen
  \bibfield  {author} {\bibinfo {author} {\bibfnamefont {M.-Q.}\ \bibnamefont
  {He}}, \bibinfo {author} {\bibfnamefont {D.-B.}\ \bibnamefont {Zhang}},\ and\
  \bibinfo {author} {\bibfnamefont {Z.~D.}\ \bibnamefont {Wang}},\ }\bibfield
  {title} {\bibinfo {title} {Quantum gaussian filter for exploring ground-state
  properties},\ }\bibfield  {journal} {\bibinfo  {journal} {Physical Review A}\
  }\textbf {\bibinfo {volume} {106}},\ \href
  {https://doi.org/10.1103/physreva.106.032420} {10.1103/physreva.106.032420}
  (\bibinfo {year} {2022}{\natexlab{b}})\BibitemShut {NoStop}%
\bibitem [{\citenamefont {Kyriienko}(2020)}]{kyriienko2020quantum}%
  \BibitemOpen
  \bibfield  {author} {\bibinfo {author} {\bibfnamefont {O.}~\bibnamefont
  {Kyriienko}},\ }\bibfield  {title} {\bibinfo {title} {Quantum inverse
  iteration algorithm for programmable quantum simulators},\ }\bibfield
  {journal} {\bibinfo  {journal} {npj Quantum Information}\ }\textbf {\bibinfo
  {volume} {6}},\ \href {https://doi.org/10.1038/s41534-019-0239-7}
  {10.1038/s41534-019-0239-7} (\bibinfo {year} {2020})\BibitemShut {NoStop}%
\bibitem [{\citenamefont {Fomichev}\ \emph {et~al.}(2023)\citenamefont
  {Fomichev}, \citenamefont {Hejazi}, \citenamefont {Zini}, \citenamefont
  {Kiser}, \citenamefont {Morales}, \citenamefont {Casares}, \citenamefont
  {Delgado}, \citenamefont {Huh}, \citenamefont {Voigt}, \citenamefont
  {Mueller} \emph {et~al.}}]{fomichev2023initial}%
  \BibitemOpen
  \bibfield  {author} {\bibinfo {author} {\bibfnamefont {S.}~\bibnamefont
  {Fomichev}}, \bibinfo {author} {\bibfnamefont {K.}~\bibnamefont {Hejazi}},
  \bibinfo {author} {\bibfnamefont {M.~S.}\ \bibnamefont {Zini}}, \bibinfo
  {author} {\bibfnamefont {M.}~\bibnamefont {Kiser}}, \bibinfo {author}
  {\bibfnamefont {J.~F.}\ \bibnamefont {Morales}}, \bibinfo {author}
  {\bibfnamefont {P.~A.~M.}\ \bibnamefont {Casares}}, \bibinfo {author}
  {\bibfnamefont {A.}~\bibnamefont {Delgado}}, \bibinfo {author} {\bibfnamefont
  {J.}~\bibnamefont {Huh}}, \bibinfo {author} {\bibfnamefont {A.-C.}\
  \bibnamefont {Voigt}}, \bibinfo {author} {\bibfnamefont {J.~E.}\ \bibnamefont
  {Mueller}}, \emph {et~al.},\ }\bibfield  {title} {\bibinfo {title} {Initial
  state preparation for quantum chemistry on quantum computers},\ }\href
  {https://arxiv.org/abs/2310.18410} {\bibfield  {journal} {\bibinfo  {journal}
  {arXiv preprint arXiv:2310.18410}\ } (\bibinfo {year} {2023})}\BibitemShut
  {NoStop}%
\bibitem [{\citenamefont {Choi}\ \emph
  {et~al.}(2023{\natexlab{b}})\citenamefont {Choi}, \citenamefont {Loaiza},\
  and\ \citenamefont {Izmaylov}}]{choi2023}%
  \BibitemOpen
  \bibfield  {author} {\bibinfo {author} {\bibfnamefont {S.}~\bibnamefont
  {Choi}}, \bibinfo {author} {\bibfnamefont {I.}~\bibnamefont {Loaiza}},\ and\
  \bibinfo {author} {\bibfnamefont {A.~F.}\ \bibnamefont {Izmaylov}},\
  }\bibfield  {title} {\bibinfo {title} {Fluid fermionic fragments for
  optimizing quantum measurements of electronic hamiltonians in the variational
  quantum eigensolver},\ }\href {https://doi.org/10.22331/q-2023-01-03-889}
  {\bibfield  {journal} {\bibinfo  {journal} {Quantum}\ }\textbf {\bibinfo
  {volume} {7}},\ \bibinfo {pages} {889} (\bibinfo {year}
  {2023}{\natexlab{b}})}\BibitemShut {NoStop}%
\bibitem [{\citenamefont {O’Brien}\ \emph {et~al.}(2023)\citenamefont
  {O’Brien}, \citenamefont {Anselmetti}, \citenamefont {Gkritsis},
  \citenamefont {Elfving}, \citenamefont {Polla}, \citenamefont {Huggins},
  \citenamefont {Oumarou}, \citenamefont {Kechedzhi}, \citenamefont {Abanin},
  \citenamefont {Acharya}, \citenamefont {Aleiner}, \citenamefont {Allen},
  \citenamefont {Andersen}, \citenamefont {Anderson}, \citenamefont {Ansmann},
  \citenamefont {Arute}, \citenamefont {Arya}, \citenamefont {Asfaw},
  \citenamefont {Atalaya}, \citenamefont {Bardin}, \citenamefont {Bengtsson},
  \citenamefont {Bortoli}, \citenamefont {Bourassa}, \citenamefont {Bovaird},
  \citenamefont {Brill}, \citenamefont {Broughton}, \citenamefont {Buckley},
  \citenamefont {Buell}, \citenamefont {Burger}, \citenamefont {Burkett},
  \citenamefont {Bushnell}, \citenamefont {Campero}, \citenamefont {Chen},
  \citenamefont {Chiaro}, \citenamefont {Chik}, \citenamefont {Cogan},
  \citenamefont {Collins}, \citenamefont {Conner}, \citenamefont {Courtney},
  \citenamefont {Crook}, \citenamefont {Curtin}, \citenamefont {Debroy},
  \citenamefont {Demura}, \citenamefont {Drozdov}, \citenamefont {Dunsworth},
  \citenamefont {Erickson}, \citenamefont {Faoro}, \citenamefont {Farhi},
  \citenamefont {Fatemi}, \citenamefont {Ferreira}, \citenamefont
  {Flores~Burgos}, \citenamefont {Forati}, \citenamefont {Fowler},
  \citenamefont {Foxen}, \citenamefont {Giang}, \citenamefont {Gidney},
  \citenamefont {Gilboa}, \citenamefont {Giustina}, \citenamefont {Gosula},
  \citenamefont {Grajales~Dau}, \citenamefont {Gross}, \citenamefont
  {Habegger}, \citenamefont {Hamilton}, \citenamefont {Hansen}, \citenamefont
  {Harrigan}, \citenamefont {Harrington}, \citenamefont {Heu}, \citenamefont
  {Hoffmann}, \citenamefont {Hong}, \citenamefont {Huang}, \citenamefont
  {Huff}, \citenamefont {Ioffe}, \citenamefont {Isakov}, \citenamefont
  {Iveland}, \citenamefont {Jeffrey}, \citenamefont {Jiang}, \citenamefont
  {Jones}, \citenamefont {Juhas}, \citenamefont {Kafri}, \citenamefont
  {Khattar}, \citenamefont {Khezri}, \citenamefont {Kieferová}, \citenamefont
  {Kim}, \citenamefont {Klimov}, \citenamefont {Klots}, \citenamefont
  {Korotkov}, \citenamefont {Kostritsa}, \citenamefont {Kreikebaum},
  \citenamefont {Landhuis}, \citenamefont {Laptev}, \citenamefont {Lau},
  \citenamefont {Laws}, \citenamefont {Lee}, \citenamefont {Lee}, \citenamefont
  {Lester}, \citenamefont {Lill}, \citenamefont {Liu}, \citenamefont
  {Livingston}, \citenamefont {Locharla}, \citenamefont {Malone}, \citenamefont
  {Mandrà}, \citenamefont {Martin}, \citenamefont {Martin}, \citenamefont
  {McClean}, \citenamefont {McCourt}, \citenamefont {McEwen}, \citenamefont
  {Mi}, \citenamefont {Mieszala}, \citenamefont {Miao}, \citenamefont
  {Mohseni}, \citenamefont {Montazeri}, \citenamefont {Morvan}, \citenamefont
  {Movassagh}, \citenamefont {Mruczkiewicz}, \citenamefont {Naaman},
  \citenamefont {Neeley}, \citenamefont {Neill}, \citenamefont {Nersisyan},
  \citenamefont {Newman}, \citenamefont {Ng}, \citenamefont {Nguyen},
  \citenamefont {Nguyen}, \citenamefont {Niu}, \citenamefont {Omonije},
  \citenamefont {Opremcak}, \citenamefont {Petukhov}, \citenamefont {Potter},
  \citenamefont {Pryadko}, \citenamefont {Quintana}, \citenamefont {Rocque},
  \citenamefont {Roushan}, \citenamefont {Saei}, \citenamefont {Sank},
  \citenamefont {Sankaragomathi}, \citenamefont {Satzinger}, \citenamefont
  {Schurkus}, \citenamefont {Schuster}, \citenamefont {Shearn}, \citenamefont
  {Shorter}, \citenamefont {Shutty}, \citenamefont {Shvarts}, \citenamefont
  {Skruzny}, \citenamefont {Smith}, \citenamefont {Somma}, \citenamefont
  {Sterling}, \citenamefont {Strain}, \citenamefont {Szalay}, \citenamefont
  {Thor}, \citenamefont {Torres}, \citenamefont {Vidal}, \citenamefont
  {Villalonga}, \citenamefont {Vollgraff~Heidweiller}, \citenamefont {White},
  \citenamefont {Woo}, \citenamefont {Xing}, \citenamefont {Yao}, \citenamefont
  {Yeh}, \citenamefont {Yoo}, \citenamefont {Young}, \citenamefont {Zalcman},
  \citenamefont {Zhang}, \citenamefont {Zhu}, \citenamefont {Zobrist},
  \citenamefont {Bacon}, \citenamefont {Boixo}, \citenamefont {Chen},
  \citenamefont {Hilton}, \citenamefont {Kelly}, \citenamefont {Lucero},
  \citenamefont {Megrant}, \citenamefont {Neven}, \citenamefont {Smelyanskiy},
  \citenamefont {Gogolin}, \citenamefont {Babbush},\ and\ \citenamefont
  {Rubin}}]{obrien2023}%
  \BibitemOpen
  \bibfield  {author} {\bibinfo {author} {\bibfnamefont {T.~E.}\ \bibnamefont
  {O’Brien}}, \bibinfo {author} {\bibfnamefont {G.}~\bibnamefont
  {Anselmetti}}, \bibinfo {author} {\bibfnamefont {F.}~\bibnamefont
  {Gkritsis}}, \bibinfo {author} {\bibfnamefont {V.~E.}\ \bibnamefont
  {Elfving}}, \bibinfo {author} {\bibfnamefont {S.}~\bibnamefont {Polla}},
  \bibinfo {author} {\bibfnamefont {W.~J.}\ \bibnamefont {Huggins}}, \bibinfo
  {author} {\bibfnamefont {O.}~\bibnamefont {Oumarou}}, \bibinfo {author}
  {\bibfnamefont {K.}~\bibnamefont {Kechedzhi}}, \bibinfo {author}
  {\bibfnamefont {D.}~\bibnamefont {Abanin}}, \bibinfo {author} {\bibfnamefont
  {R.}~\bibnamefont {Acharya}}, \bibinfo {author} {\bibfnamefont
  {I.}~\bibnamefont {Aleiner}}, \bibinfo {author} {\bibfnamefont
  {R.}~\bibnamefont {Allen}}, \bibinfo {author} {\bibfnamefont {T.~I.}\
  \bibnamefont {Andersen}}, \bibinfo {author} {\bibfnamefont {K.}~\bibnamefont
  {Anderson}}, \bibinfo {author} {\bibfnamefont {M.}~\bibnamefont {Ansmann}},
  \bibinfo {author} {\bibfnamefont {F.}~\bibnamefont {Arute}}, \bibinfo
  {author} {\bibfnamefont {K.}~\bibnamefont {Arya}}, \bibinfo {author}
  {\bibfnamefont {A.}~\bibnamefont {Asfaw}}, \bibinfo {author} {\bibfnamefont
  {J.}~\bibnamefont {Atalaya}}, \bibinfo {author} {\bibfnamefont {J.~C.}\
  \bibnamefont {Bardin}}, \bibinfo {author} {\bibfnamefont {A.}~\bibnamefont
  {Bengtsson}}, \bibinfo {author} {\bibfnamefont {G.}~\bibnamefont {Bortoli}},
  \bibinfo {author} {\bibfnamefont {A.}~\bibnamefont {Bourassa}}, \bibinfo
  {author} {\bibfnamefont {J.}~\bibnamefont {Bovaird}}, \bibinfo {author}
  {\bibfnamefont {L.}~\bibnamefont {Brill}}, \bibinfo {author} {\bibfnamefont
  {M.}~\bibnamefont {Broughton}}, \bibinfo {author} {\bibfnamefont
  {B.}~\bibnamefont {Buckley}}, \bibinfo {author} {\bibfnamefont {D.~A.}\
  \bibnamefont {Buell}}, \bibinfo {author} {\bibfnamefont {T.}~\bibnamefont
  {Burger}}, \bibinfo {author} {\bibfnamefont {B.}~\bibnamefont {Burkett}},
  \bibinfo {author} {\bibfnamefont {N.}~\bibnamefont {Bushnell}}, \bibinfo
  {author} {\bibfnamefont {J.}~\bibnamefont {Campero}}, \bibinfo {author}
  {\bibfnamefont {Z.}~\bibnamefont {Chen}}, \bibinfo {author} {\bibfnamefont
  {B.}~\bibnamefont {Chiaro}}, \bibinfo {author} {\bibfnamefont
  {D.}~\bibnamefont {Chik}}, \bibinfo {author} {\bibfnamefont {J.}~\bibnamefont
  {Cogan}}, \bibinfo {author} {\bibfnamefont {R.}~\bibnamefont {Collins}},
  \bibinfo {author} {\bibfnamefont {P.}~\bibnamefont {Conner}}, \bibinfo
  {author} {\bibfnamefont {W.}~\bibnamefont {Courtney}}, \bibinfo {author}
  {\bibfnamefont {A.~L.}\ \bibnamefont {Crook}}, \bibinfo {author}
  {\bibfnamefont {B.}~\bibnamefont {Curtin}}, \bibinfo {author} {\bibfnamefont
  {D.~M.}\ \bibnamefont {Debroy}}, \bibinfo {author} {\bibfnamefont
  {S.}~\bibnamefont {Demura}}, \bibinfo {author} {\bibfnamefont
  {I.}~\bibnamefont {Drozdov}}, \bibinfo {author} {\bibfnamefont
  {A.}~\bibnamefont {Dunsworth}}, \bibinfo {author} {\bibfnamefont
  {C.}~\bibnamefont {Erickson}}, \bibinfo {author} {\bibfnamefont
  {L.}~\bibnamefont {Faoro}}, \bibinfo {author} {\bibfnamefont
  {E.}~\bibnamefont {Farhi}}, \bibinfo {author} {\bibfnamefont
  {R.}~\bibnamefont {Fatemi}}, \bibinfo {author} {\bibfnamefont {V.~S.}\
  \bibnamefont {Ferreira}}, \bibinfo {author} {\bibfnamefont {L.}~\bibnamefont
  {Flores~Burgos}}, \bibinfo {author} {\bibfnamefont {E.}~\bibnamefont
  {Forati}}, \bibinfo {author} {\bibfnamefont {A.~G.}\ \bibnamefont {Fowler}},
  \bibinfo {author} {\bibfnamefont {B.}~\bibnamefont {Foxen}}, \bibinfo
  {author} {\bibfnamefont {W.}~\bibnamefont {Giang}}, \bibinfo {author}
  {\bibfnamefont {C.}~\bibnamefont {Gidney}}, \bibinfo {author} {\bibfnamefont
  {D.}~\bibnamefont {Gilboa}}, \bibinfo {author} {\bibfnamefont
  {M.}~\bibnamefont {Giustina}}, \bibinfo {author} {\bibfnamefont
  {R.}~\bibnamefont {Gosula}}, \bibinfo {author} {\bibfnamefont
  {A.}~\bibnamefont {Grajales~Dau}}, \bibinfo {author} {\bibfnamefont {J.~A.}\
  \bibnamefont {Gross}}, \bibinfo {author} {\bibfnamefont {S.}~\bibnamefont
  {Habegger}}, \bibinfo {author} {\bibfnamefont {M.~C.}\ \bibnamefont
  {Hamilton}}, \bibinfo {author} {\bibfnamefont {M.}~\bibnamefont {Hansen}},
  \bibinfo {author} {\bibfnamefont {M.~P.}\ \bibnamefont {Harrigan}}, \bibinfo
  {author} {\bibfnamefont {S.~D.}\ \bibnamefont {Harrington}}, \bibinfo
  {author} {\bibfnamefont {P.}~\bibnamefont {Heu}}, \bibinfo {author}
  {\bibfnamefont {M.~R.}\ \bibnamefont {Hoffmann}}, \bibinfo {author}
  {\bibfnamefont {S.}~\bibnamefont {Hong}}, \bibinfo {author} {\bibfnamefont
  {T.}~\bibnamefont {Huang}}, \bibinfo {author} {\bibfnamefont
  {A.}~\bibnamefont {Huff}}, \bibinfo {author} {\bibfnamefont {L.~B.}\
  \bibnamefont {Ioffe}}, \bibinfo {author} {\bibfnamefont {S.~V.}\ \bibnamefont
  {Isakov}}, \bibinfo {author} {\bibfnamefont {J.}~\bibnamefont {Iveland}},
  \bibinfo {author} {\bibfnamefont {E.}~\bibnamefont {Jeffrey}}, \bibinfo
  {author} {\bibfnamefont {Z.}~\bibnamefont {Jiang}}, \bibinfo {author}
  {\bibfnamefont {C.}~\bibnamefont {Jones}}, \bibinfo {author} {\bibfnamefont
  {P.}~\bibnamefont {Juhas}}, \bibinfo {author} {\bibfnamefont
  {D.}~\bibnamefont {Kafri}}, \bibinfo {author} {\bibfnamefont
  {T.}~\bibnamefont {Khattar}}, \bibinfo {author} {\bibfnamefont
  {M.}~\bibnamefont {Khezri}}, \bibinfo {author} {\bibfnamefont
  {M.}~\bibnamefont {Kieferová}}, \bibinfo {author} {\bibfnamefont
  {S.}~\bibnamefont {Kim}}, \bibinfo {author} {\bibfnamefont {P.~V.}\
  \bibnamefont {Klimov}}, \bibinfo {author} {\bibfnamefont {A.~R.}\
  \bibnamefont {Klots}}, \bibinfo {author} {\bibfnamefont {A.~N.}\ \bibnamefont
  {Korotkov}}, \bibinfo {author} {\bibfnamefont {F.}~\bibnamefont {Kostritsa}},
  \bibinfo {author} {\bibfnamefont {J.~M.}\ \bibnamefont {Kreikebaum}},
  \bibinfo {author} {\bibfnamefont {D.}~\bibnamefont {Landhuis}}, \bibinfo
  {author} {\bibfnamefont {P.}~\bibnamefont {Laptev}}, \bibinfo {author}
  {\bibfnamefont {K.-M.}\ \bibnamefont {Lau}}, \bibinfo {author} {\bibfnamefont
  {L.}~\bibnamefont {Laws}}, \bibinfo {author} {\bibfnamefont {J.}~\bibnamefont
  {Lee}}, \bibinfo {author} {\bibfnamefont {K.}~\bibnamefont {Lee}}, \bibinfo
  {author} {\bibfnamefont {B.~J.}\ \bibnamefont {Lester}}, \bibinfo {author}
  {\bibfnamefont {A.~T.}\ \bibnamefont {Lill}}, \bibinfo {author}
  {\bibfnamefont {W.}~\bibnamefont {Liu}}, \bibinfo {author} {\bibfnamefont
  {W.~P.}\ \bibnamefont {Livingston}}, \bibinfo {author} {\bibfnamefont
  {A.}~\bibnamefont {Locharla}}, \bibinfo {author} {\bibfnamefont {F.~D.}\
  \bibnamefont {Malone}}, \bibinfo {author} {\bibfnamefont {S.}~\bibnamefont
  {Mandrà}}, \bibinfo {author} {\bibfnamefont {O.}~\bibnamefont {Martin}},
  \bibinfo {author} {\bibfnamefont {S.}~\bibnamefont {Martin}}, \bibinfo
  {author} {\bibfnamefont {J.~R.}\ \bibnamefont {McClean}}, \bibinfo {author}
  {\bibfnamefont {T.}~\bibnamefont {McCourt}}, \bibinfo {author} {\bibfnamefont
  {M.}~\bibnamefont {McEwen}}, \bibinfo {author} {\bibfnamefont
  {X.}~\bibnamefont {Mi}}, \bibinfo {author} {\bibfnamefont {A.}~\bibnamefont
  {Mieszala}}, \bibinfo {author} {\bibfnamefont {K.~C.}\ \bibnamefont {Miao}},
  \bibinfo {author} {\bibfnamefont {M.}~\bibnamefont {Mohseni}}, \bibinfo
  {author} {\bibfnamefont {S.}~\bibnamefont {Montazeri}}, \bibinfo {author}
  {\bibfnamefont {A.}~\bibnamefont {Morvan}}, \bibinfo {author} {\bibfnamefont
  {R.}~\bibnamefont {Movassagh}}, \bibinfo {author} {\bibfnamefont
  {W.}~\bibnamefont {Mruczkiewicz}}, \bibinfo {author} {\bibfnamefont
  {O.}~\bibnamefont {Naaman}}, \bibinfo {author} {\bibfnamefont
  {M.}~\bibnamefont {Neeley}}, \bibinfo {author} {\bibfnamefont
  {C.}~\bibnamefont {Neill}}, \bibinfo {author} {\bibfnamefont
  {A.}~\bibnamefont {Nersisyan}}, \bibinfo {author} {\bibfnamefont
  {M.}~\bibnamefont {Newman}}, \bibinfo {author} {\bibfnamefont {J.~H.}\
  \bibnamefont {Ng}}, \bibinfo {author} {\bibfnamefont {A.}~\bibnamefont
  {Nguyen}}, \bibinfo {author} {\bibfnamefont {M.}~\bibnamefont {Nguyen}},
  \bibinfo {author} {\bibfnamefont {M.~Y.}\ \bibnamefont {Niu}}, \bibinfo
  {author} {\bibfnamefont {S.}~\bibnamefont {Omonije}}, \bibinfo {author}
  {\bibfnamefont {A.}~\bibnamefont {Opremcak}}, \bibinfo {author}
  {\bibfnamefont {A.}~\bibnamefont {Petukhov}}, \bibinfo {author}
  {\bibfnamefont {R.}~\bibnamefont {Potter}}, \bibinfo {author} {\bibfnamefont
  {L.~P.}\ \bibnamefont {Pryadko}}, \bibinfo {author} {\bibfnamefont
  {C.}~\bibnamefont {Quintana}}, \bibinfo {author} {\bibfnamefont
  {C.}~\bibnamefont {Rocque}}, \bibinfo {author} {\bibfnamefont
  {P.}~\bibnamefont {Roushan}}, \bibinfo {author} {\bibfnamefont
  {N.}~\bibnamefont {Saei}}, \bibinfo {author} {\bibfnamefont {D.}~\bibnamefont
  {Sank}}, \bibinfo {author} {\bibfnamefont {K.}~\bibnamefont
  {Sankaragomathi}}, \bibinfo {author} {\bibfnamefont {K.~J.}\ \bibnamefont
  {Satzinger}}, \bibinfo {author} {\bibfnamefont {H.~F.}\ \bibnamefont
  {Schurkus}}, \bibinfo {author} {\bibfnamefont {C.}~\bibnamefont {Schuster}},
  \bibinfo {author} {\bibfnamefont {M.~J.}\ \bibnamefont {Shearn}}, \bibinfo
  {author} {\bibfnamefont {A.}~\bibnamefont {Shorter}}, \bibinfo {author}
  {\bibfnamefont {N.}~\bibnamefont {Shutty}}, \bibinfo {author} {\bibfnamefont
  {V.}~\bibnamefont {Shvarts}}, \bibinfo {author} {\bibfnamefont
  {J.}~\bibnamefont {Skruzny}}, \bibinfo {author} {\bibfnamefont {W.~C.}\
  \bibnamefont {Smith}}, \bibinfo {author} {\bibfnamefont {R.~D.}\ \bibnamefont
  {Somma}}, \bibinfo {author} {\bibfnamefont {G.}~\bibnamefont {Sterling}},
  \bibinfo {author} {\bibfnamefont {D.}~\bibnamefont {Strain}}, \bibinfo
  {author} {\bibfnamefont {M.}~\bibnamefont {Szalay}}, \bibinfo {author}
  {\bibfnamefont {D.}~\bibnamefont {Thor}}, \bibinfo {author} {\bibfnamefont
  {A.}~\bibnamefont {Torres}}, \bibinfo {author} {\bibfnamefont
  {G.}~\bibnamefont {Vidal}}, \bibinfo {author} {\bibfnamefont
  {B.}~\bibnamefont {Villalonga}}, \bibinfo {author} {\bibfnamefont
  {C.}~\bibnamefont {Vollgraff~Heidweiller}}, \bibinfo {author} {\bibfnamefont
  {T.}~\bibnamefont {White}}, \bibinfo {author} {\bibfnamefont {B.~W.~K.}\
  \bibnamefont {Woo}}, \bibinfo {author} {\bibfnamefont {C.}~\bibnamefont
  {Xing}}, \bibinfo {author} {\bibfnamefont {Z.~J.}\ \bibnamefont {Yao}},
  \bibinfo {author} {\bibfnamefont {P.}~\bibnamefont {Yeh}}, \bibinfo {author}
  {\bibfnamefont {J.}~\bibnamefont {Yoo}}, \bibinfo {author} {\bibfnamefont
  {G.}~\bibnamefont {Young}}, \bibinfo {author} {\bibfnamefont
  {A.}~\bibnamefont {Zalcman}}, \bibinfo {author} {\bibfnamefont
  {Y.}~\bibnamefont {Zhang}}, \bibinfo {author} {\bibfnamefont
  {N.}~\bibnamefont {Zhu}}, \bibinfo {author} {\bibfnamefont {N.}~\bibnamefont
  {Zobrist}}, \bibinfo {author} {\bibfnamefont {D.}~\bibnamefont {Bacon}},
  \bibinfo {author} {\bibfnamefont {S.}~\bibnamefont {Boixo}}, \bibinfo
  {author} {\bibfnamefont {Y.}~\bibnamefont {Chen}}, \bibinfo {author}
  {\bibfnamefont {J.}~\bibnamefont {Hilton}}, \bibinfo {author} {\bibfnamefont
  {J.}~\bibnamefont {Kelly}}, \bibinfo {author} {\bibfnamefont
  {E.}~\bibnamefont {Lucero}}, \bibinfo {author} {\bibfnamefont
  {A.}~\bibnamefont {Megrant}}, \bibinfo {author} {\bibfnamefont
  {H.}~\bibnamefont {Neven}}, \bibinfo {author} {\bibfnamefont
  {V.}~\bibnamefont {Smelyanskiy}}, \bibinfo {author} {\bibfnamefont
  {C.}~\bibnamefont {Gogolin}}, \bibinfo {author} {\bibfnamefont
  {R.}~\bibnamefont {Babbush}},\ and\ \bibinfo {author} {\bibfnamefont {N.~C.}\
  \bibnamefont {Rubin}},\ }\bibfield  {title} {\bibinfo {title}
  {Purification-based quantum error mitigation of pair-correlated electron
  simulations},\ }\bibfield  {journal} {\bibinfo  {journal} {Nature Physics}\
  }\href {https://doi.org/10.1038/s41567-023-02240-y}
  {10.1038/s41567-023-02240-y} (\bibinfo {year} {2023})\BibitemShut {NoStop}%
\bibitem [{\citenamefont {Arute}\ \emph {et~al.}(2019)\citenamefont {Arute},
  \citenamefont {Arya}, \citenamefont {Babbush}, \citenamefont {Bacon},
  \citenamefont {Bardin}, \citenamefont {Barends}, \citenamefont {Biswas},
  \citenamefont {Boixo}, \citenamefont {Brandao}, \citenamefont {Buell} \emph
  {et~al.}}]{arute2019quantum}%
  \BibitemOpen
  \bibfield  {author} {\bibinfo {author} {\bibfnamefont {F.}~\bibnamefont
  {Arute}}, \bibinfo {author} {\bibfnamefont {K.}~\bibnamefont {Arya}},
  \bibinfo {author} {\bibfnamefont {R.}~\bibnamefont {Babbush}}, \bibinfo
  {author} {\bibfnamefont {D.}~\bibnamefont {Bacon}}, \bibinfo {author}
  {\bibfnamefont {J.~C.}\ \bibnamefont {Bardin}}, \bibinfo {author}
  {\bibfnamefont {R.}~\bibnamefont {Barends}}, \bibinfo {author} {\bibfnamefont
  {R.}~\bibnamefont {Biswas}}, \bibinfo {author} {\bibfnamefont
  {S.}~\bibnamefont {Boixo}}, \bibinfo {author} {\bibfnamefont {F.~G.}\
  \bibnamefont {Brandao}}, \bibinfo {author} {\bibfnamefont {D.~A.}\
  \bibnamefont {Buell}}, \emph {et~al.},\ }\bibfield  {title} {\bibinfo {title}
  {Quantum supremacy using a programmable superconducting processor},\ }\href
  {https://doi.org/10.1038/s41586-019-1666-5} {\bibfield  {journal} {\bibinfo
  {journal} {Nature}\ }\textbf {\bibinfo {volume} {574}},\ \bibinfo {pages}
  {505} (\bibinfo {year} {2019})}\BibitemShut {NoStop}%
\bibitem [{\citenamefont {Helsen}\ and\ \citenamefont
  {Walter}(2022)}]{helsen2022thrifty}%
  \BibitemOpen
  \bibfield  {author} {\bibinfo {author} {\bibfnamefont {J.}~\bibnamefont
  {Helsen}}\ and\ \bibinfo {author} {\bibfnamefont {M.}~\bibnamefont
  {Walter}},\ }\href@noop {} {\bibinfo {title} {Thrifty shadow estimation:
  re-using quantum circuits and bounding tails}} (\bibinfo {year} {2022}),\
  \Eprint {https://arxiv.org/abs/2212.06240} {arXiv:2212.06240 [quant-ph]}
  \BibitemShut {NoStop}%
\bibitem [{\citenamefont {Zhou}\ and\ \citenamefont
  {Liu}(2023)}]{zhou2023performanceanalysis}%
  \BibitemOpen
  \bibfield  {author} {\bibinfo {author} {\bibfnamefont {Y.}~\bibnamefont
  {Zhou}}\ and\ \bibinfo {author} {\bibfnamefont {Q.}~\bibnamefont {Liu}},\
  }\bibfield  {title} {\bibinfo {title} {Performance analysis of multi-shot
  shadow estimation},\ }\href {https://doi.org/10.22331/q-2023-06-29-1044}
  {\bibfield  {journal} {\bibinfo  {journal} {{Quantum}}\ }\textbf {\bibinfo
  {volume} {7}},\ \bibinfo {pages} {1044} (\bibinfo {year} {2023})}\BibitemShut
  {NoStop}%
\bibitem [{\citenamefont {VandeVondele}\ and\ \citenamefont
  {Hutter}(2007)}]{vandevondele2007gaussian}%
  \BibitemOpen
  \bibfield  {author} {\bibinfo {author} {\bibfnamefont {J.}~\bibnamefont
  {VandeVondele}}\ and\ \bibinfo {author} {\bibfnamefont {J.}~\bibnamefont
  {Hutter}},\ }\bibfield  {title} {\bibinfo {title} {Gaussian basis sets for
  accurate calculations on molecular systems in gas and condensed phases},\
  }\bibfield  {journal} {\bibinfo  {journal} {The Journal of chemical physics}\
  }\textbf {\bibinfo {volume} {127}},\ \href
  {https://doi.org/10.1063/1.2770708} {10.1063/1.2770708} (\bibinfo {year}
  {2007})\BibitemShut {NoStop}%
\bibitem [{\citenamefont {Goedecker}\ \emph {et~al.}(1996)\citenamefont
  {Goedecker}, \citenamefont {Teter},\ and\ \citenamefont
  {Hutter}}]{goedecker1996separable}%
  \BibitemOpen
  \bibfield  {author} {\bibinfo {author} {\bibfnamefont {S.}~\bibnamefont
  {Goedecker}}, \bibinfo {author} {\bibfnamefont {M.}~\bibnamefont {Teter}},\
  and\ \bibinfo {author} {\bibfnamefont {J.}~\bibnamefont {Hutter}},\
  }\bibfield  {title} {\bibinfo {title} {Separable dual-space gaussian
  pseudopotentials},\ }\href {https://doi.org/10.1103/PhysRevB.54.1703}
  {\bibfield  {journal} {\bibinfo  {journal} {Physical Review B}\ }\textbf
  {\bibinfo {volume} {54}},\ \bibinfo {pages} {1703} (\bibinfo {year}
  {1996})}\BibitemShut {NoStop}%
\bibitem [{\citenamefont {Jnane}\ \emph {et~al.}(2023)\citenamefont {Jnane},
  \citenamefont {Steinberg}, \citenamefont {Cai}, \citenamefont {Nguyen},\ and\
  \citenamefont {Koczor}}]{jnane2023quantum}%
  \BibitemOpen
  \bibfield  {author} {\bibinfo {author} {\bibfnamefont {H.}~\bibnamefont
  {Jnane}}, \bibinfo {author} {\bibfnamefont {J.}~\bibnamefont {Steinberg}},
  \bibinfo {author} {\bibfnamefont {Z.}~\bibnamefont {Cai}}, \bibinfo {author}
  {\bibfnamefont {H.~C.}\ \bibnamefont {Nguyen}},\ and\ \bibinfo {author}
  {\bibfnamefont {B.}~\bibnamefont {Koczor}},\ }\href@noop {} {\bibinfo {title}
  {Quantum error mitigated classical shadows}} (\bibinfo {year} {2023}),\
  \Eprint {https://arxiv.org/abs/2305.04956} {arXiv:2305.04956 [quant-ph]}
  \BibitemShut {NoStop}%
\bibitem [{\citenamefont {Chan}\ \emph {et~al.}(2023)\citenamefont {Chan},
  \citenamefont {Meister}, \citenamefont {Goh},\ and\ \citenamefont
  {Koczor}}]{chan2023algorithmic}%
  \BibitemOpen
  \bibfield  {author} {\bibinfo {author} {\bibfnamefont {H.~H.~S.}\
  \bibnamefont {Chan}}, \bibinfo {author} {\bibfnamefont {R.}~\bibnamefont
  {Meister}}, \bibinfo {author} {\bibfnamefont {M.~L.}\ \bibnamefont {Goh}},\
  and\ \bibinfo {author} {\bibfnamefont {B.}~\bibnamefont {Koczor}},\
  }\href@noop {} {\bibinfo {title} {Algorithmic shadow spectroscopy}} (\bibinfo
  {year} {2023}),\ \Eprint {https://arxiv.org/abs/2212.11036} {arXiv:2212.11036
  [quant-ph]} \BibitemShut {NoStop}%
\bibitem [{\citenamefont {Brieger}\ \emph {et~al.}(2023)\citenamefont
  {Brieger}, \citenamefont {Heinrich}, \citenamefont {Roth},\ and\
  \citenamefont {Kliesch}}]{brieger2023stability}%
  \BibitemOpen
  \bibfield  {author} {\bibinfo {author} {\bibfnamefont {R.}~\bibnamefont
  {Brieger}}, \bibinfo {author} {\bibfnamefont {M.}~\bibnamefont {Heinrich}},
  \bibinfo {author} {\bibfnamefont {I.}~\bibnamefont {Roth}},\ and\ \bibinfo
  {author} {\bibfnamefont {M.}~\bibnamefont {Kliesch}},\ }\href@noop {}
  {\bibinfo {title} {Stability of classical shadows under gate-dependent
  noise}} (\bibinfo {year} {2023}),\ \Eprint {https://arxiv.org/abs/2310.19947}
  {arXiv:2310.19947 [quant-ph]} \BibitemShut {NoStop}%
\bibitem [{\citenamefont {Scheurer}\ \emph {et~al.}(2024)\citenamefont
  {Scheurer}, \citenamefont {Anselmetti}, \citenamefont {Oumarou},
  \citenamefont {Gogolin},\ and\ \citenamefont {Rubin}}]{zenododata}%
  \BibitemOpen
  \bibfield  {author} {\bibinfo {author} {\bibfnamefont {M.}~\bibnamefont
  {Scheurer}}, \bibinfo {author} {\bibfnamefont {G.-L.}\ \bibnamefont
  {Anselmetti}}, \bibinfo {author} {\bibfnamefont {O.}~\bibnamefont {Oumarou}},
  \bibinfo {author} {\bibfnamefont {C.}~\bibnamefont {Gogolin}},\ and\ \bibinfo
  {author} {\bibfnamefont {N.~C.}\ \bibnamefont {Rubin}},\ }\bibfield  {title}
  {\bibinfo {title} {{Data for "Tailored and Externally Corrected Coupled
  Cluster with Quantum Inputs"}},\ }\href
  {https://doi.org/10.5281/zenodo.10470740} {10.5281/zenodo.10470740} (\bibinfo
  {year} {2024})\BibitemShut {NoStop}%
\bibitem [{\citenamefont {Rubin}\ and\ \citenamefont
  {DePrince}(2021)}]{rubin2021pdaggerq}%
  \BibitemOpen
  \bibfield  {author} {\bibinfo {author} {\bibfnamefont {N.~C.}\ \bibnamefont
  {Rubin}}\ and\ \bibinfo {author} {\bibfnamefont {A.~E.}\ \bibnamefont
  {DePrince}},\ }\bibfield  {title} {\bibinfo {title} {p†q: a tool for
  prototyping many-body methods for quantum chemistry},\ }\bibfield  {journal}
  {\bibinfo  {journal} {Molecular Physics}\ }\textbf {\bibinfo {volume}
  {119}},\ \href {https://doi.org/10.1080/00268976.2021.1954709}
  {10.1080/00268976.2021.1954709} (\bibinfo {year} {2021})\BibitemShut
  {NoStop}%
\bibitem [{\citenamefont {Sun}\ \emph {et~al.}(2020)\citenamefont {Sun},
  \citenamefont {Zhang}, \citenamefont {Banerjee}, \citenamefont {Bao},
  \citenamefont {Barbry}, \citenamefont {Blunt}, \citenamefont {Bogdanov},
  \citenamefont {Booth}, \citenamefont {Chen}, \citenamefont {Cui},
  \citenamefont {Eriksen}, \citenamefont {Gao}, \citenamefont {Guo},
  \citenamefont {Hermann}, \citenamefont {Hermes}, \citenamefont {Koh},
  \citenamefont {Koval}, \citenamefont {Lehtola}, \citenamefont {Li},
  \citenamefont {Liu}, \citenamefont {Mardirossian}, \citenamefont {McClain},
  \citenamefont {Motta}, \citenamefont {Mussard}, \citenamefont {Pham},
  \citenamefont {Pulkin}, \citenamefont {Purwanto}, \citenamefont {Robinson},
  \citenamefont {Ronca}, \citenamefont {Sayfutyarova}, \citenamefont
  {Scheurer}, \citenamefont {Schurkus}, \citenamefont {Smith}, \citenamefont
  {Sun}, \citenamefont {Sun}, \citenamefont {Upadhyay}, \citenamefont {Wagner},
  \citenamefont {Wang}, \citenamefont {White}, \citenamefont {Whitfield},
  \citenamefont {Williamson}, \citenamefont {Wouters}, \citenamefont {Yang},
  \citenamefont {Yu}, \citenamefont {Zhu}, \citenamefont {Berkelbach},
  \citenamefont {Sharma}, \citenamefont {Sokolov},\ and\ \citenamefont
  {Chan}}]{sun2020recent}%
  \BibitemOpen
  \bibfield  {author} {\bibinfo {author} {\bibfnamefont {Q.}~\bibnamefont
  {Sun}}, \bibinfo {author} {\bibfnamefont {X.}~\bibnamefont {Zhang}}, \bibinfo
  {author} {\bibfnamefont {S.}~\bibnamefont {Banerjee}}, \bibinfo {author}
  {\bibfnamefont {P.}~\bibnamefont {Bao}}, \bibinfo {author} {\bibfnamefont
  {M.}~\bibnamefont {Barbry}}, \bibinfo {author} {\bibfnamefont {N.~S.}\
  \bibnamefont {Blunt}}, \bibinfo {author} {\bibfnamefont {N.~A.}\ \bibnamefont
  {Bogdanov}}, \bibinfo {author} {\bibfnamefont {G.~H.}\ \bibnamefont {Booth}},
  \bibinfo {author} {\bibfnamefont {J.}~\bibnamefont {Chen}}, \bibinfo {author}
  {\bibfnamefont {Z.-H.}\ \bibnamefont {Cui}}, \bibinfo {author} {\bibfnamefont
  {J.~J.}\ \bibnamefont {Eriksen}}, \bibinfo {author} {\bibfnamefont
  {Y.}~\bibnamefont {Gao}}, \bibinfo {author} {\bibfnamefont {S.}~\bibnamefont
  {Guo}}, \bibinfo {author} {\bibfnamefont {J.}~\bibnamefont {Hermann}},
  \bibinfo {author} {\bibfnamefont {M.~R.}\ \bibnamefont {Hermes}}, \bibinfo
  {author} {\bibfnamefont {K.}~\bibnamefont {Koh}}, \bibinfo {author}
  {\bibfnamefont {P.}~\bibnamefont {Koval}}, \bibinfo {author} {\bibfnamefont
  {S.}~\bibnamefont {Lehtola}}, \bibinfo {author} {\bibfnamefont
  {Z.}~\bibnamefont {Li}}, \bibinfo {author} {\bibfnamefont {J.}~\bibnamefont
  {Liu}}, \bibinfo {author} {\bibfnamefont {N.}~\bibnamefont {Mardirossian}},
  \bibinfo {author} {\bibfnamefont {J.~D.}\ \bibnamefont {McClain}}, \bibinfo
  {author} {\bibfnamefont {M.}~\bibnamefont {Motta}}, \bibinfo {author}
  {\bibfnamefont {B.}~\bibnamefont {Mussard}}, \bibinfo {author} {\bibfnamefont
  {H.~Q.}\ \bibnamefont {Pham}}, \bibinfo {author} {\bibfnamefont
  {A.}~\bibnamefont {Pulkin}}, \bibinfo {author} {\bibfnamefont
  {W.}~\bibnamefont {Purwanto}}, \bibinfo {author} {\bibfnamefont {P.~J.}\
  \bibnamefont {Robinson}}, \bibinfo {author} {\bibfnamefont {E.}~\bibnamefont
  {Ronca}}, \bibinfo {author} {\bibfnamefont {E.~R.}\ \bibnamefont
  {Sayfutyarova}}, \bibinfo {author} {\bibfnamefont {M.}~\bibnamefont
  {Scheurer}}, \bibinfo {author} {\bibfnamefont {H.~F.}\ \bibnamefont
  {Schurkus}}, \bibinfo {author} {\bibfnamefont {J.~E.~T.}\ \bibnamefont
  {Smith}}, \bibinfo {author} {\bibfnamefont {C.}~\bibnamefont {Sun}}, \bibinfo
  {author} {\bibfnamefont {S.-N.}\ \bibnamefont {Sun}}, \bibinfo {author}
  {\bibfnamefont {S.}~\bibnamefont {Upadhyay}}, \bibinfo {author}
  {\bibfnamefont {L.~K.}\ \bibnamefont {Wagner}}, \bibinfo {author}
  {\bibfnamefont {X.}~\bibnamefont {Wang}}, \bibinfo {author} {\bibfnamefont
  {A.}~\bibnamefont {White}}, \bibinfo {author} {\bibfnamefont {J.~D.}\
  \bibnamefont {Whitfield}}, \bibinfo {author} {\bibfnamefont {M.~J.}\
  \bibnamefont {Williamson}}, \bibinfo {author} {\bibfnamefont
  {S.}~\bibnamefont {Wouters}}, \bibinfo {author} {\bibfnamefont
  {J.}~\bibnamefont {Yang}}, \bibinfo {author} {\bibfnamefont {J.~M.}\
  \bibnamefont {Yu}}, \bibinfo {author} {\bibfnamefont {T.}~\bibnamefont
  {Zhu}}, \bibinfo {author} {\bibfnamefont {T.~C.}\ \bibnamefont {Berkelbach}},
  \bibinfo {author} {\bibfnamefont {S.}~\bibnamefont {Sharma}}, \bibinfo
  {author} {\bibfnamefont {A.~Y.}\ \bibnamefont {Sokolov}},\ and\ \bibinfo
  {author} {\bibfnamefont {G.~K.-L.}\ \bibnamefont {Chan}},\ }\bibfield
  {title} {\bibinfo {title} {{Recent developments in the PySCF program
  package}},\ }\href {https://doi.org/10.1063/5.0006074} {\bibfield  {journal}
  {\bibinfo  {journal} {The Journal of Chemical Physics}\ }\textbf {\bibinfo
  {volume} {153}},\ \bibinfo {pages} {024109} (\bibinfo {year}
  {2020})}\BibitemShut {NoStop}%
\bibitem [{\citenamefont {Bradbury}\ \emph {et~al.}(2018)\citenamefont
  {Bradbury}, \citenamefont {Frostig}, \citenamefont {Hawkins}, \citenamefont
  {Johnson}, \citenamefont {Leary}, \citenamefont {Maclaurin}, \citenamefont
  {Necula}, \citenamefont {Paszke}, \citenamefont {Vander{P}las}, \citenamefont
  {Wanderman-{M}ilne},\ and\ \citenamefont {Zhang}}]{jax2018github}%
  \BibitemOpen
  \bibfield  {author} {\bibinfo {author} {\bibfnamefont {J.}~\bibnamefont
  {Bradbury}}, \bibinfo {author} {\bibfnamefont {R.}~\bibnamefont {Frostig}},
  \bibinfo {author} {\bibfnamefont {P.}~\bibnamefont {Hawkins}}, \bibinfo
  {author} {\bibfnamefont {M.~J.}\ \bibnamefont {Johnson}}, \bibinfo {author}
  {\bibfnamefont {C.}~\bibnamefont {Leary}}, \bibinfo {author} {\bibfnamefont
  {D.}~\bibnamefont {Maclaurin}}, \bibinfo {author} {\bibfnamefont
  {G.}~\bibnamefont {Necula}}, \bibinfo {author} {\bibfnamefont
  {A.}~\bibnamefont {Paszke}}, \bibinfo {author} {\bibfnamefont
  {J.}~\bibnamefont {Vander{P}las}}, \bibinfo {author} {\bibfnamefont
  {S.}~\bibnamefont {Wanderman-{M}ilne}},\ and\ \bibinfo {author}
  {\bibfnamefont {Q.}~\bibnamefont {Zhang}},\ }\href
  {http://github.com/google/jax} {\bibinfo {title} {{JAX}: composable
  transformations of {P}ython+{N}um{P}y programs}} (\bibinfo {year}
  {2018})\BibitemShut {NoStop}%
\bibitem [{\citenamefont {Rubin}\ \emph {et~al.}(2021)\citenamefont {Rubin},
  \citenamefont {Gunst}, \citenamefont {White}, \citenamefont {Freitag},
  \citenamefont {Throssell}, \citenamefont {Chan}, \citenamefont {Babbush},\
  and\ \citenamefont {Shiozaki}}]{rubin2021fermionic}%
  \BibitemOpen
  \bibfield  {author} {\bibinfo {author} {\bibfnamefont {N.~C.}\ \bibnamefont
  {Rubin}}, \bibinfo {author} {\bibfnamefont {K.}~\bibnamefont {Gunst}},
  \bibinfo {author} {\bibfnamefont {A.}~\bibnamefont {White}}, \bibinfo
  {author} {\bibfnamefont {L.}~\bibnamefont {Freitag}}, \bibinfo {author}
  {\bibfnamefont {K.}~\bibnamefont {Throssell}}, \bibinfo {author}
  {\bibfnamefont {G.~K.-L.}\ \bibnamefont {Chan}}, \bibinfo {author}
  {\bibfnamefont {R.}~\bibnamefont {Babbush}},\ and\ \bibinfo {author}
  {\bibfnamefont {T.}~\bibnamefont {Shiozaki}},\ }\bibfield  {title} {\bibinfo
  {title} {The fermionic quantum emulator},\ }\href
  {https://doi.org/10.22331/q-2021-10-27-568} {\bibfield  {journal} {\bibinfo
  {journal} {Quantum}\ }\textbf {\bibinfo {volume} {5}},\ \bibinfo {pages}
  {568} (\bibinfo {year} {2021})}\BibitemShut {NoStop}%
\bibitem [{\citenamefont {Bergholm}\ \emph {et~al.}(2018)\citenamefont
  {Bergholm}, \citenamefont {Izaac}, \citenamefont {Schuld}, \citenamefont
  {Gogolin}, \citenamefont {Ahmed}, \citenamefont {Ajith}, \citenamefont
  {Alam}, \citenamefont {Alonso-Linaje}, \citenamefont {AkashNarayanan},
  \citenamefont {Asadi}, \citenamefont {Arrazola}, \citenamefont {Azad},
  \citenamefont {Banning}, \citenamefont {Blank}, \citenamefont {Bromley},
  \citenamefont {Cordier}, \citenamefont {Ceroni}, \citenamefont {Delgado},
  \citenamefont {Di~Matteo}, \citenamefont {Dusko}, \citenamefont {Garg},
  \citenamefont {Guala}, \citenamefont {Hayes}, \citenamefont {Hill},
  \citenamefont {Ijaz}, \citenamefont {Isacsson}, \citenamefont {Ittah},
  \citenamefont {Jahangiri}, \citenamefont {Jain}, \citenamefont {Jiang},
  \citenamefont {Khandelwal}, \citenamefont {Kottmann}, \citenamefont {Lang},
  \citenamefont {Lee}, \citenamefont {Loke}, \citenamefont {Lowe},
  \citenamefont {McKiernan}, \citenamefont {Meyer}, \citenamefont
  {Montañez-Barrera}, \citenamefont {Moyard}, \citenamefont {Niu},
  \citenamefont {O'Riordan}, \citenamefont {Oud}, \citenamefont {Panigrahi},
  \citenamefont {Park}, \citenamefont {Polatajko}, \citenamefont {Quesada},
  \citenamefont {Roberts}, \citenamefont {Sá}, \citenamefont {Schoch},
  \citenamefont {Shi}, \citenamefont {Shu}, \citenamefont {Sim}, \citenamefont
  {Singh}, \citenamefont {Strandberg}, \citenamefont {Soni}, \citenamefont
  {Száva}, \citenamefont {Thabet}, \citenamefont {Vargas-Hernández},
  \citenamefont {Vincent}, \citenamefont {Vitucci}, \citenamefont {Weber},
  \citenamefont {Wierichs}, \citenamefont {Wiersema}, \citenamefont {Willmann},
  \citenamefont {Wong}, \citenamefont {Zhang},\ and\ \citenamefont
  {Killoran}}]{pennylane}%
  \BibitemOpen
  \bibfield  {author} {\bibinfo {author} {\bibfnamefont {V.}~\bibnamefont
  {Bergholm}}, \bibinfo {author} {\bibfnamefont {J.}~\bibnamefont {Izaac}},
  \bibinfo {author} {\bibfnamefont {M.}~\bibnamefont {Schuld}}, \bibinfo
  {author} {\bibfnamefont {C.}~\bibnamefont {Gogolin}}, \bibinfo {author}
  {\bibfnamefont {S.}~\bibnamefont {Ahmed}}, \bibinfo {author} {\bibfnamefont
  {V.}~\bibnamefont {Ajith}}, \bibinfo {author} {\bibfnamefont {M.~S.}\
  \bibnamefont {Alam}}, \bibinfo {author} {\bibfnamefont {G.}~\bibnamefont
  {Alonso-Linaje}}, \bibinfo {author} {\bibfnamefont {B.}~\bibnamefont
  {AkashNarayanan}}, \bibinfo {author} {\bibfnamefont {A.}~\bibnamefont
  {Asadi}}, \bibinfo {author} {\bibfnamefont {J.~M.}\ \bibnamefont {Arrazola}},
  \bibinfo {author} {\bibfnamefont {U.}~\bibnamefont {Azad}}, \bibinfo {author}
  {\bibfnamefont {S.}~\bibnamefont {Banning}}, \bibinfo {author} {\bibfnamefont
  {C.}~\bibnamefont {Blank}}, \bibinfo {author} {\bibfnamefont {T.~R.}\
  \bibnamefont {Bromley}}, \bibinfo {author} {\bibfnamefont {B.~A.}\
  \bibnamefont {Cordier}}, \bibinfo {author} {\bibfnamefont {J.}~\bibnamefont
  {Ceroni}}, \bibinfo {author} {\bibfnamefont {A.}~\bibnamefont {Delgado}},
  \bibinfo {author} {\bibfnamefont {O.}~\bibnamefont {Di~Matteo}}, \bibinfo
  {author} {\bibfnamefont {A.}~\bibnamefont {Dusko}}, \bibinfo {author}
  {\bibfnamefont {T.}~\bibnamefont {Garg}}, \bibinfo {author} {\bibfnamefont
  {D.}~\bibnamefont {Guala}}, \bibinfo {author} {\bibfnamefont
  {A.}~\bibnamefont {Hayes}}, \bibinfo {author} {\bibfnamefont
  {R.}~\bibnamefont {Hill}}, \bibinfo {author} {\bibfnamefont {A.}~\bibnamefont
  {Ijaz}}, \bibinfo {author} {\bibfnamefont {T.}~\bibnamefont {Isacsson}},
  \bibinfo {author} {\bibfnamefont {D.}~\bibnamefont {Ittah}}, \bibinfo
  {author} {\bibfnamefont {S.}~\bibnamefont {Jahangiri}}, \bibinfo {author}
  {\bibfnamefont {P.}~\bibnamefont {Jain}}, \bibinfo {author} {\bibfnamefont
  {E.}~\bibnamefont {Jiang}}, \bibinfo {author} {\bibfnamefont
  {A.}~\bibnamefont {Khandelwal}}, \bibinfo {author} {\bibfnamefont
  {K.}~\bibnamefont {Kottmann}}, \bibinfo {author} {\bibfnamefont {R.~A.}\
  \bibnamefont {Lang}}, \bibinfo {author} {\bibfnamefont {C.}~\bibnamefont
  {Lee}}, \bibinfo {author} {\bibfnamefont {T.}~\bibnamefont {Loke}}, \bibinfo
  {author} {\bibfnamefont {A.}~\bibnamefont {Lowe}}, \bibinfo {author}
  {\bibfnamefont {K.}~\bibnamefont {McKiernan}}, \bibinfo {author}
  {\bibfnamefont {J.~J.}\ \bibnamefont {Meyer}}, \bibinfo {author}
  {\bibfnamefont {J.~A.}\ \bibnamefont {Montañez-Barrera}}, \bibinfo {author}
  {\bibfnamefont {R.}~\bibnamefont {Moyard}}, \bibinfo {author} {\bibfnamefont
  {Z.}~\bibnamefont {Niu}}, \bibinfo {author} {\bibfnamefont {L.~J.}\
  \bibnamefont {O'Riordan}}, \bibinfo {author} {\bibfnamefont {S.}~\bibnamefont
  {Oud}}, \bibinfo {author} {\bibfnamefont {A.}~\bibnamefont {Panigrahi}},
  \bibinfo {author} {\bibfnamefont {C.-Y.}\ \bibnamefont {Park}}, \bibinfo
  {author} {\bibfnamefont {D.}~\bibnamefont {Polatajko}}, \bibinfo {author}
  {\bibfnamefont {N.}~\bibnamefont {Quesada}}, \bibinfo {author} {\bibfnamefont
  {C.}~\bibnamefont {Roberts}}, \bibinfo {author} {\bibfnamefont
  {N.}~\bibnamefont {Sá}}, \bibinfo {author} {\bibfnamefont {I.}~\bibnamefont
  {Schoch}}, \bibinfo {author} {\bibfnamefont {B.}~\bibnamefont {Shi}},
  \bibinfo {author} {\bibfnamefont {S.}~\bibnamefont {Shu}}, \bibinfo {author}
  {\bibfnamefont {S.}~\bibnamefont {Sim}}, \bibinfo {author} {\bibfnamefont
  {A.}~\bibnamefont {Singh}}, \bibinfo {author} {\bibfnamefont
  {I.}~\bibnamefont {Strandberg}}, \bibinfo {author} {\bibfnamefont
  {J.}~\bibnamefont {Soni}}, \bibinfo {author} {\bibfnamefont {A.}~\bibnamefont
  {Száva}}, \bibinfo {author} {\bibfnamefont {S.}~\bibnamefont {Thabet}},
  \bibinfo {author} {\bibfnamefont {R.~A.}\ \bibnamefont {Vargas-Hernández}},
  \bibinfo {author} {\bibfnamefont {T.}~\bibnamefont {Vincent}}, \bibinfo
  {author} {\bibfnamefont {N.}~\bibnamefont {Vitucci}}, \bibinfo {author}
  {\bibfnamefont {M.}~\bibnamefont {Weber}}, \bibinfo {author} {\bibfnamefont
  {D.}~\bibnamefont {Wierichs}}, \bibinfo {author} {\bibfnamefont
  {R.}~\bibnamefont {Wiersema}}, \bibinfo {author} {\bibfnamefont
  {M.}~\bibnamefont {Willmann}}, \bibinfo {author} {\bibfnamefont
  {V.}~\bibnamefont {Wong}}, \bibinfo {author} {\bibfnamefont {S.}~\bibnamefont
  {Zhang}},\ and\ \bibinfo {author} {\bibfnamefont {N.}~\bibnamefont
  {Killoran}},\ }\href {https://arxiv.org/abs/1811.04968} {\bibinfo {title}
  {Pennylane: Automatic differentiation of hybrid quantum-classical
  computations}} (\bibinfo {year} {2018})\BibitemShut {NoStop}%
\bibitem [{\citenamefont {Harris}\ \emph {et~al.}(2020)\citenamefont {Harris},
  \citenamefont {Millman}, \citenamefont {van~der Walt}, \citenamefont
  {Gommers}, \citenamefont {Virtanen}, \citenamefont {Cournapeau},
  \citenamefont {Wieser}, \citenamefont {Taylor}, \citenamefont {Berg},
  \citenamefont {Smith}, \citenamefont {Kern}, \citenamefont {Picus},
  \citenamefont {Hoyer}, \citenamefont {van Kerkwijk}, \citenamefont {Brett},
  \citenamefont {Haldane}, \citenamefont {del Río}, \citenamefont {Wiebe},
  \citenamefont {Peterson}, \citenamefont {Gérard-Marchant}, \citenamefont
  {Sheppard}, \citenamefont {Reddy}, \citenamefont {Weckesser}, \citenamefont
  {Abbasi}, \citenamefont {Gohlke},\ and\ \citenamefont
  {Oliphant}}]{harris2020numpy}%
  \BibitemOpen
  \bibfield  {author} {\bibinfo {author} {\bibfnamefont {C.~R.}\ \bibnamefont
  {Harris}}, \bibinfo {author} {\bibfnamefont {K.~J.}\ \bibnamefont {Millman}},
  \bibinfo {author} {\bibfnamefont {S.~J.}\ \bibnamefont {van~der Walt}},
  \bibinfo {author} {\bibfnamefont {R.}~\bibnamefont {Gommers}}, \bibinfo
  {author} {\bibfnamefont {P.}~\bibnamefont {Virtanen}}, \bibinfo {author}
  {\bibfnamefont {D.}~\bibnamefont {Cournapeau}}, \bibinfo {author}
  {\bibfnamefont {E.}~\bibnamefont {Wieser}}, \bibinfo {author} {\bibfnamefont
  {J.}~\bibnamefont {Taylor}}, \bibinfo {author} {\bibfnamefont
  {S.}~\bibnamefont {Berg}}, \bibinfo {author} {\bibfnamefont {N.~J.}\
  \bibnamefont {Smith}}, \bibinfo {author} {\bibfnamefont {R.}~\bibnamefont
  {Kern}}, \bibinfo {author} {\bibfnamefont {M.}~\bibnamefont {Picus}},
  \bibinfo {author} {\bibfnamefont {S.}~\bibnamefont {Hoyer}}, \bibinfo
  {author} {\bibfnamefont {M.~H.}\ \bibnamefont {van Kerkwijk}}, \bibinfo
  {author} {\bibfnamefont {M.}~\bibnamefont {Brett}}, \bibinfo {author}
  {\bibfnamefont {A.}~\bibnamefont {Haldane}}, \bibinfo {author} {\bibfnamefont
  {J.~F.}\ \bibnamefont {del Río}}, \bibinfo {author} {\bibfnamefont
  {M.}~\bibnamefont {Wiebe}}, \bibinfo {author} {\bibfnamefont
  {P.}~\bibnamefont {Peterson}}, \bibinfo {author} {\bibfnamefont
  {P.}~\bibnamefont {Gérard-Marchant}}, \bibinfo {author} {\bibfnamefont
  {K.}~\bibnamefont {Sheppard}}, \bibinfo {author} {\bibfnamefont
  {T.}~\bibnamefont {Reddy}}, \bibinfo {author} {\bibfnamefont
  {W.}~\bibnamefont {Weckesser}}, \bibinfo {author} {\bibfnamefont
  {H.}~\bibnamefont {Abbasi}}, \bibinfo {author} {\bibfnamefont
  {C.}~\bibnamefont {Gohlke}},\ and\ \bibinfo {author} {\bibfnamefont {T.~E.}\
  \bibnamefont {Oliphant}},\ }\bibfield  {title} {\bibinfo {title} {Array
  programming with numpy},\ }\href {https://doi.org/10.1038/s41586-020-2649-2}
  {\bibfield  {journal} {\bibinfo  {journal} {Nature}\ }\textbf {\bibinfo
  {volume} {585}},\ \bibinfo {pages} {357} (\bibinfo {year}
  {2020})}\BibitemShut {NoStop}%
\bibitem [{\citenamefont {Virtanen}\ \emph {et~al.}(2020)\citenamefont
  {Virtanen}, \citenamefont {Gommers}, \citenamefont {Oliphant}, \citenamefont
  {Haberland}, \citenamefont {Reddy}, \citenamefont {Cournapeau}, \citenamefont
  {Burovski}, \citenamefont {Peterson}, \citenamefont {Weckesser},
  \citenamefont {Bright}, \citenamefont {{van der Walt}}, \citenamefont
  {Brett}, \citenamefont {Wilson}, \citenamefont {Millman}, \citenamefont
  {Mayorov}, \citenamefont {Nelson}, \citenamefont {Jones}, \citenamefont
  {Kern}, \citenamefont {Larson}, \citenamefont {Carey}, \citenamefont {Polat},
  \citenamefont {Feng}, \citenamefont {Moore}, \citenamefont {{VanderPlas}},
  \citenamefont {Laxalde}, \citenamefont {Perktold}, \citenamefont {Cimrman},
  \citenamefont {Henriksen}, \citenamefont {Quintero}, \citenamefont {Harris},
  \citenamefont {Archibald}, \citenamefont {Ribeiro}, \citenamefont
  {Pedregosa}, \citenamefont {{van Mulbregt}},\ and\ \citenamefont {{SciPy 1.0
  Contributors}}}]{scipy}%
  \BibitemOpen
  \bibfield  {author} {\bibinfo {author} {\bibfnamefont {P.}~\bibnamefont
  {Virtanen}}, \bibinfo {author} {\bibfnamefont {R.}~\bibnamefont {Gommers}},
  \bibinfo {author} {\bibfnamefont {T.~E.}\ \bibnamefont {Oliphant}}, \bibinfo
  {author} {\bibfnamefont {M.}~\bibnamefont {Haberland}}, \bibinfo {author}
  {\bibfnamefont {T.}~\bibnamefont {Reddy}}, \bibinfo {author} {\bibfnamefont
  {D.}~\bibnamefont {Cournapeau}}, \bibinfo {author} {\bibfnamefont
  {E.}~\bibnamefont {Burovski}}, \bibinfo {author} {\bibfnamefont
  {P.}~\bibnamefont {Peterson}}, \bibinfo {author} {\bibfnamefont
  {W.}~\bibnamefont {Weckesser}}, \bibinfo {author} {\bibfnamefont
  {J.}~\bibnamefont {Bright}}, \bibinfo {author} {\bibfnamefont {S.~J.}\
  \bibnamefont {{van der Walt}}}, \bibinfo {author} {\bibfnamefont
  {M.}~\bibnamefont {Brett}}, \bibinfo {author} {\bibfnamefont
  {J.}~\bibnamefont {Wilson}}, \bibinfo {author} {\bibfnamefont {K.~J.}\
  \bibnamefont {Millman}}, \bibinfo {author} {\bibfnamefont {N.}~\bibnamefont
  {Mayorov}}, \bibinfo {author} {\bibfnamefont {A.~R.~J.}\ \bibnamefont
  {Nelson}}, \bibinfo {author} {\bibfnamefont {E.}~\bibnamefont {Jones}},
  \bibinfo {author} {\bibfnamefont {R.}~\bibnamefont {Kern}}, \bibinfo {author}
  {\bibfnamefont {E.}~\bibnamefont {Larson}}, \bibinfo {author} {\bibfnamefont
  {C.~J.}\ \bibnamefont {Carey}}, \bibinfo {author} {\bibfnamefont
  {{\.I}.}~\bibnamefont {Polat}}, \bibinfo {author} {\bibfnamefont
  {Y.}~\bibnamefont {Feng}}, \bibinfo {author} {\bibfnamefont {E.~W.}\
  \bibnamefont {Moore}}, \bibinfo {author} {\bibfnamefont {J.}~\bibnamefont
  {{VanderPlas}}}, \bibinfo {author} {\bibfnamefont {D.}~\bibnamefont
  {Laxalde}}, \bibinfo {author} {\bibfnamefont {J.}~\bibnamefont {Perktold}},
  \bibinfo {author} {\bibfnamefont {R.}~\bibnamefont {Cimrman}}, \bibinfo
  {author} {\bibfnamefont {I.}~\bibnamefont {Henriksen}}, \bibinfo {author}
  {\bibfnamefont {E.~A.}\ \bibnamefont {Quintero}}, \bibinfo {author}
  {\bibfnamefont {C.~R.}\ \bibnamefont {Harris}}, \bibinfo {author}
  {\bibfnamefont {A.~M.}\ \bibnamefont {Archibald}}, \bibinfo {author}
  {\bibfnamefont {A.~H.}\ \bibnamefont {Ribeiro}}, \bibinfo {author}
  {\bibfnamefont {F.}~\bibnamefont {Pedregosa}}, \bibinfo {author}
  {\bibfnamefont {P.}~\bibnamefont {{van Mulbregt}}},\ and\ \bibinfo {author}
  {\bibnamefont {{SciPy 1.0 Contributors}}},\ }\bibfield  {title} {\bibinfo
  {title} {{{SciPy} 1.0: Fundamental Algorithms for Scientific Computing in
  Python}},\ }\href {https://doi.org/10.1038/s41592-019-0686-2} {\bibfield
  {journal} {\bibinfo  {journal} {Nature Methods}\ }\textbf {\bibinfo {volume}
  {17}},\ \bibinfo {pages} {261} (\bibinfo {year} {2020})}\BibitemShut
  {NoStop}%
\bibitem [{\citenamefont {pandas~development team}(2020)}]{reback2020pandas}%
  \BibitemOpen
  \bibfield  {author} {\bibinfo {author} {\bibfnamefont {T.}~\bibnamefont
  {pandas~development team}},\ }\href {https://doi.org/10.5281/zenodo.3509134}
  {\bibinfo {title} {pandas-dev/pandas: Pandas}} (\bibinfo {year}
  {2020})\BibitemShut {NoStop}%
\bibitem [{\citenamefont {{W}es
  {M}c{K}inney}(2010)}]{mckinney-proc-scipy-2010}%
  \BibitemOpen
  \bibfield  {author} {\bibinfo {author} {\bibnamefont {{W}es {M}c{K}inney}},\
  }\bibfield  {title} {\bibinfo {title} {{D}ata {S}tructures for {S}tatistical
  {C}omputing in {P}ython},\ }in\ \href
  {https://doi.org/10.25080/Majora-92bf1922-00a} {\emph {\bibinfo {booktitle}
  {{P}roceedings of the 9th {P}ython in {S}cience {C}onference}}},\ \bibinfo
  {editor} {edited by\ \bibinfo {editor} {\bibnamefont {{S}t\'efan van~der
  {W}alt}}\ and\ \bibinfo {editor} {\bibnamefont {{J}arrod {M}illman}}}\
  (\bibinfo {year} {2010})\ pp.\ \bibinfo {pages} {56 -- 61}\BibitemShut
  {NoStop}%
\bibitem [{\citenamefont {Hunter}(2007)}]{Hunter:2007}%
  \BibitemOpen
  \bibfield  {author} {\bibinfo {author} {\bibfnamefont {J.~D.}\ \bibnamefont
  {Hunter}},\ }\bibfield  {title} {\bibinfo {title} {Matplotlib: A 2d graphics
  environment},\ }\href {https://doi.org/10.1109/MCSE.2007.55} {\bibfield
  {journal} {\bibinfo  {journal} {Computing in Science \& Engineering}\
  }\textbf {\bibinfo {volume} {9}},\ \bibinfo {pages} {90} (\bibinfo {year}
  {2007})}\BibitemShut {NoStop}%
\bibitem [{\citenamefont {Waskom}(2021)}]{Waskom2021}%
  \BibitemOpen
  \bibfield  {author} {\bibinfo {author} {\bibfnamefont {M.~L.}\ \bibnamefont
  {Waskom}},\ }\bibfield  {title} {\bibinfo {title} {seaborn: statistical data
  visualization},\ }\href {https://doi.org/10.21105/joss.03021} {\bibfield
  {journal} {\bibinfo  {journal} {Journal of Open Source Software}\ }\textbf
  {\bibinfo {volume} {6}},\ \bibinfo {pages} {3021} (\bibinfo {year}
  {2021})}\BibitemShut {NoStop}%
\bibitem [{\citenamefont {Meurer}\ \emph {et~al.}(2017)\citenamefont {Meurer},
  \citenamefont {Smith}, \citenamefont {Paprocki}, \citenamefont
  {\v{C}ert\'{i}k}, \citenamefont {Kirpichev}, \citenamefont {Rocklin},
  \citenamefont {Kumar}, \citenamefont {Ivanov}, \citenamefont {Moore},
  \citenamefont {Singh}, \citenamefont {Rathnayake}, \citenamefont {Vig},
  \citenamefont {Granger}, \citenamefont {Muller}, \citenamefont {Bonazzi},
  \citenamefont {Gupta}, \citenamefont {Vats}, \citenamefont {Johansson},
  \citenamefont {Pedregosa}, \citenamefont {Curry}, \citenamefont {Terrel},
  \citenamefont {Rou\v{c}ka}, \citenamefont {Saboo}, \citenamefont {Fernando},
  \citenamefont {Kulal}, \citenamefont {Cimrman},\ and\ \citenamefont
  {Scopatz}}]{meurer2017sympy}%
  \BibitemOpen
  \bibfield  {author} {\bibinfo {author} {\bibfnamefont {A.}~\bibnamefont
  {Meurer}}, \bibinfo {author} {\bibfnamefont {C.~P.}\ \bibnamefont {Smith}},
  \bibinfo {author} {\bibfnamefont {M.}~\bibnamefont {Paprocki}}, \bibinfo
  {author} {\bibfnamefont {O.}~\bibnamefont {\v{C}ert\'{i}k}}, \bibinfo
  {author} {\bibfnamefont {S.~B.}\ \bibnamefont {Kirpichev}}, \bibinfo {author}
  {\bibfnamefont {M.}~\bibnamefont {Rocklin}}, \bibinfo {author} {\bibfnamefont
  {A.}~\bibnamefont {Kumar}}, \bibinfo {author} {\bibfnamefont
  {S.}~\bibnamefont {Ivanov}}, \bibinfo {author} {\bibfnamefont {J.~K.}\
  \bibnamefont {Moore}}, \bibinfo {author} {\bibfnamefont {S.}~\bibnamefont
  {Singh}}, \bibinfo {author} {\bibfnamefont {T.}~\bibnamefont {Rathnayake}},
  \bibinfo {author} {\bibfnamefont {S.}~\bibnamefont {Vig}}, \bibinfo {author}
  {\bibfnamefont {B.~E.}\ \bibnamefont {Granger}}, \bibinfo {author}
  {\bibfnamefont {R.~P.}\ \bibnamefont {Muller}}, \bibinfo {author}
  {\bibfnamefont {F.}~\bibnamefont {Bonazzi}}, \bibinfo {author} {\bibfnamefont
  {H.}~\bibnamefont {Gupta}}, \bibinfo {author} {\bibfnamefont
  {S.}~\bibnamefont {Vats}}, \bibinfo {author} {\bibfnamefont {F.}~\bibnamefont
  {Johansson}}, \bibinfo {author} {\bibfnamefont {F.}~\bibnamefont
  {Pedregosa}}, \bibinfo {author} {\bibfnamefont {M.~J.}\ \bibnamefont
  {Curry}}, \bibinfo {author} {\bibfnamefont {A.~R.}\ \bibnamefont {Terrel}},
  \bibinfo {author} {\bibfnamefont {v.}~\bibnamefont {Rou\v{c}ka}}, \bibinfo
  {author} {\bibfnamefont {A.}~\bibnamefont {Saboo}}, \bibinfo {author}
  {\bibfnamefont {I.}~\bibnamefont {Fernando}}, \bibinfo {author}
  {\bibfnamefont {S.}~\bibnamefont {Kulal}}, \bibinfo {author} {\bibfnamefont
  {R.}~\bibnamefont {Cimrman}},\ and\ \bibinfo {author} {\bibfnamefont
  {A.}~\bibnamefont {Scopatz}},\ }\bibfield  {title} {\bibinfo {title} {Sympy:
  symbolic computing in python},\ }\href {https://doi.org/10.7717/peerj-cs.103}
  {\bibfield  {journal} {\bibinfo  {journal} {PeerJ Computer Science}\ }\textbf
  {\bibinfo {volume} {3}},\ \bibinfo {pages} {e103} (\bibinfo {year}
  {2017})}\BibitemShut {NoStop}%
\bibitem [{\citenamefont {Wimmer}(2011)}]{wimmer2011efficient}%
  \BibitemOpen
  \bibfield  {author} {\bibinfo {author} {\bibfnamefont {M.}~\bibnamefont
  {Wimmer}},\ }\bibfield  {title} {\bibinfo {title} {Efficient numerical
  computation of the pfaffian for dense and banded skew-symmetric matrices},\
  }\bibfield  {journal} {\bibinfo  {journal} {ACM Transactions on Mathematical
  Software}\ }\textbf {\bibinfo {volume} {38}},\ \href
  {https://doi.org/10.1145/2331130.2331138} {10.1145/2331130.2331138} (\bibinfo
  {year} {2011})\BibitemShut {NoStop}%
\bibitem [{\citenamefont {Hehre}\ \emph {et~al.}(1969)\citenamefont {Hehre},
  \citenamefont {Stewart},\ and\ \citenamefont {Pople}}]{hehre1969a}%
  \BibitemOpen
  \bibfield  {author} {\bibinfo {author} {\bibfnamefont {W.~J.}\ \bibnamefont
  {Hehre}}, \bibinfo {author} {\bibfnamefont {R.~F.}\ \bibnamefont {Stewart}},\
  and\ \bibinfo {author} {\bibfnamefont {J.~A.}\ \bibnamefont {Pople}},\
  }\bibfield  {title} {\bibinfo {title} {Self-consistent molecular-orbital
  methods. i. use of gaussian expansions of slater-type atomic orbitals},\
  }\href {https://doi.org/10.1063/1.1672392} {\bibfield  {journal} {\bibinfo
  {journal} {J. Chem. Phys.}\ }\textbf {\bibinfo {volume} {51}},\ \bibinfo
  {pages} {2657} (\bibinfo {year} {1969})}\BibitemShut {NoStop}%
\bibitem [{\citenamefont {Hehre}\ \emph {et~al.}(1970)\citenamefont {Hehre},
  \citenamefont {Ditchfield}, \citenamefont {Stewart},\ and\ \citenamefont
  {Pople}}]{hehre1970a}%
  \BibitemOpen
  \bibfield  {author} {\bibinfo {author} {\bibfnamefont {W.~J.}\ \bibnamefont
  {Hehre}}, \bibinfo {author} {\bibfnamefont {R.}~\bibnamefont {Ditchfield}},
  \bibinfo {author} {\bibfnamefont {R.~F.}\ \bibnamefont {Stewart}},\ and\
  \bibinfo {author} {\bibfnamefont {J.~A.}\ \bibnamefont {Pople}},\ }\bibfield
  {title} {\bibinfo {title} {Self-consistent molecular orbital methods. iv. use
  of gaussian expansions of slater-type orbitals. extension to second-row
  molecules},\ }\href {https://doi.org/10.1063/1.1673374} {\bibfield  {journal}
  {\bibinfo  {journal} {J. Chem. Phys.}\ }\textbf {\bibinfo {volume} {52}},\
  \bibinfo {pages} {2769} (\bibinfo {year} {1970})}\BibitemShut {NoStop}%
\bibitem [{\citenamefont {Ditchfield}\ \emph {et~al.}(1971)\citenamefont
  {Ditchfield}, \citenamefont {Hehre},\ and\ \citenamefont
  {Pople}}]{ditchfield1971a}%
  \BibitemOpen
  \bibfield  {author} {\bibinfo {author} {\bibfnamefont {R.}~\bibnamefont
  {Ditchfield}}, \bibinfo {author} {\bibfnamefont {W.~J.}\ \bibnamefont
  {Hehre}},\ and\ \bibinfo {author} {\bibfnamefont {J.~A.}\ \bibnamefont
  {Pople}},\ }\bibfield  {title} {\bibinfo {title} {Self-consistent
  molecular-orbital methods. ix. an extended gaussian-type basis for
  molecular-orbital studies of organic molecules},\ }\href
  {https://doi.org/10.1063/1.1674902} {\bibfield  {journal} {\bibinfo
  {journal} {J. Chem. Phys.}\ }\textbf {\bibinfo {volume} {54}},\ \bibinfo
  {pages} {724} (\bibinfo {year} {1971})}\BibitemShut {NoStop}%
\bibitem [{\citenamefont {Francl}\ \emph {et~al.}(1982)\citenamefont {Francl},
  \citenamefont {Pietro}, \citenamefont {Hehre}, \citenamefont {Binkley},
  \citenamefont {Gordon}, \citenamefont {DeFrees},\ and\ \citenamefont
  {Pople}}]{francl1982a}%
  \BibitemOpen
  \bibfield  {author} {\bibinfo {author} {\bibfnamefont {M.~M.}\ \bibnamefont
  {Francl}}, \bibinfo {author} {\bibfnamefont {W.~J.}\ \bibnamefont {Pietro}},
  \bibinfo {author} {\bibfnamefont {W.~J.}\ \bibnamefont {Hehre}}, \bibinfo
  {author} {\bibfnamefont {J.~S.}\ \bibnamefont {Binkley}}, \bibinfo {author}
  {\bibfnamefont {M.~S.}\ \bibnamefont {Gordon}}, \bibinfo {author}
  {\bibfnamefont {D.~J.}\ \bibnamefont {DeFrees}},\ and\ \bibinfo {author}
  {\bibfnamefont {J.~A.}\ \bibnamefont {Pople}},\ }\bibfield  {title} {\bibinfo
  {title} {Self-consistent molecular orbital methods. xxiii. a
  polarization-type basis set for second-row elements},\ }\href
  {https://doi.org/10.1063/1.444267} {\bibfield  {journal} {\bibinfo  {journal}
  {J. Chem. Phys.}\ }\textbf {\bibinfo {volume} {77}},\ \bibinfo {pages} {3654}
  (\bibinfo {year} {1982})}\BibitemShut {NoStop}%
\bibitem [{\citenamefont {Gordon}\ \emph {et~al.}(1982)\citenamefont {Gordon},
  \citenamefont {Binkley}, \citenamefont {Pople}, \citenamefont {Pietro},\ and\
  \citenamefont {Hehre}}]{gordon1982a}%
  \BibitemOpen
  \bibfield  {author} {\bibinfo {author} {\bibfnamefont {M.~S.}\ \bibnamefont
  {Gordon}}, \bibinfo {author} {\bibfnamefont {J.~S.}\ \bibnamefont {Binkley}},
  \bibinfo {author} {\bibfnamefont {J.~A.}\ \bibnamefont {Pople}}, \bibinfo
  {author} {\bibfnamefont {W.~J.}\ \bibnamefont {Pietro}},\ and\ \bibinfo
  {author} {\bibfnamefont {W.~J.}\ \bibnamefont {Hehre}},\ }\bibfield  {title}
  {\bibinfo {title} {Self-consistent molecular-orbital methods. 22. small
  split-valence basis sets for second-row elements},\ }\href
  {https://doi.org/10.1021/ja00374a017} {\bibfield  {journal} {\bibinfo
  {journal} {J. Am. Chem. Soc.}\ }\textbf {\bibinfo {volume} {104}},\ \bibinfo
  {pages} {2797} (\bibinfo {year} {1982})}\BibitemShut {NoStop}%
\bibitem [{\citenamefont {Hehre}\ \emph {et~al.}(1972)\citenamefont {Hehre},
  \citenamefont {Ditchfield},\ and\ \citenamefont {Pople}}]{hehre1972a}%
  \BibitemOpen
  \bibfield  {author} {\bibinfo {author} {\bibfnamefont {W.~J.}\ \bibnamefont
  {Hehre}}, \bibinfo {author} {\bibfnamefont {R.}~\bibnamefont {Ditchfield}},\
  and\ \bibinfo {author} {\bibfnamefont {J.~A.}\ \bibnamefont {Pople}},\
  }\bibfield  {title} {\bibinfo {title} {Self-consistent molecular orbital
  methods. xii. further extensions of gaussian-type basis sets for use in
  molecular orbital studies of organic molecules},\ }\href
  {https://doi.org/10.1063/1.1677527} {\bibfield  {journal} {\bibinfo
  {journal} {J. Chem. Phys.}\ }\textbf {\bibinfo {volume} {56}},\ \bibinfo
  {pages} {2257} (\bibinfo {year} {1972})}\BibitemShut {NoStop}%
\bibitem [{\citenamefont {Kendall}\ \emph {et~al.}(1992)\citenamefont
  {Kendall}, \citenamefont {Dunning},\ and\ \citenamefont
  {Harrison}}]{kendall1992a}%
  \BibitemOpen
  \bibfield  {author} {\bibinfo {author} {\bibfnamefont {R.~A.}\ \bibnamefont
  {Kendall}}, \bibinfo {author} {\bibfnamefont {T.~H.}\ \bibnamefont
  {Dunning}},\ and\ \bibinfo {author} {\bibfnamefont {R.~J.}\ \bibnamefont
  {Harrison}},\ }\bibfield  {title} {\bibinfo {title} {Electron affinities of
  the first-row atoms revisited. systematic basis sets and wave functions},\
  }\href {https://doi.org/10.1063/1.462569} {\bibfield  {journal} {\bibinfo
  {journal} {J. Chem. Phys.}\ }\textbf {\bibinfo {volume} {96}},\ \bibinfo
  {pages} {6796} (\bibinfo {year} {1992})}\BibitemShut {NoStop}%
\bibitem [{\citenamefont {Woon}\ and\ \citenamefont
  {Dunning}(1993)}]{woon1993a}%
  \BibitemOpen
  \bibfield  {author} {\bibinfo {author} {\bibfnamefont {D.~E.}\ \bibnamefont
  {Woon}}\ and\ \bibinfo {author} {\bibfnamefont {T.~H.}\ \bibnamefont
  {Dunning}},\ }\bibfield  {title} {\bibinfo {title} {Gaussian basis sets for
  use in correlated molecular calculations. iii. the atoms aluminum through
  argon},\ }\href {https://doi.org/10.1063/1.464303} {\bibfield  {journal}
  {\bibinfo  {journal} {J. Chem. Phys.}\ }\textbf {\bibinfo {volume} {98}},\
  \bibinfo {pages} {1358} (\bibinfo {year} {1993})}\BibitemShut {NoStop}%
\bibitem [{\citenamefont {Malone}(2023)}]{malone_qcqmc_data}%
  \BibitemOpen
  \bibfield  {author} {\bibinfo {author} {\bibfnamefont {F.}~\bibnamefont
  {Malone}},\ }\bibfield  {title} {\bibinfo {title} {{Data for Unbiasing
  fermionic quantum Monte Carlo with a quantum computer}},\ }\href
  {https://doi.org/10.5281/zenodo.10141262} {10.5281/zenodo.10141262} (\bibinfo
  {year} {2023})\BibitemShut {NoStop}%
\end{thebibliography}%

\clearpage
\appendix
\counterwithin{figure}{section}
\counterwithin{table}{section}

\section{Theoretical Background of Coupled Cluster Methods}\label{apx:tcc_eccc_theory}
The single reference coupled cluster (SRCC) wavefunction
is constructed through an exponential ansatz acting on the Hartree-Fock (or Fermi vacuum) reference determinant \hfdet,\cite{crawford2006introduction,shavitt2009many}
\begin{align}
\ccwfn = e^{\hat T} \hfdet,
\end{align}
with the cluster operator $\hat T$. The cluster operator is formed by
linear combination of individual operators
\begin{align}
\hat T = \sum_\nu \hat T_\nu \label{eq:cc_t_operator_general}
\end{align}
with excitation level $\nu$.
Truncating this expression at a certain excitation level $\nu_\mathrm{max}$ yields the well-known hierarchy of truncated CC methods:
CC singles (CCS, $\nu_\mathrm{max} = 1$), CC singles and doubles (CCSD, $\nu_\mathrm{max} = 2$), CC singles, doubles, and triples (CCSDT, $\nu_\mathrm{max} = 3$), and so on.
The excitation operators are generally defined as
\begin{align}
\hat T_\nu = \frac{1}{(\nu!)^2} \sum_{\substack{abc\ldots \\ ijk\ldots}} t_{ijk\ldots}^{abc\ldots} \cre{a} \cre{b} \cre{c} \ldots \ani{i} \ani{j} \ani{k} \ldots, \label{eq:cluster_operator_general}
\end{align}
with the T$\nu$ amplitudes $t_{ijk\ldots}^{abc\ldots}$ for a given excitation level $\nu$
and the fermionic creation and annihilation operators \cre{p} and \ani{p}, respectively.
The indices $a,b,c,\ldots$ denote virtual orbital indices, $i,j,k,\ldots$ refer to occupied orbitals in \hfdet, and $p,q,r,s,\ldots$ refer to general spin orbitals.
The T amplitudes are solved for by projecting excited determinant manifolds ${\ket{\Phi_\nu}}$ onto the similarity-transformed Hamiltonian $\simH \equiv e^{-\hat T} \hat H e^{\hat T}$,
\begin{align}
    \mathbf{r}_\nu \equiv \braket{\Phi_\nu |\simH | \Phi_0} \stackrel{!}{=} 0,
\end{align}
where $\mathbf{r}_\nu$ is referred to as residual vector.
The resulting non-linear amplitude equations can be solved iteratively. Note that, due to
the similarity transformation required for obtaining computationally tractable amplitude equations,
the resulting Hamiltonian is no longer Hermitian.
By projecting with \hfdet~from the left, one obtains the SRCC energy as
\begin{align}
E_\mathrm{CC} =& \braket{\Phi_0|\simH|\Phi_0} \nonumber \\
              =& E_\mathrm{HF} + \sum_{ia} f_{i}^{a} t_{i}^{a} + \frac12 \sum_{ijab} t_{i}^{a} t_{j}^b \eri{ij}{ab} + \frac14 \sum_{ijab} t_{ij}^{ab} \eri{ij}{ab},
\end{align}
with the reference Hartree-Fock energy $E_\mathrm{HF}$, the Fock matrix elements $f_p^q$, and the antisymmetrized two-electron repulsion integrals in Physicists' notation $\eri{pq}{rs}$. Note that only T1 and T2 amplitudes enter the SRCC energy expression directly, independent of the
truncation level, since the higher excitation cluster operators cannot produce fully contracted
terms with the Hamiltonian. The implicit energy contribution of higher-order T amplitudes originates from
the coupling of all amplitudes through the projection equations.
For CCSD, the amplitude equations require projections
of the singly and doubly excited determinants, i.e.,
\begin{align}
0 &= \braket{\Phi_{i}^{a}|\simH|\Phi_0},\\
0 &= \braket{\Phi_{ij}^{ab}|\simH|\Phi_0}.
\end{align}
In the following sections, we use the normal-ordered Hamiltonian,
\begin{align}
\hat H_N = \hat H - \braket{\Phi_0| \hat H |\Phi_0},
\end{align}
to express the Baker-Campbell-Hausdorff (BCH) expansion of $\simH$ more conveniently as
\begin{align}
    \simH \rightarrow (\hat H_N \hat T)_c,
\end{align}
where the $(\ldots)_c$ denotes that only the connected contributions from the expansion survive.\cite{crawford2006introduction}
With the programmable expressions for the projection equations at hand, which can for example be derived
through code generation (see Appendix \ref{apx:implementation}),\cite{rubin2021pdaggerq} the amplitudes are solved iteratively through standard numerical
techniques.\cite{crawford2006introduction}

\subsection{Tailored Coupled Cluster}
The tailored coupled cluster (TCC) ansatz aims to encode static correlation effects from an active space (AS) method in a SRCC wavefunction
through a split-amplitude ansatz\cite{kinoshita2005coupled}
\begin{align}
\hat T = \hat T^\mathrm{ext} + \hat T^\mathrm{AS}.
\end{align}
The amplitudes in the active space cluster operator $\hat T^\mathrm{AS}$
are extracted from an exact or approximate active space wavefunction
through the relationship of the linear configuration interaction (CI) ansatz and the exponential CC ansatz
\begin{align}
1 + \sum_\nu \hat C_\nu = e^{\hat T}.
\end{align}
The CI excitation operators $\hat C_\nu$ are defined as the cluster operator in
eq \eqref{eq:cluster_operator_general}, but with $c_{ijk\ldots}^{abc
\ldots}$ as corresponding amplitudes.
Conversion from CI to CC amplitudes is achieved by matching excitation levels and recursively determining the amplitudes,\cite{monkhorst1977calculation,lehtola2017cluster} here up to four-fold excitations, as
\begin{align}
\hat T_1 &= \hat C_1 \label{eq:t1_from_ci}\\
\hat T_2 &= \hat C_2 - \frac12 \hat{T}_1^2 \label{eq:t2_from_ci} \\
\hat T_3 &= \hat C_3 - \frac12 \left(\hat T_1 \hat T_2 + \hat T_2 \hat T_1\right) - \frac1{3!} \hat T_1^3\\
\hat T_4 &= \hat C_4 - \frac{1}{2} \left(\hat T_1 \hat T_3 + \hat T_3 \hat T_1\right) - \frac{1}{2} \hat T_2^2 - \frac{1}{4!} \hat T_1^4. \label{eq:t4_from_ci}
\end{align}
Evaluating these terms as Wick contractions yields
the following programmable expressions,\cite{lehtola2017cluster}
\begin{align}
&t_i^a = c_i^a \\
&t_{ij}^{ab} = c_{ij}^{ab} - t_i^a t_j^b + t_i^b t_j^a\\
&\begin{aligned}
t_{ijk}^{abc} =&~c_{ijk}^{abc}-t_i^at_{jk}^{bc}+t_i^bt_{jk}^{ac}-t_i^ct_{jk}^{ab}+t_j^at_{ik}^{bc}-t_j^bt_{ik}^{ac}\\ &+t_j^ct_{ik}^{ab}-t_k^at_{ij}^{bc}+t_k^bt_{ij}^{ac}-t_k^ct_{ij}^{ab}-t_i^at_j^bt_k^c\\& +t_i^at_j^ct_k^b+t_i^bt_j^at_k^c-t_i^bt_j^ct_k^a-t_i^ct_j^at_k^b+t_i^ct_j^bt_k^a
\end{aligned}\\
&t_{ijkl}^{abcd} = \ldots
\end{align}
The expressions contain single ``non-redundant'' outer products of T amplitudes, e.g., $t_{i}^{a} t_{jk}^{bc}$, and the terms with permuted indices ensure the correct anti-symmetry of the resulting amplitude, cancelling exactly the corresponding prefactor from the Taylor expansion.
If the coefficient of the reference determinant in the underlying
CI expansion is not equal to one, one has to renormalize the above
conversion equations accordingly. As all CC methods, TCC is only
applicable if the coefficient of the reference determinant
extracted from the active space method is non-zero.
With a set of CI amplitudes at hand, the corresponding spin-orbital T amplitudes can be built with the expressions above, yielding an active space cluster operator $\hat T^\mathrm{AS}$. The orbital indices of the cluster amplitudes are fully contained in the active space, i.e., $ijk\ldots abc\ldots \in \mathrm{AS}$.
Contrary to that, the \emph{external} cluster amplitudes belonging to $\hat T^\mathrm{ext}$ comprise an orbital space where at least one orbital index is not part of the active space, i.e., $ijk\ldots abc\ldots \not\subset \mathrm{AS}$.
Now, the TCC energy functional is given by
\begin{align}
    E_\mathrm{TCC} =& \Braket{\Phi_0|\left(\hat H_N e^{\hat T^\mathrm{ext} + \hat T^\mathrm{AS}}\right)_c|\Phi_0}.
\end{align}
Since the external and active space cluster operators commute by construction, one can obtain the active space energy through the CC energy functional via
\begin{align}
    E_\mathrm{TCC} &= \Braket{\Phi_0 |\left(\hat H_N e^{\hat T_\mathrm{AS}}\right)_c | \Phi_0} + \Braket{\Phi_0 | \left(\hat H_N e^{\hat T_\mathrm{ext}}\right)_c | \Phi_0} \\
    &= E_\mathrm{AS} + E_\mathrm{ext}.
\end{align}
The relationship holds due to the equivalence of the CI and CC expansions
for exact wavefunctions, i.e., the T1 and T2 amplitudes extracted from the active space wavefunction are exact and can be viewed as
optimized in presence of all higher-order T amplitudes through the active space method. For approximate active space methods, the mapping is not exact, and the resulting energy computed through the CC energy functional
is not necessarily equivalent to the (variational) energy of the active space wavefunction.\footnote{
The energies are equivalent though if singles and doubles are treated fully variationally in the approximate active space wavefunction.
}
To retain the static correlation information in the TCC wavefunction, the active space amplitudes are kept frozen during optimization of the external amplitudes. That is, the following amplitude equations are to
be solved in TCCSD, the most popular flavor of TCC:
\begin{align}
    0 =& \Braket{\Phi_{i}^{a}| \left(\hat H_N e^{(\hat T_1^\mathrm{AS} + \hat T_2^\mathrm{AS})} e^{(\hat T_1^\mathrm{ext} + \hat T_2^\mathrm{ext})}\right)_c| \Phi_0}, \quad i,a \not\subset \mathrm{AS} \\
    0 =& \Braket{\Phi_{ij}^{ab}| \left(\hat H_N e^{(\hat T_1^\mathrm{AS} + \hat T_2^\mathrm{AS})} e^{(\hat T_1^\mathrm{ext} + \hat T_2^\mathrm{ext})}\right)_c | \Phi_0}, \quad i,j,a,b \not\subset \mathrm{AS}.
\end{align}
Note that the ``output'' amplitudes are external, however, the active space amplitudes still appear in the algebraic expressions of the projection equations.
Thus, the frozen active space T amplitudes impact the TCC solution
through the active space energy, i.e., the energy evaluated directly from the frozen amplitudes, and additionally through the contractions with
the external amplitudes.
Once the active space T amplitudes are known and mapped to the corresponding orbitals in the full molecular orbital space, the standard
SRCC framework is used to solve for the external amplitudes, with the only
constraint that a certain stride of the full space amplitudes be frozen.
In practice, this can be easily achieved by setting the stride of active
space amplitudes in the CC residual vector to zero.
As for standard CCSD, a perturbative triples correction, i.e., CCSD(T)
can be employed,\cite{lyakh2011tailored} which is only evaluated from
the external amplitudes for consistency.

\subsection{Externally Corrected Coupled Cluster}
Externally corrected CC (ec-CC)\cite{paldus2017externally} builds a similar split-amplitude ansatz as TCC, but aims to encode static correlation into
the SRCC wavefunction through higher-order cluster operators.
Inspecting the un-truncated expressions for the singles and doubles residual equations, it is clear that only
$\hat T_3$ and $\hat T_4$ can contribute to achieve the correct total excitation level, i.e.,
\begin{align}
r_{i}^{a} \leftarrow \Braket{ \Phi_{i}^{a} |\left(\hat H_{N} \hat T_{3}\right)_{c}|\Phi_0}, \\
r_{ij}^{ab} \leftarrow \Braket{ \Phi_{ij}^{ab} |\left(\hat H_{N} \hat T_{3} + \hat T_{4} + \hat T_{1}\hat T_{3}\right)_{c}|\Phi_0}.
\end{align}
These expressions then, of course,
correspond to the CCSDTQ residual equations for T1 and T2, however, higher-order cluster operators cannot produce any further
contributions to these terms.
Hence, the traditional CCSD approach corresponds to the approximate case where $\hat T_3 = 0$ and $\hat T_4 = 0$.
This means that, adding the direct contributions of T3 and T4 amplitudes (e.g., extracted from a high-quality correlated multi-reference
wavefunction)
to the T1 and T2 projection equations, one can in principle obtain an improved treatment of the overall electron correlation, including
non-dynamic correlation effects, compared to CCSD. The gist here is that convergence to full CI is \emph{guaranteed} if the input T3 and T4
amplitudes become exact. Thus, the ec-CC procedure optimizes the T1 and T2 amplitudes in presence of (approximate) T3 and T4.
The resulting ec-CC cluster operator is given by
\begin{align}
    \hat T = \hat T_1 + \hat T_2 + \hat T_3^{\mathrm{input}} + \hat T_4^{\mathrm{input}},
\end{align}
where the T3 and T4 amplitudes are extracted from some input wavefunction. If the wavefunction only covers a subset
of the molecular orbitals, i.e., in case of an active space method, the operator is modified to project the
sliced T3 and T4 amplitudes onto the full orbital space, like in TCC, that is
\begin{align}
    \hat T = \hat T_1 + \hat T_2 + \hat P_3 \hat T_3^{\mathrm{AS}} + \hat P_4 \hat T_4^{\mathrm{AS}},
\end{align}
with the appropriate projection operators $\hat P_3$ and $\hat P_4$. For ec-CC, the same recursion relations to extract
T amplitudes from a CI-like wavefunction apply, see eqs \eqref{eq:t1_from_ci} through \eqref{eq:t4_from_ci}. If an active space wavefunction is used as input,
the convergence guarantee toward FCI does not hold anymore, however, improvements
can still be achieved by including the dominant T3 and T4 amplitudes from the
multi-reference treatment compared to completely neglecting the triples and quadruples.
The structure of the ec-CC ansatz directly implies that, depending on the quality of the underlying input wavefunction, the ec-CC
result might not always provide better results than the input wavefunction itself.\cite{magoulas2021is}
We reproduce the corresponding proof in the following Section~\ref{app:eccc_proof}.
\subsection{Type-I ec-CC with Configuration Interaction}\label{app:eccc_proof}
This section largely follows the arguments laid out in Ref.~\citenum{magoulas2021is} with proofs abbreviated for conciseness. The ec-CC singles and doubles equations with T3 and T4 generated from any coupled-cluster theory involving anything beyond doubles would return an identical coupled-cluster energy. Naturally, one would get a different energy if T1 and T2 from the T3 and T4 source calculation are not treated fully. It turns out that truncated configuration interaction can be shown to satisfy the ec-CC T1 and T2 equations.

To see this we will convert the set of singles and doubles truncated CI correlation energy equations into a form equivalent to the coupled-cluster singles and doubles equations. Thus satisfying the truncated CI equations for singles and doubles corresponds to solving the ec-CC equations singles and doubles equations.  This implies that the energy would be equivalent.
\begin{theorem}
The solution to a truncated configuration interaction set of equations which includes full singles, full doubles, and any set of higher excitations satisfies the coupled-cluster singles and doubles equations.
\end{theorem}
\begin{proof}
First, consider a truncated CI wavefunction expansion represented in intermediate normalization
\begin{align}
\ket{\Psi} = \left(1 + \hat C_{1} + \hat C_{2} + \hat C^{B}\right) \hfdet
\end{align}
where $\hat C_{1}$ is the full set of one-body excitations and their coefficients, $\hat C_{2}$ is the full set of two-body excitations and their variational coefficients, $\hat C^{B}$ is any higher excitations and their variational coefficients, and $\hfdet$~ is the reference determinant. We note that the set $\hat C^{B}$ does not need to be complete. For example, only subsets of triples or quadruples could be considered in $\hat C^{B}$. The coupled equations for the CI correlation energy for singles and doubles are 
\begin{align}
\braket{\Phi_{i}^{a}| \hat H_N \left( 1 + \hat C_{1} + \hat C_{2} + \hat C^{B}\right)| \Phi_0 } &= E_{\rm corr} \braket{\Phi_{i}^{a}| \hat C_1 |\Phi_0},\\
\braket{\Phi_{ij}^{ab}| \hat H_N \left( 1 + \hat C_{1} + \hat C_{2} + \hat C^{B}\right)| \Phi_0 } &= E_{\rm corr} \braket{\Phi_{ij}^{ab}| \hat C_2 |\Phi_0}.
\end{align}
where $E_{\rm corr} = \braket{\Phi_0| \hat H_N \left(\hat C_{1} + \hat C_{2}\right)|\Phi_0}$, and $|\Phi_{i}^{a}\rangle$ and $|\Phi_{ij}^{ab}\rangle$ are the singly and doubly excited determinants, respectively. Noting that the connected components of the wavefunction are given by
\begin{align}
\hat T = \mathrm{ln}(1 + \hat C_{1} + \hat C_{2} + \hat C^{B})
\end{align}
which is evaluated using the Mercator series and is a specific case of the more general cluster analysis equation $T = \mathrm{ln}(1 + C)$ can be used to write the CI equations in exponential ansatz form
\begin{align}
\langle \Phi_{i}^{a}|e^{\hat T}e^{-\hat T}\hat H_{N}e^{\hat T} |\Phi_0\rangle = \langle \Phi_{i}^{a}|e^{\hat T}\left(\hat H_{N}e^{\hat T}\right)_{c} |\Phi_0\rangle  = E_{\rm corr} \langle \Phi_{i}^{a}|\hat C_{1}|\Phi_0\rangle, \\
\langle \Phi_{ij}^{ab}|e^{\hat T}e^{-\hat T}\hat H_{N}e^{\hat T} |\Phi_0\rangle = \langle \Phi_{ij}^{ab}|e^{\hat T}\left(\hat H_{N}e^{\hat T}\right)_{c} |\Phi_0\rangle = E_{\rm corr} \langle \Phi_{ij}^{ab}|\hat C_{2}|\Phi_0\rangle.
\end{align}
Now introducing projectors on to each relevant subspace of Hilbert space
\begin{align}
P + P_{\tilde{C}_{1},\tilde{C}_{2}} + P_{\tilde{C}^{B}} + Q = \mathbb{1} \\
P = |\Phi_0\rangle\langle\Phi_0| \;\;, \;\; P_{\tilde{C}_{1},\tilde{C}_{2}} = \tilde{C}_{1}|\Phi_0\rangle\langle\Phi_0|\tilde{C}_{1}^{\dagger} + \tilde{C}_{2}|\Phi_0\rangle\langle\Phi_0|\tilde{C}_{2}^{\dagger}  \;\; P_{\tilde{C}^{B}} = \tilde{C}^{B}|\Phi_0\rangle\langle \Phi_0|\tilde{C}_{B}^{\dagger} 
\end{align}
where $Q$ is defined as the orthogonal compliment to $P + P_{\tilde{C}_{1},\tilde{C}_{2}} + P_{\tilde{C}^{B}}$ and the tilde on the $C$ indicates the excitation operators of $\hat C$ without the variational coefficients. Inserting the resolution of identity we get
\begin{align}
\langle \Phi_{i}^{a}|e^{\hat T}\left( P + P_{\tilde{C}_{1},\tilde{C}_{2}} + P_{\tilde{C}^{B}} + Q \right)\left(\hat H_{N}e^{\hat T}\right)_{c} |\Phi_0\rangle  =&~ E_{\rm corr} \langle \Phi_{i}^{a}|\hat C_{1}|\Phi_0\rangle  \\
\langle \Phi_{i}^{a}|e^{\hat T} P\left(H_{N}e^{T}\right)_{c} |\Phi\rangle + \langle \Phi_{i}^{a}|e^{T} P_{\tilde{C}_{1},\tilde{C}_{2}}\left(\hat H_{N}e^{\hat T}\right)_{c} |\Phi_0\rangle  +& \nonumber \\
\langle \Phi_{i}^{a}|e^{\hat T} P_{\tilde{C}^{B}}\left(\hat H_{N}e^{\hat T}\right)_{c} |\Phi_0\rangle  +\langle \Phi_{i}^{a}|e^{\hat T} Q \left(\hat H_{N}e^{\hat T}\right)_{c} |\Phi_0\rangle  =&~
  E_{\rm corr} \langle \Phi_{i}^{a}|\hat C_{1}|\Phi_0\rangle  \label{eq:expanded_res_identity} \\
  \langle \Phi_{i}^{a}|e^{\hat T} P_{\tilde{C}_{1},\tilde{C}_{2}}\left(\hat H_{N}e^{\hat T}\right)_{c} |\Phi_0\rangle =&~ 0 \label{eq:singles_projector_penultimate}
\end{align}
where the first term on the left-hand-side of Eq.~\eqref{eq:expanded_res_identity} cancels the right-hand-side and the last two terms on the left-hand-side are zero because the cluster operator only excites from a given determinant. Following the same protocol we obtain a similar expression for the doubles equation which is 
\begin{align}
  \langle \Phi_{ij}^{ab}|e^{\hat T} P_{\tilde{C}_{1},\tilde{C}_{2}}\left(\hat H_{N}e^{\hat T}\right)_{c} |\Phi_0\rangle = 0.
\end{align}
Expanding the projector in Eq.~\eqref{eq:singles_projector_penultimate} and noting that the $C_{1}$ operator only excites from a determinant we obtain the singles coupled-cluster equation $\langle \Phi_{i}^{a}|\left(\hat H_{N}e^{\hat T}\right)_{c}|\Phi_0\rangle = 0$. Expanding the doubles equation in the projectors and the exponential operator one obtains the doubles coupled-cluster equation $\langle \Phi_{ij}^{ab}|\left(\hat H_{N}e^{\hat T}\right)_{c}|\Phi_0 \rangle = 0$.  
\end{proof}
It should be re-emphasized that the full singles and doubles are required for this proof. If a subset of the singles are used then the last term on the left-hand-side of Eq.~\eqref{eq:singles_projector_penultimate} cannot necessarily be concluded to be zero.  The same can be said for the analogous term in the doubles derivation.  As is concluded in Ref.~\citenum{magoulas2021is}, the way to break the truncated CI coupled-cluster energy symmetry is to modify the cluster equations to remove disconnected components of T3 and T4. 

\subsection{Computational Scaling and Simplification of ec-CC} \label{apx:ec_cc_scaling}
The ec-CC ansatz borrows only the T1 and T2 residual equations from CCSDTQ, i.e., the expensive T3 and T4 residual equations
are not required. In addition, the CI expansion of the input wavefunction needs to be converted to the corresponding
T amplitudes, requiring outer products of T amplitudes that produce up to eight-index output tensors. This scales as $\mathcal{O}(N^8)$,
but the prefactor is quite large in that case because all non-redundant anti-symmetric permutations need to be computed through
transpositions of each unique outer product. However, the conversion only occurs once per calculation and is thus not repeatedly evaluated.
The most expensive contraction for ec-CC occurs in the T2 residual equation as the product of T4 amplitudes and the electron-repulsion intergrals,
\begin{align}
    r_{ij}^{ab} \leftarrow \sum_{klcd} t_{ijlk}^{cdab} \eri{lk}{cd}.
\end{align}
Since T4 is purely dictated by the input wavefunction method, this $\mathcal{O}(N^8)$ contraction is executed only once and can be stored
for the rest of the ec-CC computation. Another contribution to the T2 amplitude equation comes from products of T1 and T3. To improve
performance, some works on ec-CC report that this term can be pre-computed using T1 from the input wavefunction, thereby neglecting
an iterative update of the T1/T3 contribution to the residual equations. This approximation seems to have only a minor influence on
the quality of the result.\cite{paldus2017externally,aroeira2020coupled} Our ec-CC implementation features ``frozen'' and iterative computation of
the T1/T3 products for better debugging in the context of quantum input.

\section{Implementation and Computational Methodology} \label{apx:implementation}
All classical electronic structure calculations (CASCI, CCSD, NEVPT2, etc.) were run using PySCF.\cite{sun2020recent}
Used basis sets and other parameters are specified in the main text.
Tailored and externally corrected coupled cluster methods were implemented in a standalone Python library,
using PySCF as input for reference states, e.g., molecular integrals, and classical CASCI calculations.
The cluster analysis working equations and CCSDTQ working equations for ec-CC were generated using $p^\dagger q$.\cite{rubin2021pdaggerq}
For improved performance, the tensor contractions are evaluated using JAX.\cite{jax2018github}
Correctness of our cluster analysis code was verified using the ClusterDec library.\cite{lehtola2017cluster}
For input from quantum hardware, our code supports FQE wavefunctions as external source.\cite{rubin2021fermionic}
The overlaps in case of PySCF's CASCI or FQE input are obtained by a simple look-up in the sparse state/CI vector,
addressing determinants with the correct excitation level and then converting them from true-vacuum sign convention
to Fermi vacuum sign convention.\cite{lehtola2017cluster}
VQE and matchgate shadow simulations were performed using Covestro's in-house quantum computing software stack based on PennyLane.\cite{pennylane}
Data analysis was performed using NumPy,\cite{harris2020numpy} SciPy,\cite{scipy} Pandas,\cite{reback2020pandas,mckinney-proc-scipy-2010}
Matplotlib,\cite{Hunter:2007} and Seaborn.\cite{Waskom2021}

\section{Split-Amplitude CC with Noisy Quantum Inputs} \label{sec:influence_of_noise}
In this section we investigate how the shot noise due to finite sampling influences the errors of energies obtained with the quantum TCCSD method.
We do this in two steps:
First, by means of explicit simulation of the matchgate shadow protocol, we investigate how the errors of estimated overlaps depend on the number of shots and number of orbitals.
By doing this we numerically obtain concrete formulas characterizing the performance of the matchgate shadow based overlap estimation method that complement the mathematically proven performance guarantees derived in Ref.~\citenum{wan2022matchgate}.
In particular we find that the estimated overlaps are essentially normal distributed. 
In a second step we use this knowledge to investigate the behavior of the TCCSD and ec-CC energy when Gaussian noise is added to exact (obtained with CASCI) or approximate (obtained from a simulated VQE) overlaps.

We find that our conclusions are largely independent of the precise molecular system and state used in the computations.
The data presented in the following was obtained using restricted Hartree Fock (RHF)/cc-pVDZ\cite{dunning1989a} orbitals of the \ce{N2} molecule at a bond distance of $1.09$~{\AA}.
In addition to the chemically reasonable \mbox{(6, 6)} active space we have performed simulations for other spin-zero active spaces characterized by $n$ and $\zeta$.
For all simulations of matchgate shadows the PennyLane library \cite{pennylane} was used to prepare approximate CASCI ground states on \num{8} qubits by means of the quantum number preserving fabrics\cite{anselmetti2021local} with \num{22} layers (optimized without noise under L-BFGS-B with $\Pi$ the identity gate starting from initialization method B with a small amount of normal distributed noise added to the initial parameters).
Overlap computation was done following the protocol from Ref.~\citenum{wan2022matchgate}, using Newton's polynomial interpolation for interpolating the Pfaffian polynomial, which was done with NumPy,\cite{harris2020numpy} Sympy,\cite{meurer2017sympy} and Pfapack.\cite{wimmer2011efficient}
\begin{figure}[ht]
    \centering
    \begin{subfigure}[b]{0.55\textwidth}
        \centering
        \includegraphics[width=\textwidth]{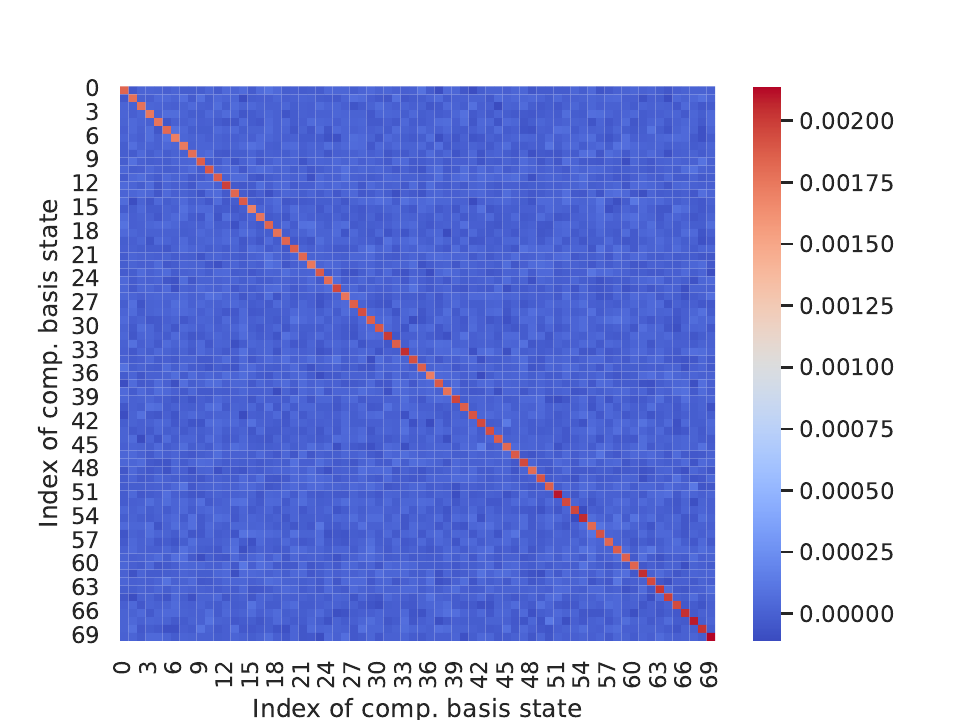}
        \caption{}
        \label{fig:overlap_convariance_matrix}
    \end{subfigure}    
    \hfill
    \begin{subfigure}[b]{0.44\textwidth}
        \centering
        \includegraphics[width=\textwidth]{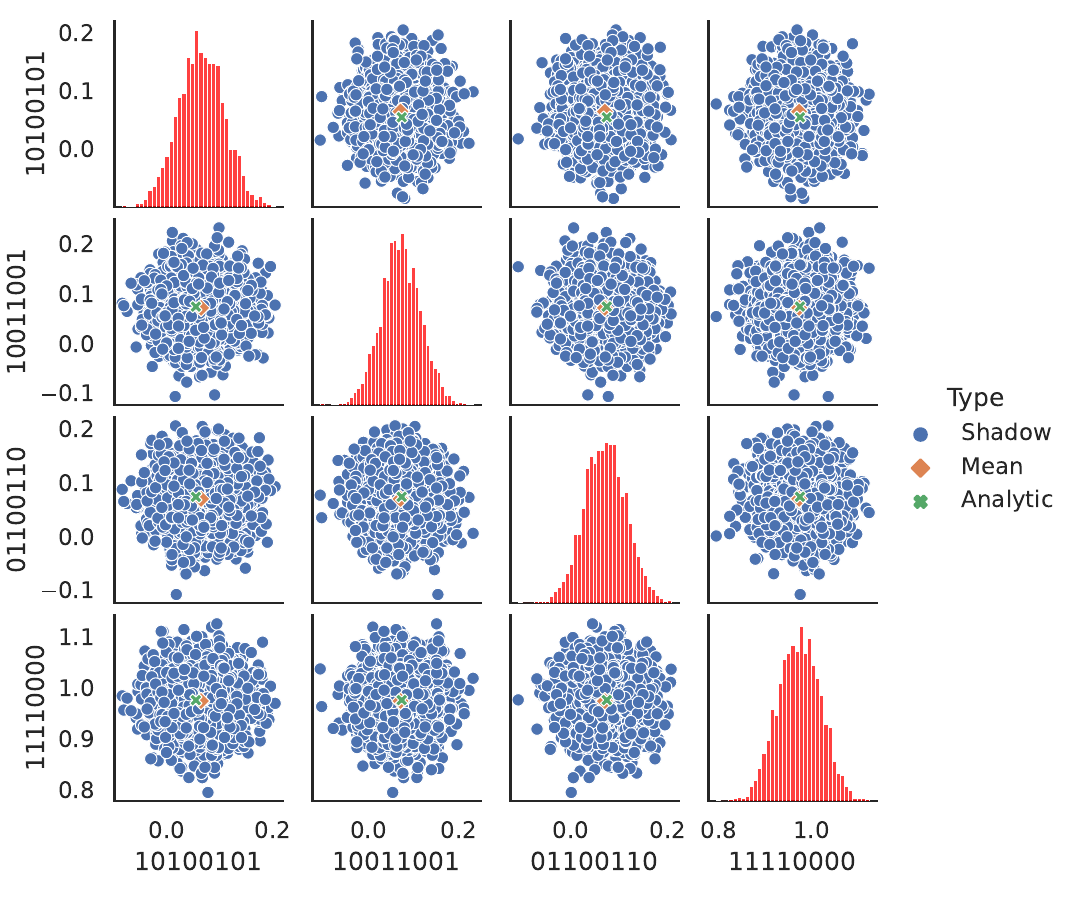}
        \caption{}
        \label{fig:overlaps_are_normal_distributed}
    \end{subfigure}
    \caption{
    Panel (\subref{fig:overlap_convariance_matrix}) shows the covariance matrix of the overlaps between all $d = 8!/(4!(8 - 4)!) = 70$ computational basis states from the half-filling subspace of $n=8$ qubits bootstrapped with $k=1000$ sub-shadows with $s = 2000$ shots each drawn randomly from a shadow of \num{200000} shots.
    Panel (\subref{fig:overlaps_are_normal_distributed}) shows on the diagonal for the \num{4} largest overlaps a histogram of the distribution of the $k$ values computed from each of the sub-shadows from which the covariance matrix was computed.
    In the scatter plots in the off diagonal positions the blue dots are the overlaps between pairs of basis states for each sub-shadow, the orange diamonds are the mean over the $k$ sub-shadows, and the green crosses are the analytically computed overlaps.
    }
    \label{fig:overlap_statistics}
\end{figure}
\subsection{Statistical Properties of Overlaps Estimated from Finite Shot Matchgate Shadows} \label{apx:stat_props}
In this section we analyze the statistical properties of overlaps estimated from finite shot matchgate shadows.
We determine variances and covariances between the estimated overlaps via bootstrapping and analyze their distributions.
We first generate a shadow of $200,000$ shots (individually evaluated circuits measured following an i.i.d.\ Gaussian unitary) and then randomly draw $k$ sub-shadows, each of size $s$, estimate all overlaps for each sub-shadow and then compute the covariance matrix of these estimates.
For simplicity we omit the median of means estimation and work with standard arithmetic means throughout.
The recorded data can be found on Zenodo.\cite{zenododata}
\begin{figure}[h]
    \centering
    \begin{subfigure}[b]{0.49\textwidth}
        \centering
        \includegraphics[width=\textwidth]{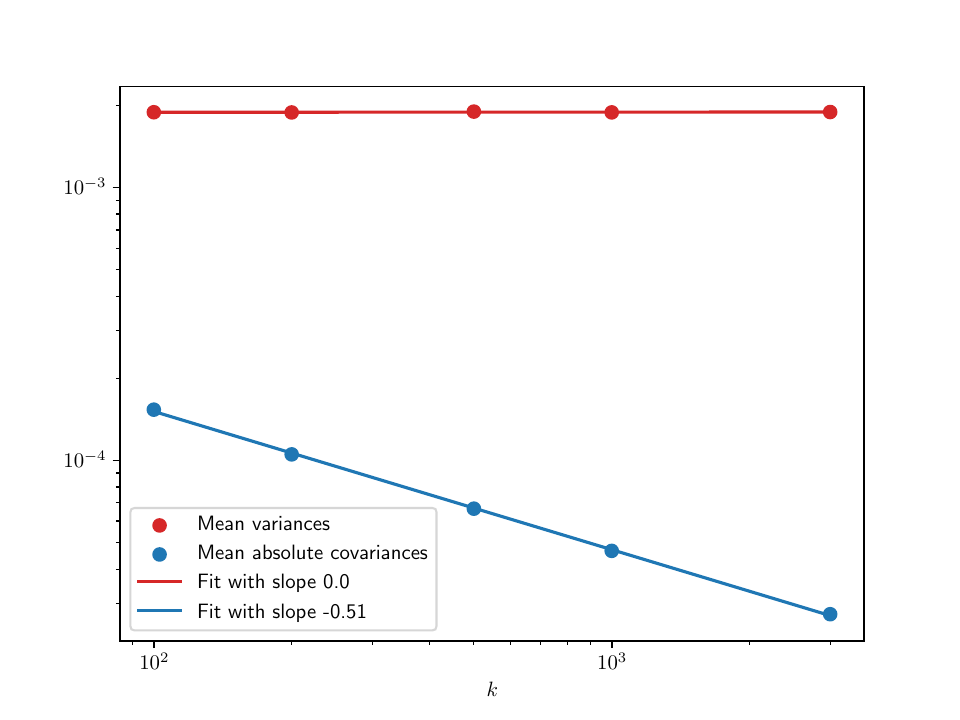}
        \caption{}
        \label{fig:k_scaling}
    \end{subfigure}
    \hfill
    \begin{subfigure}[b]{0.49\textwidth}
        \centering
        \includegraphics[width=\textwidth]{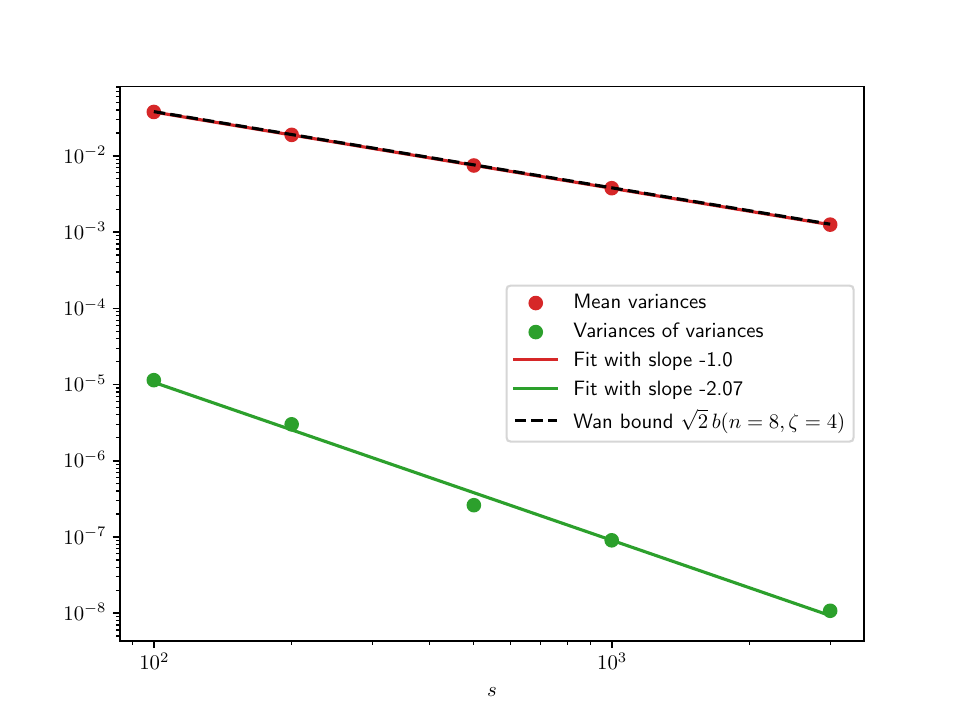}
        \caption{}
        \label{fig:s_scaling}
    \end{subfigure}
    \caption{
    Panel~(\subref{fig:k_scaling}) shows the mean of the estimated variances and the mean of the absolute values of the covariances of all basis state overlaps as a function of the number $k$ of shadows with each $s=2000$ shots that were used in the bootstrapping.
    Figure~\ref{fig:overlap_convariance_matrix} shows the estimated covariance matrix at $k=1000$ where the variances are well converged and covariances are about one order of magnitude smaller.
    Them decaying further with $k$, suggests that the actual covariances are zero and the residual values are an artifact of the finite $k$ used in their estimation.
    Panel~(\subref{fig:s_scaling}) shows the variances and the variance of their distribution (``variance of variances'') as a function of $s$ and the bound from Ref.~\citenum{wan2022matchgate}.
    }
    \label{fig:variance_of_overlaps_statistics}
\end{figure}
We find solid evidence that the overlaps with different computational basis states are approximately uncorrelated and normal distributed (Figure~\ref{fig:overlap_statistics}).
The covariance matrix is diagonally dominated with no visible structure at all in the off-diagonal entries, and the the scatter plots display no correlation.
As $s$ and $k$ grow, the bootstrapped variances increasingly concentrate around their mean and bootstrapped covariances decrease faster than the variances (Figure~\ref{fig:overlap_convariance_matrix}).

From eq (43) in Ref.~\citenum{wan2022matchgate} one obtains the bound $\sigma^2 \leq 4 \, b(n, \zeta) / s$, where the factor of $4$ is due to the additional factor of $2$ added to step $4.1.ii$ of Algorithm 1 in the \href{https://arxiv.org/pdf/2207.13723v3.pdf}{third arXiv version}, which accounts for the fact that (in the notation of Ref.~\citenum{wan2022matchgate}) $\langle \rho | (|\mathbf{0} \rangle\langle \varphi|) |\rho\rangle =  \langle \varphi | \psi \rangle/2$ and $b(\zeta, n)$ is a bound on the variance of the left hand side.
We instead find numerically that the mean variance $\bar\sigma^2$ (decaying as $1/s$, as expected) is in nearly perfect agreement with the slightly tighter bound
\begin{equation} \label{eq:shadow_overlap_error}
    \bar\sigma^2 \lessapprox \sqrt{2} \, b(n, \zeta) / s.
\end{equation}
Evaluating the bound $b(n, \zeta)$ in a numerically stable way becomes costly for large $n$ and $\zeta$.
Fortunately we find that in the half-filling case $\zeta = n/2$ we have the simple boud
\begin{equation}
    b(n, n/2) \lessapprox \sqrt{n},
\end{equation}
as can be seen from the power law fit in Figure~\ref{fig:bound_sqrt_approx}.
\begin{figure}[ht]
\centering
\includegraphics[width=0.5\textwidth]{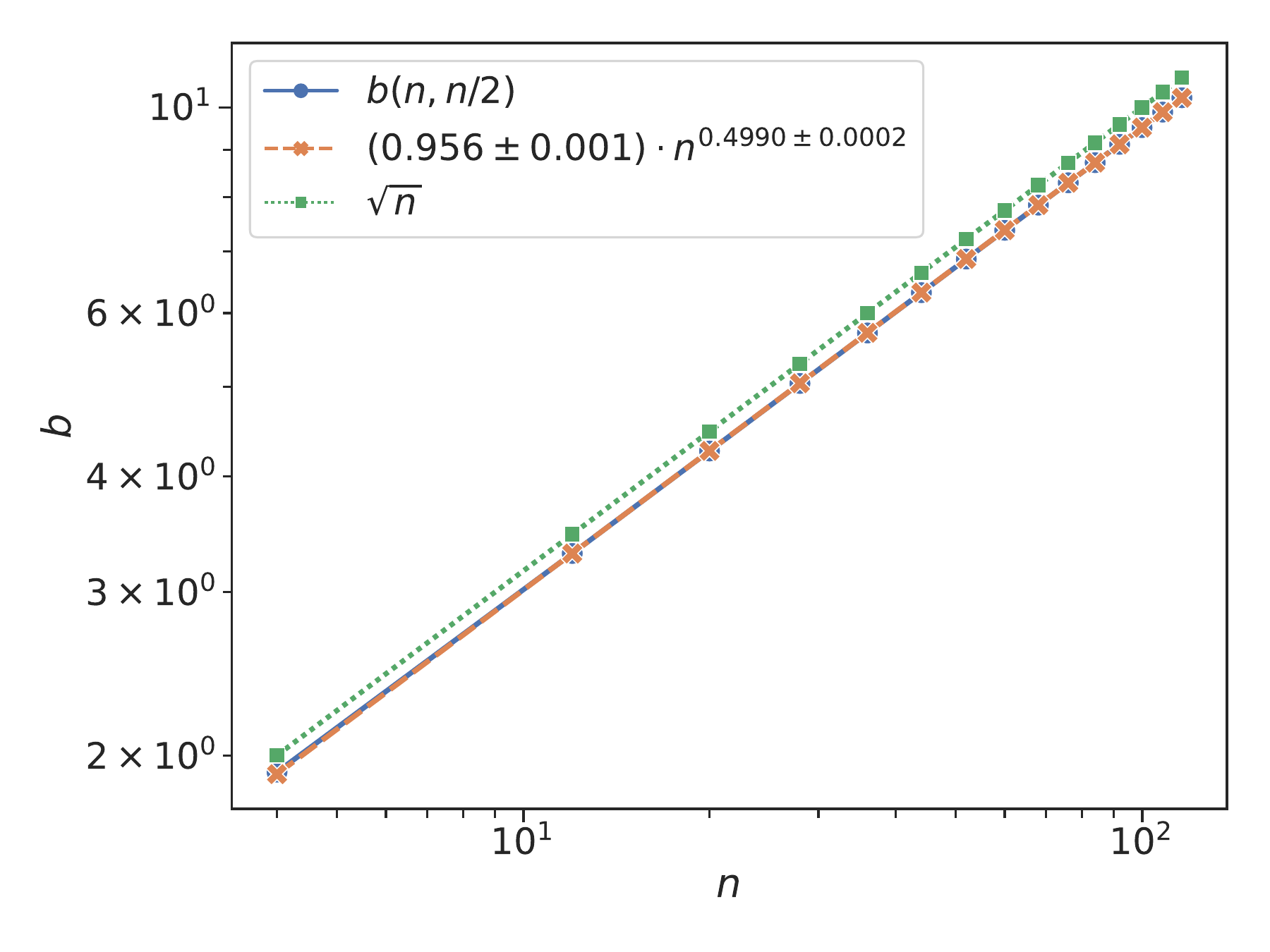}
\caption{Power law fit for $b(n, n/2)$ (eq \eqref{eq:shadow_overlap_error}), i.e., half-filling, including $\sqrt{n}$ as an upper bound for efficient numerical evaluation.
} \label{fig:bound_sqrt_approx}
\end{figure}

We have seen qualitatively identical statistical properties for other numbers of qubits and states.
This motivates the use of a Gaussian additive noise model when investigating the stability of TCCSD computations against imperfections in the input overlaps in the next section.

\subsection{TCCSD and ec-CC under Gaussian Noise} \label{sec:tcc_robustness_to_noise}
\begin{table}[ht]
\centering
\caption{Summary of active space sizes and required overlaps
$d$ used in Figure~\ref{fig:tcc_error_under_noise}.\textsuperscript{a)}} \label{tab:active_space_sigma_sampling}
\begin{tabular}{crr}
\toprule
CAS($\zeta$, $n/2$) & $d$ (TCCSD) & $d$ (ec-CC) \\
\midrule
(2, 2) & 4 & 4 \\
(4, 4) & 27 & 36 \\
(4, 6) & 93 & 225 \\
\color{red}(6, 6) & 118 & 381 \\
\color{red}(4, 8) & 199 & 784 \\
\color{red}(6, 8) & 316 & 2436 \\
\color{red}(8, 8) & 361 & 3355 \\
\color{red}(10, 10) & 876 & 21126 \\
\color{red}(12, 10) & 805 & 17255 \\
\color{red}(12, 12) & 1819 & 98694 \\
(10, 14) & 2836 & 243376 \\
(6, 16) & 2068 & 97956 \\
(8, 16) & 3193 & 285255 \\
(10, 16) & 4236 & 555336 \\
(12, 16) & 5071 & 840796 \\
\bottomrule
\end{tabular}
\begin{flushleft}
\textsuperscript{a)} Only active spaces highlighted in {\color{red}{red}} were used in the case of ec-CC.
\end{flushleft}
\end{table}
With the additive Gaussian noise model at hand, we next investigate how the energies
of the hybrid TCCSD and ec-CC method are influenced by the noise strength $\sigma$.
To this end, classical CASCI calculations with increasing active space sizes were run, the CI amplitudes were extracted and Gaussian noise (with standard deviation $\sigma$ and mean equal to zero) was added to the amplitude vectors. Subsequently, a TCCSD/ec-CC energy calculation was run for each realization of the additive noise. This procedure was conducted for $\sigma=10^{-10},\ldots,10^{-1}$ and the used active spaces are shown in Table~\ref{tab:active_space_sigma_sampling}. For each $\sigma$ and active space, the ``noisy'' TCCSD/ec-CC energy was sampled 100 times to obtain a well converged estimator for the energy error due to noise.
The data were generated using the \ce{N2} molecule
with a bond distance of $1.09$\AA~ from restricted HF/cc-pVDZ\cite{dunning1989a} reference orbitals, artificially
constructing the AS around the HOMO-LUMO gap.
Note that this procedure would be computationally prohibitive if the overlaps were actually
computed using a classical simulation of the matchgate shadow protocol.
\begin{figure}[ht]
    \centering
    \begin{subfigure}[b]{0.48\textwidth}
        \centering
        \includegraphics[width=\textwidth]{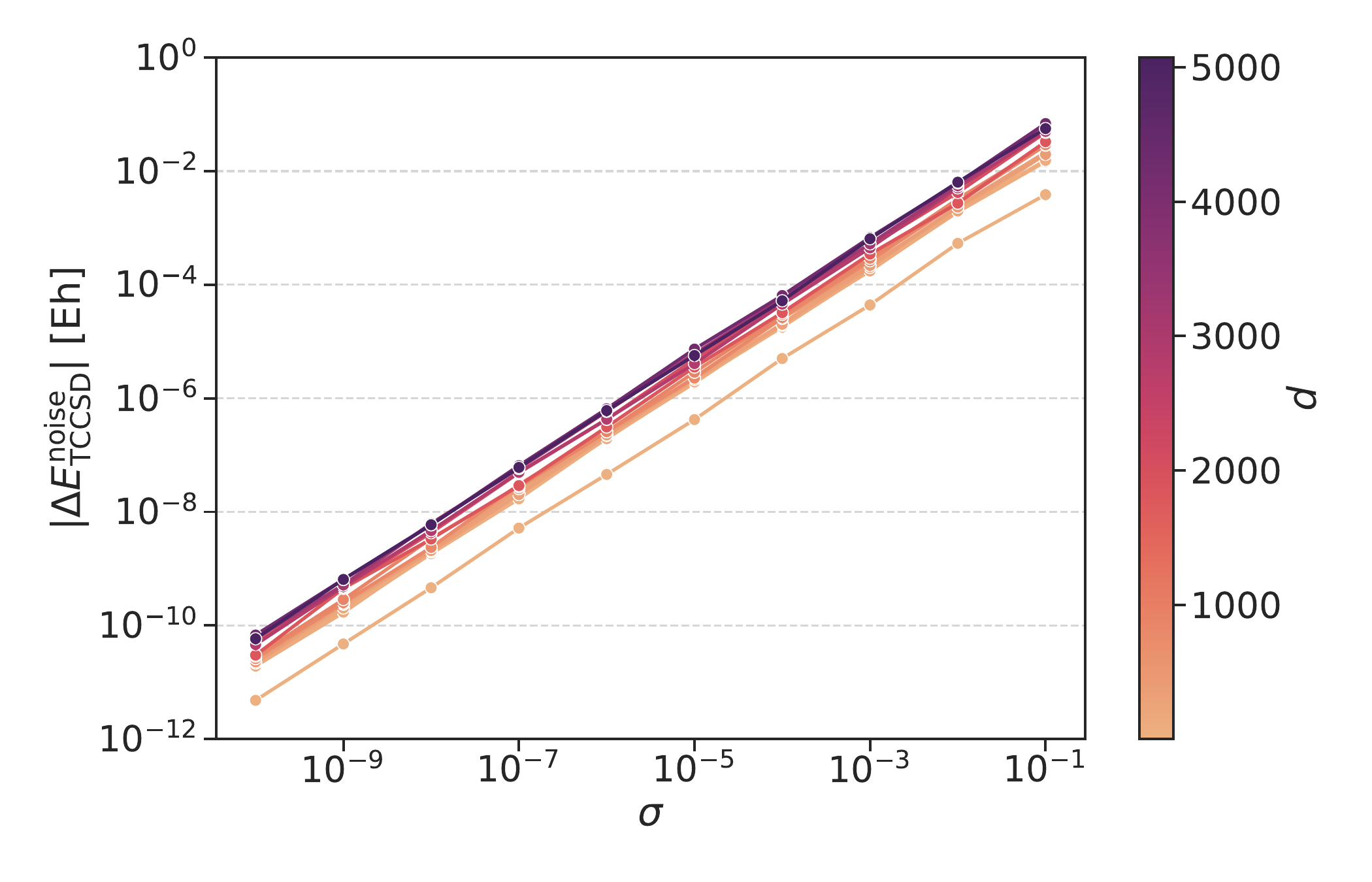}
        \caption{}
    \end{subfigure}
    \hfill
    \begin{subfigure}[b]{0.48\textwidth}
        \centering
        \includegraphics[width=\textwidth]{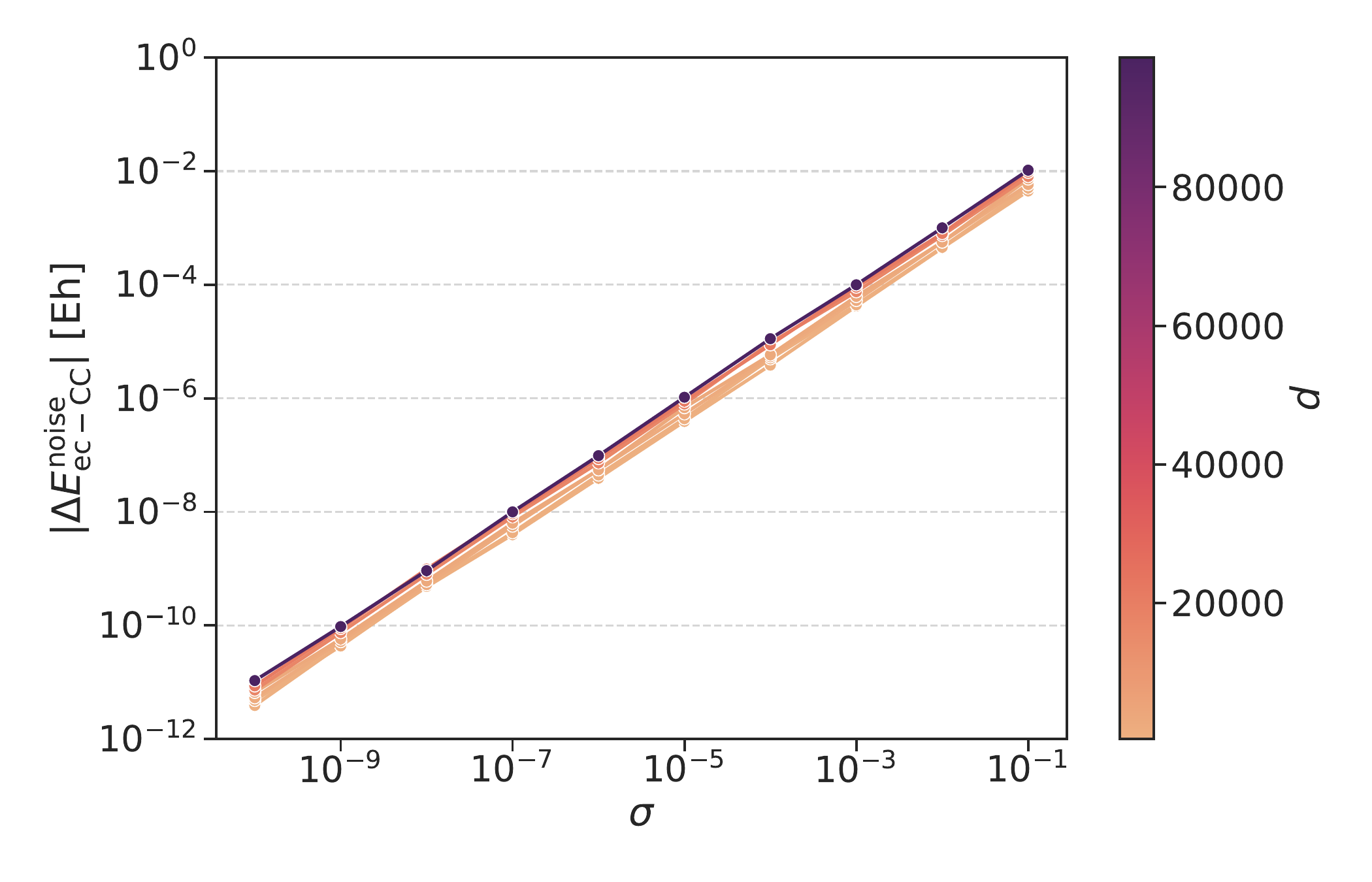}
        \caption{}
    \end{subfigure}
    \caption{Absolute TCCSD energy error (a) and absolute ec-CC energy error (b) as a function of the noise strength $\sigma$ and the number of required overlaps $d$ for each method.
    Note that the vertical axis intercept depends on the molecular system.}
    \label{fig:tcc_error_under_noise}
\end{figure}
Depending on the active space size and hybrid CC method, the required number of overlaps $d$ to run
the final energy calculation varies (see Table~\ref{tab:active_space_sigma_sampling}).
The dependence of the absolute TCCSD and ec-CC energy errors, \detcc~ and \decc, on $\sigma$ and $d$ is shown in Figure~\ref{fig:tcc_error_under_noise}.
A power law fit of the form $|\Delta E| = a \sigma^{\beta_\sigma}$ yielded an exponent
of $\beta_\sigma = 0.9998 \pm 0.0016$ and $\beta_\sigma = 1.001 \pm 0.002$ for TCCSD and ec-CC, respectively. Thus, the absolute energy error of both methods scales approximately linearly
with the noise strength $\sigma$. Furthermore, no convergence issues of the classical CC procedures were observed, even for large noise strengths. Consequently,
TCCSD and ec-CC when run with noisy input amplitudes, as is expected if they are
obtained from actual quantum measurements, are very robust to the strength of the underlying noise
due to the linear dependency on $\sigma$. This is a rather beneficial behavior of the hybrid methods,
given that a large part of the amplitudes are usually rather small and might be even sign-flipped under
strong noise. The exponent $\beta_\sigma$ remains approximately one for other molecular systems as well (data not shown). Note that, however, the prefactor $a$ heavily depends on the system at hand
and is irrelevant at this stage of the analysis.
Another interesting observation from Figure~\ref{fig:tcc_error_under_noise} is that the absolute
errors of both methods are comparable even though ec-CC requires many more overlaps to be evaluated (see Table~\ref{tab:active_space_sigma_sampling}).
This behavior might be biased by the choice of molecular test system, i.e., if the magnitude of the T3 and T4 amplitudes (extracted from CASCI in this case) is small,
the contribution to the projection equations is only mildly influenced by noise. In the TCCSD case, however,
the lower-order frozen T1 and T2 amplitudes are larger in magnitude even for smaller systems, and
these amplitudes directly enter the TCCSD energy expression, contrary to ec-CC.

\subsection{TCCSD Energy Error Extrapolation} \label{apx:tccsd_energy_error_extrapolation}
With the variance bound for overlaps measured through matchgate shadows, it becomes possible to determine the shot budget given an active space size
to obtain the overlaps to a certain accuracy. As we observerd above, the TCCSD energy error under noise depends linearly on the 
standard deviation of the underlying Gaussian distribution. 
This dependence is, however, not sufficient to estimate the required variance for other molecular systems, since
many properties of the system enter the TCCSD calculation and have an influence on the final TCCSD energy error.
Another problem is that some properties determining this error cannot be directly extracted just from the molecular Hamiltonian, i.e., without
solving the underlying electronic structure problem, which becomes computationally prohibitive for large active spaces, for example.
We thus seek a simple extrapolation model for noisy TCCSD energies, taking as input easy to obtain parameters of a molecular system.
Even though we outlined some noise analysis on ec-CC in the section above, we are going to focus exclusively on TCCSD in the following due to
its favorable scaling and small number of overlaps required.
For the error extrapolation model, which will serve as an input for shot budget estimation, we created a small data set containing
13 small molecules for which we first ran CASCI with a) different active space sizes (see Table~\ref{tab:cas_list_dataset}) and b) different basis sets (STO-3G, 6-31G, 6-31G**, cc-pVDZ, aug-cc-pVDZ).\cite{hehre1969a,hehre1970a,ditchfield1971a,francl1982a,gordon1982a,hehre1972a,dunning1989a,kendall1992a,woon1993a}
From the resulting CASCI states, we extracted the required overlaps and added Gaussian noise with $\sigma = 10^{-3}$, which were then put into a TCCSD energy calculation.
\begin{table}[h]
    \centering
    \caption{Active spaces used in the test data set for TCCSD energy error extrapolation, including the resulting number of required overlaps $d$ as input for TCCSD.}
    \begin{tabular}{lllr}
    \toprule
    $n_{\rm sp}$ & $\zeta_\alpha$ & $\zeta_\beta$ &  $d$ \\
    \midrule
    6 & 3 & 3 & 118 \\
    \cline{1-4} \cline{2-4}
    \multirow[t]{3}{*}{8} & 2 & 2 & 199 \\
    \cline{2-4}
     & 3 & 3 & 316 \\
    \cline{2-4}
     & 4 & 4 & 361 \\
    \cline{1-4} \cline{2-4}
    \multirow[t]{2}{*}{10} & 5 & 5 & 876 \\
    \cline{2-4}
     & 6 & 6 & 805 \\
    \cline{1-4} \cline{2-4}
    12 & 6 & 6 & 1819 \\
    \cline{1-4} \cline{2-4}
    14 & 5 & 5 & 2836 \\
    \cline{1-4} \cline{2-4}
    \multirow[t]{5}{*}{16} & 3 & 3 & 2068 \\
    \cline{2-4}
     & 4 & 4 & 3193 \\
    \cline{2-4}
     & 5 & 5 & 4236 \\
    \cline{2-4}
     & 6 & 6 & 5071 \\
    \cline{2-4}
     & 8 & 8 & 5793 \\
    \bottomrule
    \end{tabular}
    \label{tab:cas_list_dataset}
\end{table}
For each molecule/basis set/active space combination, the noise was sampled 30 times, such that an averaged \detcc~ could be obtained for each combination.
This amounts to approximately 20,000 single-point TCCSD calculations, therefore, we only ran the data collection with TCCSD on a small number of molecules.
The molecular geometries and the sampled noisy TCCSD energies, including
code to produce the respective plots can be found on Zenodo.\cite{zenododata}
The settings thus vary a) the required number of overlaps $d$ (eq \eqref{eq:explicit_num_overlaps}) and b) the number of total spin orbitals $N$, which are easy to obtain for any molecule of interest.
For completeness, the number of non-redundant excitation amplitudes for a spin block $\sigma \in \{\alpha, \beta\}$ in $n_\mathrm{sp}$ spatial orbitals and $\zeta_\sigma$ electrons with excitation level $\nu$ is given by
\begin{align}
    d(n_\mathrm{sp}, \zeta_\sigma, \nu) = \frac{1}{(\nu!)^2} \prod_{i = 0}^{\nu - 1} (\zeta_\sigma - i)(n_\mathrm{sp} - \zeta_\sigma - i).
\end{align}
For TCCSD, we need the overlap with $\hfdet$, all non-redundant single excitations, i.e., $\alpha$ and $\beta$ single excitations, for double excitations, we have $\alpha\alpha$, $\beta\beta$, and $\alpha\beta$ excitations, where the latter one corresponds to a single $\alpha$ excitation coupled with a single $\beta$ excitation.
The total number of overlaps for TCCSD thus amounts to
\begin{align}
d \equiv 1 +
d(n_\mathrm{sp}, \zeta_\alpha, 1) + d(n_\mathrm{sp}, \zeta_\beta, 1)
+ d(n_\mathrm{sp}, \zeta_\alpha, 2) + d(n_\mathrm{sp}, \zeta_\beta, 2)
+ d(n_\mathrm{sp}, \zeta_\alpha, 1)d(n_\mathrm{sp}, \zeta_\beta, 1). \label{eq:explicit_num_overlaps}
\end{align}

One ``variable'' contributing to \detcc, which is highly system-specific, is the ``amount'' of static or dynamic correlation in the system. For example, a weakly
correlated molecule, where the dynamic correlation energy is dominant, is only mildly affected by noisy input amplitudes for TCCSD, whereas the opposite might be the case
for strongly correlated systems. There is, however, no single quantity that could serve to model this properly.
With our choice of molecules, we try to cover several different correlation scenarios, e.g., diatomic molecules in equilibrium, stretched diatomics, and some small
organic molecules.
All things considered, after several attempts to build a good error model, we came up with the following power law,
\begin{align}
    \detcc = a d^\beta N^\gamma \sigma. \label{eq:tcc_error_powerlaw_global}
\end{align}
In the randomized noise simulations above, we fixed the noise strength, yielding the unknown parameters on the right-hand side as
\begin{align}
    \frac{\detcc}{\sigma}= a d^\beta N^\gamma.
\end{align}
We fitted the above power law to the recorded data over all molecules, basis sets, active space sizes,
and noise samples, yielding the exponent parameters shown in Table~\ref{tab:extrapolation_results}
together with standard errors determined through bootstrapping (50,000 bootstrap samples).
\begin{table}[ht]
    \centering
    \caption{Parameters and errors obtained from power law fit over the entire TCCSD noise error data set.}
    \label{tab:extrapolation_results}
    \begin{tabular}{cSS}
        \toprule
       Parameter & {Value} & {Bootstrap Std. Error} \\
       \midrule
        $\beta$ & 0.277 & 0.054\\
        $\gamma$ & -1.074 & 0.116 \\
        \bottomrule
    \end{tabular}
\end{table}
The prefactor $a$ carries the dependencies on omitted variables (e.g., the correlation strength, etc.),
thus, it is not advisable to fit this parameter globally, i.e., for all molecules at once.
Through running the power law fit on subsets of the sample (leave-one-out, etc.), we could confirm that
the exponents seem to be largely independent of the molecular system (data not shown), indicating
that the fit parameters describe the trend of the TCCSD error through noise quite well in the simple
power law. In the subset power law fits, we found that the prefactor varies largely depending on the
molecular systems. For this reason, with the globally fitted exponents $\beta$ and $\gamma$ at hand,
we re-ran the fit for $a$ for each molecule, keeping the exponents fixed at the fitted values
from the entire data set. This yielded a prefactor $a$ for each molecule, shown in Table~\ref{tab:prefactor_per_molecule}
and plotted in Figure~\ref{fig:prefactors_molecules}.
\begin{table}[ht]
    \centering
    \caption{
    Prefactor $a$ obtained in a per-molecule fit of the TCCSD noise error data with exponents fixed at values from Table~\ref{tab:extrapolation_results},
    together with the $\mathcal{T}_1$ and $\mathcal{D}_1$ values obtained from plain CCSD averaged over basis sets.
    }
    \label{tab:prefactor_per_molecule}
    \begin{tabular}{cS[round-mode=places,round-precision=2]S[round-mode=places,round-precision=3]S[round-mode=places,round-precision=3]}
    \toprule
    Molecule & {$a$} & {$\mathcal{T}_1$} & {$\mathcal{D}_1$} \\
    \midrule
    N$_2$ (stretched) & 18.669051 & 0.049073 & 0.106929 \\
    F$_2$ (stretched) & 6.476190 & 0.022695 & 0.088622 \\
    p-Benzyne & 5.650327 & 0.014870 & 0.061429 \\
    Cl$_2$ (stretched) & 4.832977 & 0.014567 & 0.081032 \\
    F$_2$ & 4.270402 & 0.009680 & 0.027547 \\
    N$_2$ & 3.203512 & 0.010011 & 0.023822 \\
    H$_2$O & 3.013966 & 0.007588 & 0.017423 \\
    Me-NCO & 3.001943 & 0.014220 & 0.052977 \\
    Benzene & 2.979277 & 0.008315 & 0.027102 \\
    Formaldehyde & 2.597973 & 0.013824 & 0.044913 \\
    Acetaldehyde & 2.218479 & 0.012532 & 0.048172 \\
    Furan & 2.557929 & 0.011529 & 0.045743 \\
    Cl$_2$ & 2.407704 & 0.004583 & 0.017974 \\
    \bottomrule
    \end{tabular}
\end{table}
The raw data for \detcc~ including the power law fits is plotted in Figure~\ref{fig:powerlaw_molecules_grid}.
One can clearly see in this plot that the globally fitted exponents represent the slope and trend of basis set dependence well,
hinting that we can indeed use the error extrapolation model for active spaces and systems beyond the data set.
As expected, the TCCSD error magnitude grows with the number of required overlaps, but decreases with the size of the total MO space (negative exponent).
The latter can be rationalized as follows: when keeping the active space size (and hence the number of overlaps $d$) fixed while increasing the overall MO space,
the contribution of the AS in relation to the external space decreases, and the error made by noisy amplitude measurements becomes more and more negligible.
The remaining hurdle is, given a molecule not present in the data set, to select the prefactor $a$ for that system.
This needs to be done by comparing the ``correlation strength'' and finding the best fit in our tabulated data. Since this
is a bit cumbersome, we sought for a more tangible metric that is easy to obtain and came up with the typical CCSD diagonstic values,
which are commonly used to assess whether the CCSD wavefunction provides a reliable, true SR result or whether MR methods are needed.
We computed the $\mathcal{T}_1$ and $\mathcal{D}_2$ diagnostics\cite{lee1989diagnostic,janssen1998new} for each molecule and basis set.
As can be seen in Figure~\ref{fig:prefactors_molecules}, the trend of the molecule-specific prefactor $a$ and their corresponding diagnostic values
are in good agreement with each other. Thus, for the determination of the prefactor of a molecule outside our data set, we suggest to first compute the $\mathcal{T}_1$ diagnostic
value, because it has the strongest correlation with $a$, and to select the prefactor accordingly (see Table~\ref{tab:prefactor_per_molecule}).
The linear model to convert $\mathcal{T}_1$ to $a$ is given by
\begin{align}
a \approx 375.9 \times \mathcal{T}_1 -0.84. \label{eq:t1diag_to_a}
\end{align}

Clearly, the stretched \ce{N2} molecule is kind of an outlier because the prefactor $a$ and its diagnostic values are far off the median/mean of the data set. This, however,
shows that the scenario for \ce{N2} can serve as a ``worst case'' scenario, which will be useful in the following to estimate upper bounds for shot budgets.
In summary, our error extrapolation protocol relies on the following steps:
\begin{enumerate}
    \item Select a molecule and basis set (determining $N$)
    \item Compute CCSD in that setting and obtain the $\mathcal{T}_1$ diagnostic value
    \item Select an active space size and compute the number of required overlaps for TCCSD $d$
    \item Map the diagnostic value to a prefactor $a$ using the conversion formula in eq \eqref{eq:t1diag_to_a} (compare with Table~\ref{tab:prefactor_per_molecule})
    \item Insert all the quantities ($d$, $N$, $a$, and the exponents in Table~\ref{tab:extrapolation_results}) into eq \eqref{eq:tcc_error_powerlaw_global}
    \item Together with a noise strength $\sigma$, one can estimate the TCCSD energy error due to noise in that precise setting
\end{enumerate}
We will use this procotol in the following to estimate shot budgets for quantum TCCSD.

\begin{figure}[ht]
\centering
\includegraphics[width=1\textwidth]{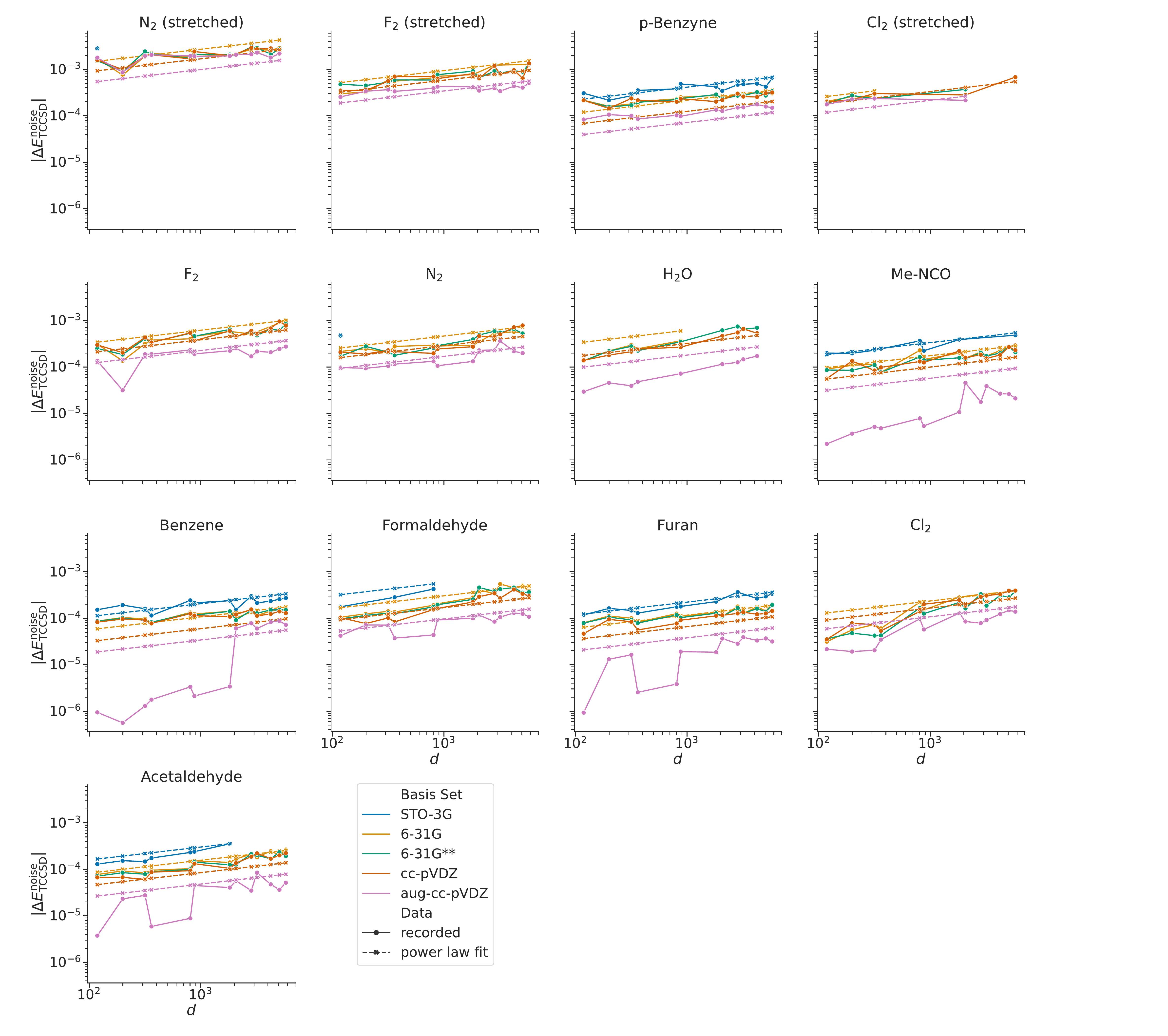}
\caption{%
Absolute TCCSD energy error \detcc~ as a function of $d$ for different molecules and basis sets with $\sigma = 10^{-3}$.
The dashed lines employ a power law fit using the globally fitted exponents of $d$ and $N$ (Table~\ref{tab:extrapolation_results}) together with the molecule-specific prefactor $a$ (Table~\ref{tab:prefactor_per_molecule}).
} \label{fig:powerlaw_molecules_grid}
\end{figure}

\subsection{Shot Budget Estimation}
Combining eq \eqref{eq:shadow_overlap_error} with eq \eqref{eq:tcc_error_powerlaw_global} we can obtain a bound on the number of shots needed for guarantee that the shot noise does not change the TCCSD energy by more than a given magnitude.
Finally, one can estimate the shot budget for the half-filling case as
\begin{align}
s \lessapprox \frac{a^2}{|\Delta E|^2} d^{2\beta} N^{2\gamma} \sqrt{2n}
\end{align}
for a given target TCCSD energy error $|\Delta E|$ and the system-dependent variables, obtained through the protocol in the previous section.
\begin{figure}[ht]
\centering
\includegraphics[width=0.6\textwidth]{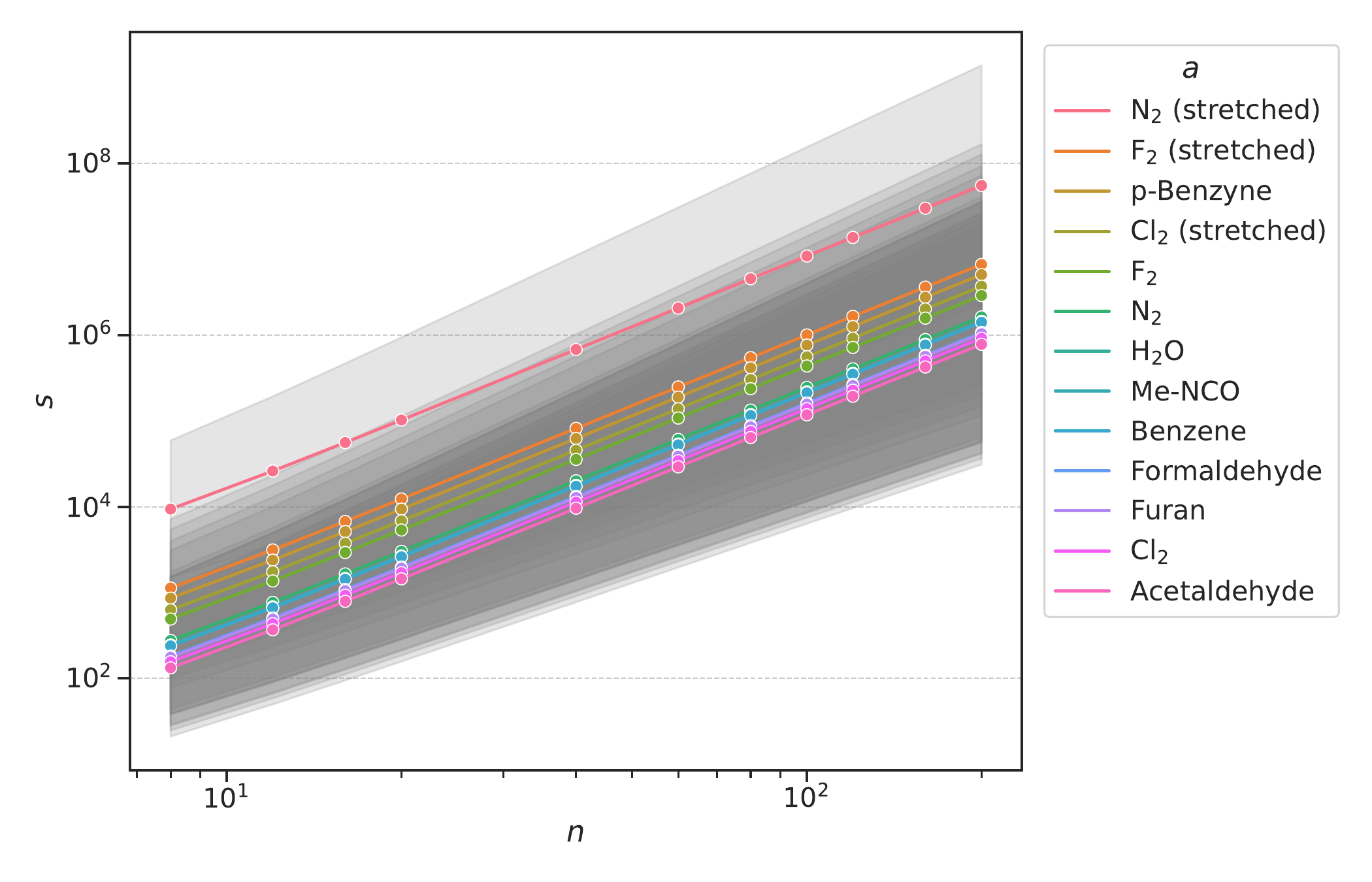}
\caption{Shot budget estimation for ``scenarios'' based on molecular systems (prefactor $a$), with globally fitted exponents ($\pm$ standard deviation
of the exponents from global fit shown as gray shaded areas), target energy error $|\Delta E| = 10^{-3}~\mathrm{Eh}$, $N = 600$, $\zeta = n/2$.
Note that the molecules here are just labels for different prefactors $a$. This is done to cover different possible scenarios of the
prefactor, i.e., figure out best/worst case scenarios for shot budgets based on TCCSD energy error extrapolation.
} \label{fig:shotbudget_all}
\end{figure}
Using the above equation, we estimated the shot budgets for the scenarios dictated by the molecule prefactors aiming for a target accuracy $|\Delta E| = 10^{-3}$ Eh with an arbitrarily chosen
MO space of $N=600$ and half-filling active spaces from $n=4$ to $n=200$ qubits. Uncertainties in the fitted exponents were included by estimating a lower bound (exponent minus bootstrap standard deviation)
and an upper bound (exponent plus bootstrap standard deviation) of the shot budget. The resulting shot budget estimates $s$ are plotted in Figure~\ref{fig:shotbudget_all}.
Surprisingly, even in the worst case scenario (strong correlation as in a stretched \ce{N2} molecule), our estimated upper bound for the shot budget to obtain the (noisy) TCCSD energy to chemical
precision barely exceeds $10^8$ shots in a 200 qubit setting.
We want to stress again that the shots budgets do not include measurements for state preparation/optimization on the quantum device, just the overlap measurements through matchgate shadows.
If there were imperfections in the trial state, i.e., the shot noise is not applied to the \emph{exact} overlaps, the error of the quantum TCCSD with respect to its (exact) classical counter part
would be much larger than what our error model predicts. Nonetheless, the broad window of different correlation scenarios covered in here should be helpful enough to get a decent understanding
on the order of magnitude of the required shot budget.
\begin{table}[h]
    \centering
    \caption{Shot budgets $s$ estimated using the error extrapolation model for the \ce{N2}/cc-pVDZ dissociation curve
    needed for obtaining a TCCSD energy in a (6, 6) active space to chemical precision, together with the $\mathcal{T}_{1}$ diagnostic
    extracted from CCSD.
    }
    \begin{tabular}{S[round-mode=places,round-precision=1]S[round-mode=places,round-precision=3]S[round-mode=places,round-precision=1]}
    \toprule
    $R$~[\AA] & $\mathcal{T}_{1}$ & $s$ \\
    \midrule
    0.8                      & 0.003971 & 5.231000000000000000e3 \\
    0.9                      & 0.005565 & 1.911300000000000000e+04 \\
    1.000000000000000000e+00 & 0.007615 & 4.974400000000000000e+04 \\
    1.100000000000000089e+00 & 0.009951 & 1.021520000000000000e+05 \\
    1.199999999999999956e+00 & 0.012471 & 1.796600000000000000e+05 \\
    1.299999999999999822e+00 & 0.015096 & 2.834590000000000000e+05 \\
    1.399999999999999911e+00 & 0.017743 & 4.119690000000000000e+05 \\
    1.500000000000000000e+00 & 0.020334 & 5.610530000000000000e+05 \\
    1.599999999999999867e+00 & 0.022822 & 7.257700000000000000e+05 \\
    1.699999999999999734e+00 & 0.025201 & 9.030200000000000000e+05 \\
    1.799999999999999822e+00 & 0.027514 & 1.093981000000000000e+06 \\
    1.899999999999999689e+00 & 0.029856 & 1.305941000000000000e+06 \\
    1.999999999999999778e+00 & 0.032356 & 1.552971000000000000e+06 \\
    2.099999999999999645e+00 & 0.035100 & 1.848706000000000000e+06 \\
    2.199999999999999734e+00 & 0.037925 & 2.179969000000000000e+06 \\
    2.299999999999999822e+00 & 0.040453 & 2.499711000000000000e+06 \\
    2.399999999999999467e+00 & 0.042473 & 2.770777000000000000e+06 \\
    2.500000000000000000e+00 & 0.043984 & 2.982769000000000000e+06 \\
    2.599999999999999645e+00 & 0.045062 & 3.138644000000000000e+06 \\
    2.699999999999999289e+00 & 0.045782 & 3.245054000000000000e+06 \\
    2.799999999999999822e+00 & 0.046208 & 3.308901000000000000e+06 \\
    \midrule
    & $\sum s$ &  2.9e7 \\
    \bottomrule
    \end{tabular}
    \label{tab:shot_budgets_n2}
\end{table}

\clearpage
\section{ec-CC on Hardware Data From Google's Sycamore Quantum Processor} \label{apx:eccc_experiment}
\begin{figure}[h]
\includegraphics[width=0.6\textwidth]{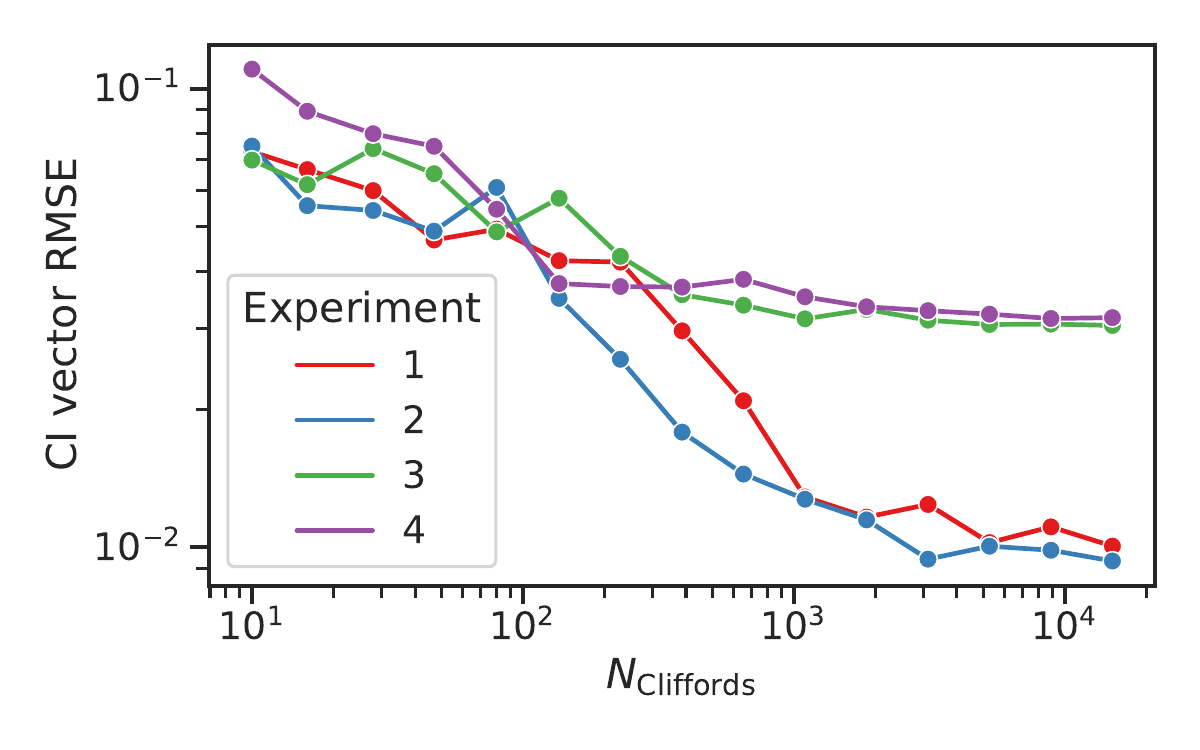}
\caption{
Root mean squared error (RMSE) of the absolute values of the wavefunction coefficients (CI vector) extracted from the trial state for four different repetitions.
The device noise in the first two repeats is far smaller than in the last two repeats.
} \label{fig:h4_sto3g_rmse}
\end{figure}

To investigate the performance of ec-CC on actual noisy data, we
used the wavefunctions extracted from Clifford shadows published in Ref.~\citenum{huggins_unbiasing_2022}.
The raw data can be obtained from Zenodo.\cite{malone_qcqmc_data}
The trial state wavefunctions (referred to as ``Q trial'' in the original paper) were first ``converted'' to real-valued wavefunctions through rotation
by the phase of the largest absolute value wavefunction coefficient (i.e., making the wavefunction as ``real'' as possible), then retaining only the real part of the wavefunction and renormalizing it. This state is referred to as ``Q trial real''.
Since the quality of the trial wavefunctions extracted from the device in the first two repetitions is
superior to the other two runs (see Figure~\ref{fig:h4_sto3g_rmse}), we only use the data from these runs.

\begin{table}[h]
\caption{H4/STO-3G energies from Q trial state converted to real-valued wavefunction.
Absolute errors in mEh with respect to FCI are shown in parentheses.} \label{tab:h4_sto3g_fqe_real}
\begin{tabular}{cll}
\toprule
$N_\mathrm{Cliffords}$ & Repeat 1 & Repeat 2\\
\midrule
10 & $-1.82121~(148.30)$ & $-1.80389~(165.63)$ \\
16 & $-1.85039~(119.13)$ & $-1.85554~(113.97)$ \\
28 & $-1.92465~(44.86)$ & $-1.88711~(82.40)$ \\
47 & $-1.94073~(28.78)$ & $-1.91560~(53.91)$ \\
80 & $-1.92072~(48.79)$ & $-1.89485~(74.66)$ \\
136 & $-1.94335~(26.16)$ & $-1.93524~(34.27)$ \\
229 & $-1.95366~(15.85)$ & $-1.95132~(18.19)$ \\
387 & $-1.95760~(11.91)$ & $-1.95954~(9.97)$ \\
652 & $-1.96250~(7.01)$ & $-1.96332~(6.20)$ \\
1100 & $-1.96616~(3.35)$ & $-1.96553~(3.98)$ \\
1856 & $-1.96634~(3.17)$ & $-1.96650~(3.02)$ \\
3129 & $-1.96549~(4.02)$ & $-1.96683~(2.69)$ \\
5276 & $-1.96604~(3.47)$ & $-1.96634~(3.18)$ \\
8896 & $-1.96540~(4.12)$ & $-1.96651~(3.00)$ \\
15000 & $-1.96562~(3.90)$ & $-1.96666~(2.85)$ \\
\bottomrule
\end{tabular}

\end{table}

\begin{table}[h]
\caption{H4/STO-3G ec-CC energies based on Q trial state converted to real-valued wavefunction.
Absolute errors in mEh with respect to FCI are shown in parentheses.
} \label{tab:h4_sto3g_ec}
\begin{tabular}{cll}
\toprule
$N_\mathrm{Cliffords}$ & Repeat 1 & Repeat 2\\
\midrule
10 & $-1.97656~(7.05)$ & $-1.97656~(7.05)$ \\
16 & $-1.97656~(7.05)$ & $-1.97656~(7.05)$ \\
28 & $-1.97637~(6.86)$ & $-1.97656~(7.05)$ \\
47 & $-1.97637~(6.86)$ & $-1.97656~(7.05)$ \\
80 & $-1.97656~(7.05)$ & $-1.97656~(7.05)$ \\
136 & $-1.97637~(6.86)$ & $-1.97656~(7.05)$ \\
229 & $-1.97637~(6.86)$ & $-1.96732~(2.19)$ \\
387 & $-1.96684~(2.67)$ & $-1.96846~(1.05)$ \\
652 & $-1.96779~(1.73)$ & $-1.96858~(0.93)$ \\
1100 & $-1.96928~(0.23)$ & $-1.96881~(0.70)$ \\
1856 & $-1.96973~(0.22)$ & $-1.96989~(0.38)$ \\
3129 & $-1.96929~(0.22)$ & $-1.97074~(1.23)$ \\
5276 & $-1.96961~(0.10)$ & $-1.97026~(0.75)$ \\
8896 & $-1.96893~(0.58)$ & $-1.97023~(0.72)$ \\
15000 & $-1.96891~(0.60)$ & $-1.97028~(0.76)$ \\
\bottomrule
\end{tabular}

\end{table}

\begin{table}[h]
\caption{H4/STO-3G ec-CC energies based on the exact FCI state with overlaps measured through
matchgate shadows. The results are averaged over 30 samples of the matchgate shadow.
Absolute errors in mEh with respect to FCI are shown in parentheses.
} \label{tab:h4_sto3g_ec_matchgate}
\begin{tabular}{cll}
\toprule
Shots & multi Matchgate & single Matchgate \\
\midrule
10000 & $-1.98029~(10.89)$ & $-1.96906~(3.32)$ \\
20000 & $-1.97662~(7.62)$ & $-1.96990~(2.33)$ \\
50000 & $-1.97160~(4.84)$ & $-1.96927~(1.35)$ \\
100000 & $-1.96922~(2.64)$ & $-1.96933~(0.88)$ \\
200000 & $-1.96965~(1.88)$ & $-1.96945~(0.55)$ \\
500000 & $-1.96937~(1.07)$ & $-1.96945~(0.39)$ \\
1000000 & $-1.96963~(0.75)$ & $-1.96949~(0.23)$ \\
\bottomrule
\end{tabular}
\end{table}

\newcolumntype{Q}{S[round-mode=places,round-precision=2]}
\begin{table}[h]
\centering
\caption{Error in the total energy (in kcal/mol) of the diamond minimal unit cell in a double-zeta basis as a function of the
lattice constant $R$.} \label{tab:diamond}
\resizebox{\textwidth}{!}{%
\begin{tabular}{lQQQQQQQQQQQ}
\toprule
{$R$}~[\AA] & {Q trial} & {Q trial real} & {ec-CC/Q trial} & {TCCSD/Q trial} & {QC-QMC} & {CAS(8,8)} & {ec-CC/CAS} & {NEVPT2} & {TCCSD/CAS} & {CCSD(T)} & {AFQMC} \\
\midrule
2.88 & 266.584750 & 174.889465 & 4.307927 & 19.339516 & 0.053966 & 51.151317 & 3.173382 & 0.679174 & 1.113588 & -0.347022 & 2.761667 \\
3.24 & 1006.438535 & 693.046334 & 6.601222 & 20.738693 & -0.854667 & 51.960262 & 4.106293 & 0.706335 & 1.040760 & -0.593070 & 3.153233 \\
3.60 & 177.295506 & 137.780469 & 6.406229 & 36.968455 & 1.166539 & 52.448060 & 5.489898 & 0.667841 & 0.716073 & -1.099351 & 4.993089 \\
3.96 & 208.043442 & 159.251081 & 10.099162 & 37.948236 & 3.458202 & 52.842058 & 7.702869 & 0.543519 & 0.173579 & -2.172873 & 8.471999 \\
4.32 & 327.486101 & 233.189706 & 15.009005 & 27.470500 & 4.436489 & 53.140043 & 11.655740 & 0.659281 & -0.180135 & -4.401393 & 16.540510 \\
\bottomrule
\end{tabular}
}
\end{table}

\clearpage

\end{document}